\documentclass{lmcs}

\usepackage{hyperref}

\usepackage[utf8x]{inputenc}
\usepackage{amsmath,latexsym,amssymb}
\usepackage{amscd}
\usepackage[all,cmtip]{xy}
\usepackage{diagxy}

\newcommand{\reals}{\mathbb{R}}

\newcommand{\integers}{\mathbb{Z}}
\newcommand{\naturals}{\mathbb{N}}
\newcommand{\rationals}{\mathbb{Q}}

\newcommand{\myemptyset}{\varnothing}
\newcommand{\defeq}{\stackrel{\text{def}}{=}}
\newcommand{\domp}[1]{\dom (#1)}
\newcommand{\rangep}[1]{\range \bigl(#1 \bigr)}
\newcommand{\directimage}[2]{#2 \bigl[ #1 \bigr]}
\newcommand{\inverseimage}[2]{#2^{-1} \bigl[ #1 \bigr]}
\newcommand{\powersetof}[1]{\mathcal{P}(#1)}

\DeclareMathOperator{\range}{range}

\DeclareMathOperator{\dom}{dom}

\newcommand{\rsphere}[2]{#2\mathbb{S}^{#1}}
\newcommand{\puncturedsphere}[1]{\widehat{\mathbb{S}}^{#1}}
\newcommand{\sphere}[1]{\rsphere{#1}{}}

\newcommand{\euclideanspace}[2]{\mathbb{#1}^{#2}}

\newcommand{\rationalsspace}[1]{\euclideanspace{Q}{#1}}

\newcommand{\euclidean}[1]{\euclideanspace{R}{#1}}

\newcommand{\realprojectivespace}[1]{\reals \text{P}^{#1}}

\newcommand{\generalstringset}[2]{#1^{#2}}
\newcommand{\stringset}[1]{\generalstringset{#1}{*}}
\newcommand{\sigmastringset}{\stringset{\Sigma}}

\newcommand{\chomeom}{\cong_{ct}}

\newcounter{linecounter}

\newcommand{\resetline}{\setcounter{linecounter}{0}}

\newcommand{\omegastringset}[1]{\generalstringset{#1}{\omega}}
\newcommand{\sigmaomegastringset}{\omegastringset{\Sigma}}

\newcommand{\fsunionbasenotation}[1]{\bigcup\nolimits_{\text{fs}} #1}
\newcommand{\fsintersectionbasenotation}[1]{\bigcap\nolimits_{\text{fs}} #1}

\newcommand{\fsunionbasenu}{\fsunionbasenotation{\nu}}
\newcommand{\fsintersectionbasenu}{\fsintersectionbasenotation{\nu}}

\newcommand{\ctopcompsubspace}[2]{#2_{#1}}

\theoremstyle{plain}
\newtheorem{theorem}{Theorem}[section]
\newtheorem{proposition}[theorem]{Proposition}
\newtheorem{lemma}[theorem]{Lemma}
\newtheorem{corollary}[theorem]{Corollary}

\theoremstyle{definition}
\newtheorem{definition}[theorem]{Definition}
\newtheorem*{definition*}{Definition}
\newtheorem*{definitions*}{Definitions}
\newtheorem{example}[theorem]{Example}
\newtheorem*{example*}{Example}

\theoremstyle{remark}

\newtheorem*{remark*}{Remark}

\begin{document}

%

\title[Computable structures on manifolds]{\bf Computable structures on topological manifolds}


\author[M.~A.~Aguilar]{Marcelo A. Aguilar
}
\address{Instituto de Matemáticas, 
Universidad Nacional Autónoma de México\\
Ciudad Universitaria, M\'exico D.F. 04510, M\'exico}
\email{marcelo@matem.unam.mx}  

\author[R.~Conde]{Rodolfo Conde
}
\address{Instituto de Matemáticas, 
Universidad Nacional Autónoma de México\\
Ciudad Universitaria, M\'exico D.F. 04510, M\'exico}
\email{rodolfocondemx@gmail.com}  


\begin{abstract}
We propose a definition of computable manifold by introducing computability
as a structure that we impose to a given topological manifold, just in the same way as
differentiability or piecewise linearity are defined for smooth and PL manifolds respectively. Using
the framework of computable topology and Type-2 theory of effectivity, we develop computable
versions of all the basic concepts needed to define manifolds, like \emph{computable atlases} and
\emph{(computably) compatible} computable atlases. We prove that given a computable atlas $\Phi$
defined on a set $M$, we can construct a computable topological space $(M, \tau_\Phi, \beta_\Phi,
\nu_\Phi)$, where $\tau_\Phi$ is the topology on $M$ induced by $\Phi$ and that the
equivalence class of this computable space characterizes the \emph{computable structure} determined
by $\Phi$. The concept of \emph{computable sub\-ma\-ni\-fold} is also investigated. We show that any
compact computable manifold which satisfies a computable version of the $T_2$-separation axiom, can
be embedded as a computable submanifold of some euclidean space $\euclidean{q}$, with a computable
embedding, where $\euclidean{q}$ is equipped with its usual topology and some canonical computable
encoding of all open rational balls.%
\end{abstract}


\maketitle

\section{Introduction}

Computability theory over continuous structures began formally in 1936 with the landmark paper of 
Alan Turing \cite{Turing1936} where he defined the notion of a single computable real number: $x\in
\reals$ is computable if its decimal expansion can be calculated in the discrete sense, that is,
output by a Turing machine. Since then, other authors have developed definitions and results to try
to build a reasonable theory of computability in the continuous setting. There are two main
approaches to modeling computations with real number inputs. The first approach is given by the
framework of \emph{Computable Analysis} studied in many papers and some books, e.g., 
\cite{Lacombe,Grzegorczyk,pourElRichards,kerikocomplexityrealfunctions,weihrauch2000computable}. 
The second approach is the algebraic one, 
its development goes back to the 1950s and focuses on the algebraic operations needed to perform 
tasks \cite{borodin1975computational,burgisser1997algebraic}. The most influential model is the so 
called \emph{BSS model} developed by Blum, Shub and Smale \cite{BSS89,BCSS98}.

Computable Analysis reflects the fact that computers can
only store finite amounts of information. Since real numbers and other objects in analysis are
``infinite'' in nature, a Turing machine can only use finite objects to approximate them and to perform
the actual computations on these finite pieces of information, thus we have that topology plays an
important role in computable analysis \cite{weihrauch2000computable}. The representation approach
and the framework of \emph{Type-2 Theory of Effectivity} (TTE) \cite{theoryrepresentations}, a
generalization of ordinary computability theory, has provided a solid background to formalize the
theory of computable analysis \cite{weihrauch2000computable}. TTE has been extensively studied and 
developed as a standalone topic in computability theory 
\cite{Weihrauch:jucs_14_6:the_computable_multi_functions}. Also, it has been generalized to a more 
general model of computability for analysis \cite{GTMs}. 
\emph{Computable metric spaces} (also known as \emph{recursive metric spaces})
\cite{KlausCompMetricStructs,Bra03ComputabilityModels} have been defined in computable
analysis and play a very important role in the subject, as many results of computability over
euclidean spaces can be generalized to the broader world of computable metric spaces. For an 
overview of basic computable analysis, see the tutorial given in \cite{CCAtutorial}.

In recent years, as a product of various publications
\cite{weihrauch2000computable,schroder2003admissible,GW05ComputableDini,GSW07Computablemetrization} trying to
consider computable topology as a foundation of computable analysis, Weihrauch and Grubba
\cite{KWECTOP} developed a solid foundation for computability over more general spaces, where the
main objects of study are called \emph{computable topological spaces}. Roughly
speaking, a computable topological space is a $T_0$-space $(X,\tau)$ in which a base $\beta\subseteq
\tau$ is provided with an encoding with strings from $\sigmastringset$, such that the set of all
valid strings that encode elements of $\beta$ is computable and under this encoding, intersection of
base elements is computable in a formal sense (see Definition \ref{defCompTOPSPACE}). This 
framework has been used to prove many important results in computable topology, like the following:  
A computable version of Dini's 
Theorem is proved in \cite{GW05ComputableDini}; A celebrated result in general topology and the theory of metric spaces 
states that every second-countable regular topological space $(X, \tau)$ is metrizable, that is, 
$X$ is homeomorphic to a metric space $(M,d)$. In \cite{GSW07Computablemetrization}, it is proved 
that every computable topological space satisfying a ``computably regular'' 
condition, has a computable embedding in a computable metric space (which topologically is its 
completion). Also, a computable Urysohn Lemma is established. 
 There is also a computable version of the Stone-Weierstrass Approximation Theorem  
\cite{GrubbaEffectiveStoneWeierstrass} for computable topological spaces which satisfy computable 
versions of the Hausdorff and locally compactness properties.

 The paper \cite{KWECTOP} is an effort to put in one place the most general and basic facts about 
computable topology, because these facts were scattered throughout many papers 
(e.g., \cite{weihrauch2000computable,GW05ComputableDini,GSW07Computablemetrization,GrubbaEffectiveStoneWeierstrass} among others). Computable versions of topological properties and 
separation axioms are an important tool to obtain computable versions of important topological 
theorems and the paper \cite{Weihrauch:jucs_16_18:computable_separation_in_topology} gives 
definitions of many computable versions of each separation axiom and proves the relations between 
these computable separation axioms. More recent results of computable topology include 
\cite{DBLP:journals/corr/RettingerW13}, where Rettinger and Weihrauch study computability aspects 
of finite and infinite products of computable topological spaces and they prove computable versions 
of Tychonoff's Theorem.

With the advent of computable topological spaces and the solid foundation of computability in
euclidean spaces provided by computable analysis, the stage is set to introduce computability
properties in a very important class of topological spaces: \emph{Topological manifolds}. A space $M$
is a topological manifold if and only if each point $x\in M$ has a neighborhood
homeomorphic to an open set of euclidean space $\euclidean{n}$ for fixed $n$ (this integer
is called the \emph{dimension} of $M$). Manifolds are one of the most important types of topological
spaces, many problems related to manifolds have been the inspiration for some of the most beautiful
mathematical constructions, using very sophisticated and advanced techniques, as it can be seen in 
the foundational work of Kirby and Siebenmann \cite{kirby1977foundational} and more recently, in 
the paper \cite{MR3402697}, where Manolescu shows that there exist high-dimensional 
topological manifolds which cannot be triangulated as simplicial complexes, thus refuting the 
\emph{Triangulation conjecture} \cite{simplicialtriangulationstop}. But despite the fact that there 
is a lot of knowledge about to\-po\-lo\-gi\-cal manifolds, there are still unanswered questions and 
hard open problems related to them. As a result of this, manifold theory is always intensively
studied by the mathematicians doing research in topology.

Additional structure can be imposed to a topological manifold, giving rise to special
classes of manifolds, like \emph{smooth manifolds} \cite{WhitneyDIFF36,brickellclarkDIFF},
\emph{analytic manifolds} \cite{hormander1990introduction} and \emph{piecewise linear
(PL) manifolds} \cite{hudson1969piecewise,rourke1982introduction}. The relationships between
manifolds equipped with any or all of these structures and standard topological manifolds are
studied in many papers \cite{kirby1977foundational,WhitneyDIFF36,KervaireManifoldNODIFF,poincareconjecturegt4,Freedman4manifolds}.

In this paper, computability enters the world of topological manifolds. We propose a definition of
 computable manifold by introducing computability as a structure that we
impose to a given topological manifold, just in the same way as differentiability or
piecewise linearity are defined for smooth and PL manifolds respectively. Using the framework of
computable topology and TTE, we give effective versions of the concepts needed to define
manifolds e.g., \emph{charts} and \emph{atlases}. Namely, the following concepts are developed:
 a) The definition of \emph{computable atlas} $\Phi$ on a set $M$;
  b) the computable topological space induced by $\Phi$ on $M$;
  c) the definition of \emph{computably compatible computable atlases};
 d) the definition of \emph{computable structure}, which is an equivalence class of
computably compatible computable atlases (characterized by the equivalence class of a computable
topological space, see Definition \ref{defEquivCompTOPSPACES}).

A \emph{computable manifold} is a set $M$ endowed with a computable structure. We present many
examples of computable manifolds and their respective computable topological spaces. We study the
relationships between computable manifolds and computable functions (with respect to the representations
induced by computable atlases) and prove that computable homeomorphisms and computable structures
behave nicely when working together. We also prove some basic properties of
computable manifolds.

\emph{Submanifolds}, that is, manifolds which are inside other manifolds are an important tool
to investigate manifolds. We define \emph{computable submanifolds} and present some properties about
them. 
Finally, we use all the previous definitions and results about 
computable manifolds to give an effective version of the following well known result concerning 
topological manifolds \cite{1974TOPPROPERTIES,munkres2000topology}: \emph{Every compact Hausdorff 
topological manifold embeds in some high dimensional euclidean space}. We show that any compact 
computable manifold $M$ which is also computably Hausdorff 
(see Definition \ref{defSCT2}) can 
be 
embedded in some euclidean space $\euclidean{q}$, where $q$ depends on $M$ and we equip 
$\euclidean{q}$ with its usual topology and the standard computable encoding of all open rational 
balls. Thus every compact computable manifold that is computably Hausdorff can be seen as a 
computable submanifold of some euclidean space.

Related work. 
In \cite{realcomputablemanifolds}, Calvert and Miller 
give another definition for the
term ``computable manifold''. Using the BSS model \cite{BCSS98}, \emph{$\reals$-computable
manifolds} are defined in \cite{realcomputablemanifolds}. Informally, an $\reals$-computable
manifold is a topological manifold $M$, together with a finite collection of $\reals$-computable
functions (e.g. computable in the BSS sense), called the inclusion functions, which describe the
inclusion relations of all the domains of the charts of a specific atlas $\{ \varphi_i
\colon U_i \to \euclidean{d} \}_{i\in \naturals}$ on $M$. This definition is used to prove some
results (in the BSS model) about the undecidability of nullhomotopy and simple connectedness in
$\reals$-computable manifolds and also how to determine a presentation of the fundamental group of
such manifolds. However, it is hard to interpret these undecidability results in terms of practical 
computing, because in the BSS model, simple subsets of $\euclidean{2}$ which can be easily 
``drawn'' (i.e., approximated), such as the Koch snowflake and the graph of $y=e^x$ are 
undecidable\footnote{If a set $C\subseteq \euclidean{n}$ is decidable in the BSS model, then it 
must be a countable disjoint union of semi-algebraic sets. This is the reason why simple sets such 
as the graph of $y=e^x$ are not decidable in the BSS model.} in the BSS model \cite{BCSS98}. But we 
know that these sets can be approximated with an arbitrary precision, thus it would seem 
that in general, uncomputability results in the BSS model do not match real world computations.

In \cite{DBLP:journals/corr/Iljazovic13}, Iljazovi\'c studies the 
computability (in the sense of TTE and computable analysis) of compact subsets of computable metric 
spaces. They use these results to show that each semi-computable compact manifold with computable 
boundary is computable, as a subset of a metric space. In \cite{DBLP:journals/corr/BurnikI14}, the 
authors prove similar results for 1-manifolds, not necessary compact. In this paper, 
our focus 
is to develop the basic concepts and results necessary to build an effective theory of manifolds.

The paper is organized as follows. Section \ref{secPrelim} contains many definitions and results
regarding topological manifolds, computable analysis and topology. The reader which is familiar
with some or all of the contents of this section can skip the corresponding parts. In Section
\ref{secComputableManifolds}, we introduce computable atlases and we define the computable topological
space associated with a computable atlas, later on, computable structures (equivalence classes of
computable atlases) are defined using computably compatible computable atlases. We show that a
computable structure on a given set is characterized by an equivalence class of computable
topological spaces (Definition \ref{defEquivCompTOPSPACES}). We also study the relationships of
computable manifolds with computable functions and finally, we prove some basic properties of
computable manifolds. Section \ref{secComputableSubmanifolds} is devoted to the important concept of
\emph{computable submanifold}, that is, a computable manifold inside another manifold. In Section 
\ref{secImbeddings}, 
we prove an effective version of one of the most useful tools for handling manifolds: We show 
that any compact computable manifold (computably Hausdorff) can be embedded as a computable 
submanifold in an euclidean space of sufficiently high dimension, with a computable embedding. 
Throughout the paper, we give examples to explain the main ideas behind concepts and
results. Section \ref{secFinalRemarks} contains our concluding remarks.

\section{Notations and preliminaries}\label{secPrelim}

In this section, we summarize many definitions, facts and technical details about topological
ma\-ni\-folds and computable topology that will be used throughout the rest of the paper. The 
reader can check the references for more complete information on these topics.

\subsection{Basic notations}\label{secTopology}

Our basic references for topology are \cite{armstrong,engelking1989general}. The power set of any 
set $A$ will be denoted by $\powersetof{A}$. A function $f\colon A \to B$ between the sets 
$A,B$, which is defined on a subset of $A$, is called a \emph{partial function} and is denoted by 
$f\colon \subseteq A \to B$. When $f$ is defined on the entire set $A$, $f$ is called a \emph{total 
function} and we omit the ``$\subseteq$'' symbol. The \emph{identity} function on $A$ is
$1_A\colon A \to A$. For a topological space $(X,\tau)$, we denote by $\mathcal{A}$ the set of
closed subsets of $X$ and $\mathcal{K}$ will denote the set of compact subsets of $X$. A set
$Z\subseteq X$ is a \emph{$G_\delta$-set} if $Z$ is a countable intersection of open sets in $X$.
The symbols $\naturals,\integers,\rationals$ and $\reals$ are used to represent the set of
\emph{natural numbers}; the set of \emph{integers}; the set of \emph{rational numbers} and the set
of \emph{real numbers} respectively. Let $\euclidean{n}=\{ (x_1, \ldots, x_n) \mid x_i \in \reals
\}$ denote \emph{euclidean space} of dimension $n \geqslant 0$.
If $(M,d)$ is a metric space with metric $d\colon M \times M \to \reals$, then for any $x\in M$ and
$0<r\in\reals$, let 
\[ B(x, r) = \{ y \in M \mid d(x,y) < r \} \text{ and } \overline{B}(x, r) = \{ y \in M \mid 
d(x,y)
\leqslant r \}\]
denote the \emph{open} and \emph{closed ball} with center $x$ and radius $r$. We endow the 
euclidean 
space $\euclidean{n}$ with its standard topology, which is induced by the metric $d(x,y) = \lVert x 
- y \rVert$; $\sphere{n-1}=\{ x \in \euclidean{n} \mid \lVert x \rVert = 1 \}$ is the $n$-sphere. 
All subsets of $\euclidean{n}$ are equipped with subspace topology.

\subsection{Topological manifolds}

The main objects of study in this paper are topological manifolds
\cite{kirby1977foundational,brickellclarkDIFF,Freedman4manifolds,1974TOPPROPERTIES}. We first
present one of the most common approaches to define them.

\begin{definition}\label{defTOPManifoldLES}
 By a \emph{topological $n$-manifold} (or just \emph{manifold}) we mean a space $M$ such that
every point in $M$ has an open neighborhood which is homeomorphic to an open subset of euclidean
space $\euclidean{n}$. The integer $n$ is called the \emph{dimension} of the manifold and is
denoted by $\dim M$.
\end{definition}

According to this definition, a topological manifold is just a \emph{locally euclidean space} in
which all the points have the same local dimension. Most authors require their manifolds to satisfy
further topological properties. Manifolds are usually assumed to be paracompact and Hausdorff. In
this paper, we will not assume any of these properties until we explicitly require them. 
It is easy to prove that the next result holds 

\begin{proposition}\label{thmEquivalencesTOPMANIFOLD}
The following are equivalent for a space $M$.
1) $M$ is a topological manifold.
 2) Every point of $M$ has a neighborhood homeomorphic to an open ball in $\euclidean{n}$.
3) Every point of $M$ has a neighborhood homeomorphic to $\euclidean{n}$ itself.
\end{proposition}

A neighborhood in $M$ homeomorphic to an open ball in $\euclidean{n}$ is called an \emph{euclidean
ball}. The set of all euclidean balls in $M$ form a basis for the topology of $M$. Being a
topological manifold is a topological property. If $M$ is a manifold and $f\colon M \to X$
is a homeomorphism, then $X$ is also a topological manifold. 
Many examples of topological manifolds exists, we present the following small list:
 i) Euclidean space $\euclidean{n}$ is the prototypical $n$-manifold; 
ii) any discrete space is a 0-manifold; 
iii) surfaces are manifolds of dimension 2; 
iv) a circle is a 1-manifold. In fact, the $n$-sphere $\sphere{n}\subset \euclidean{n+1}$ is a
compact $n$-manifold; 
v) any open subset of a $n$-manifold is a $n$-manifold with the subspace topology; 
vi) if $M$ is a $m$-manifold and $N$ is a $n$-manifold, the topological product $M\times N$ is a
manifold such that $\dim (M\times N)=n+m$; 
vii) the disjoint union of a family of $n$-manifolds is a $n$-manifold.
We will see more examples later.

\subsubsection*{Charts and atlases}

There is another way to define when a set $X$ is a topological manifold, without any explicit reference to some topology on $X$. This is done via the concepts of \emph{chart} and \emph{atlas}.

\begin{definition}\label{defChartAtlas}
 Let $X$ be a set. A \emph{coordinate chart} (or just \emph{chart}) on $X$ is a pair
$(\varphi, U)$ where $U\subseteq X$ and $\varphi$ is a bijective function of $U$ onto an open subset
of $\euclidean{n}$. An $n$-dimensional \emph{(topological) atlas} on $X$ is a collection $\{
(\varphi_i, U_i)
\}_{i\in I}$ of coordinate charts on $X$ such that
\begin{itemize}
 \item [a)] the sets $U_i$'s cover $X$;
\item [b)] for each $i,j\in I$, $\directimage{U_i \cap U_j}{\varphi_i}$ is an open subset of $\euclidean{n}$;
\item [c)] Each map (called a \emph{transition function}) $\varphi_j \varphi_i^{-1}\colon
\directimage{U_i \cap U_j}{\varphi_i} \to \directimage{U_i \cap U_j}{\varphi_j}$ is a homeomorphism
between open subsets of
$\euclidean{n}$.
\end{itemize}
\end{definition}

\begin{definition}\label{defInducedTOPAtlas}
Suppose that $\Phi=\{ (\varphi_i, U_i) \}_{i\in I}$ is a topological atlas on $X$. The topology $\tau_\Phi$ induced by $\Phi$ in $X$ is defined as follows: A set $V\subseteq X$ is open if and only if
$\directimage{V\cap U_i}{\varphi_i}$ is open in $\euclidean{n}$ for each $i\in I$.
The topology $\tau_\Phi$ has two very important properties:
\begin{itemize}
 \item Each set $U_i$ is open in $X$.
\item Each map $\varphi_i\colon U_i \to \directimage{U_i}{\varphi_i}$ is a homeomorphism.
\end{itemize}
\end{definition}

This says that the set $X$, equipped with the topology $\tau_\Phi$ becomes a topological manifold as
per Definition \ref{defTOPManifoldLES}. However, a set $X$ may have several different atlases
defined on it and two such atlases need not induce the same topology on $X$. We need to be able to
say when two distinct atlases induce the same topology on $X$.

\begin{definition}\label{defCompatibleAtlases}
 Two $n$-dimensional atlases $\{ (\varphi_i, U_i) \}_{i\in I}, \{ (\psi_\alpha, V_\alpha) \}_{\alpha \in \Lambda}$ on $X$ are \emph{compatible} if their union is an atlas on $X$, so that $\directimage{U_i \cap V_\alpha}{\varphi_i}$ and $\directimage{U_i \cap V_\alpha}{\psi_\alpha}$ are open and $\varphi_i \psi_\alpha^{-1}$ and $\psi_\alpha \varphi_i^{-1}$ are homeomorphisms for all $i\in I, \alpha \in \Lambda$.
\end{definition}

Compatibility is clearly an equivalence relation on the set of all atlases of
dimension $n$ on $X$. An equivalence class $\left[ \Phi \right]$ of atlases on $X$ is called a
\emph{TOP structure} on $X$. 

\begin{lemma}\label{lemmaTOPatlasescompatible}
Let $X$ be a set and $\Phi,\Psi$ two $n$-dimensional atlases on $X$. Then $\Phi$ and $\Psi$
are compatible if and only if they induce the same topology on $X$.
\end{lemma}

\begin{proof}
Suppose that $\Phi=\{ (\varphi_i, U_i) \}_{i\in I},\Psi=\{ (\psi_\alpha, V_\alpha) \}_{\alpha \in \Lambda}$ are
compatible atlases on $X$ and let $U\in \tau_{\Phi}$. By hypothesis, $\Phi$ and $\Psi$ are
compatible and that means, by Definition \ref{defCompatibleAtlases}, that $\Phi \cup \Psi$ is
an atlas on $X$. Using property b) of Definition \ref{defChartAtlas}, we can prove that $V_\alpha\in
\tau_{\Phi}$ for each $\alpha$ and as $U\in \tau_{\Phi}$, $U \cap V_\alpha$ is open in the topology
$\tau_{\Phi}$ and that implies that $\directimage{U \cap V_\alpha \cap U_i}{\varphi_i}$ is open in
$\euclidean{n}$ for all $i$. By c) of Definition \ref{defChartAtlas}, the function
\[ \psi_\alpha \varphi^{-1}_i \colon \directimage{V_\alpha \cap U_i}{\varphi_i} \to \directimage{V_\alpha \cap
U_i}{\psi_\alpha}\]
is an homeomorphism, thus an open map, so that it sends the open set $\directimage{U \cap V_\alpha
\cap U_i}{\varphi_i}$ onto the open set $\directimage{U \cap V_\alpha \cap U_i}{\psi_\alpha}$. Now we know
that $\directimage{U \cap V_\alpha}{\psi_\alpha}=\directimage{(U \cap X) \cap V_\alpha}{\psi_\alpha}=\directimage{(U
\cap (\bigcup_i U_i)) \cap V_\alpha}{\psi_\alpha}=\directimage{(\bigcup_i (U \cap U_i)) \cap
V_\alpha}{\psi_\alpha}=\directimage{\bigcup_i (U \cap U_i \cap V_\alpha)}{\psi_\alpha}=\bigcup_i \directimage{U \cap
U_i \cap V_\alpha}{\psi_\alpha}$. In conclusion, $\directimage{U \cap V_\alpha}{\psi_\alpha} = \bigcup_i \directimage{U
\cap U_i \cap V_\alpha}{\psi_\alpha}$ and this says that $\directimage{U \cap V_\alpha}{\psi_\alpha}$ is open in
$\euclidean{n}$ for each $\alpha$, so that $U\in \tau_{\Psi}$. We have proved that
$\tau_{\Phi}\subseteq \tau_{\Psi}$ and to prove the other inclusion, the argument is symmetric.
Hence, $\tau_{\Phi} =\tau_{\Psi}$ and the two atlases induce the same topology on $X$.

Conversely, assume that $\tau_{\Phi} =\tau_{\Psi}$. Since $V_\alpha \in \tau_\Psi=\tau_\Phi$ then $\directimage{U_i \cap V_\alpha}{\varphi_i}$ is open and since $U_i \in \tau_\Phi=\tau_\Psi$ then $\directimage{V_\alpha \cap U_i}{\psi_\alpha}$ is open. By the properties of the topology $\tau_\Phi=\tau_\Psi$ (see \ref{defInducedTOPAtlas}), $\varphi_i$ and $\psi_\alpha$ are homeomorphisms, therefore $\varphi_i\psi_\alpha^{-1}$ and $\psi_\alpha\varphi_i^{-1}$ are homeomorphisms for all $i\in I, \alpha \in \Lambda$, so that $\Phi$ and $\Psi$ are compatible
atlases, according to Definition \ref{defCompatibleAtlases}.
\end{proof}

Also, if $M$ has a topology $\tau$
and an atlas defined on it, then the topology induced by the atlas is the same as $\tau$ if and only
if each chart $\varphi \colon U \to \euclidean{n}$ is a homeomorphism in the topology $\tau$. In
summary, a TOP structure $\left[ \Phi \right]$ on $X$ is characterized by the topology induced by
every atlas that belongs to $\left[ \Phi \right]$. It is easy to prove the 

\begin{lemma}\label{lemmaTOPSTRUCTTOPMANIFOLD}
 A set $M$ is a topological manifold if and only if it has a TOP structure.
\end{lemma}

The most common way to define topological manifolds is using locally
euclidean spaces. The definitions of charts and atlases are very handy in the context of
differentiable and PL manifolds \cite{brickellclarkDIFF,hudson1969piecewise},
where the transition maps are required to have the additional properties of being $C^k$ $(k\in
\naturals \cup \{ \infty \} )$ and/or piecewise linear (PL) functions respectively.

\subsection{Submanifolds and embeddings}
  
Manifolds which are subsets of other manifolds are an important tool to study topological manifolds
and their properties. A topological manifold $N$ is a \emph{submanifold} of the manifold $M$ if and
only if $N\subset M$ and $N$ is a subspace of $M$ (that is, the topology of $N$ is subspace
topology). An easy example of a submanifold is any open set $U\subseteq M$ of any given manifold. 
The $n$-sphere $\sphere{n}$ is a compact submanifold of $\euclidean{n+1}$.

A
continuous function $h\colon X \to Y$ is a \emph{(topological) embedding} if $h$ is a homeomorphism of 
$X$ onto $\directimage{X}{h}\subset Y$, where $\directimage{X}{h}$ is
equipped with the subspace topology induced by $Y$. Such an embedding is denoted by 
$h\colon X \hookrightarrow Y 
$
When there exists an embedding of $X$ onto $Y$, we can view $X$ as a subspace of $Y$. 

There are many facts
\cite{kirby1977foundational} in the theory of manifolds (topological, differentiable and/or PL)
which can be shown to be true using the well known result that every Hausdorff $n$-manifold 
$M$ (of any kind) embeds in some high dimensional euclidean space $\euclidean{q}$, where $q$ depends
on $n$. Numerous versions of this embedding theorem exist, the big difference between them being 
the dimension of the space $\euclidean{q}$. It was proven by Whitney \cite{WhitneyDIFF36} that if
$M$ is smooth ($C^\infty$), then it embeds in $\euclidean{2n}$ and this is the best result possible.
The same is true for the piecewise linear case using similar constructions to those used in the smooth
case. If $M$ has no additional structure, it can be embedded in $\euclidean{2n+1}$ and again, this
is the lowest possible dimension for the euclidean space. This last result can be proven by means of
dimension theory \cite{munkres2000topology,hurewicz1948dimension,engelking1995theory}.

We will study submanifolds and embeddings from the computational point of view in Sections
\ref{secComputableSubmanifolds} and \ref{secImbeddings}, where we will prove an 
effective (i.e. computable) version of the embedding Theorem.


\subsection{Computability theory}\label{secComputabilityTheory}

We only give a brief summary of definitions and terminology that we will be using in the rest of
the paper. The reader that wishes to check a full introduction to basic
computability theory and Type-2 theory of effectivity (TTE) can see the references
\cite{weihrauch2000computable,theoryrepresentations,Weihrauch:jucs_14_6:the_computable_multi_functions,kozen1997automata,cooper2004computability}. All the
concepts and prior results given here can be found mostly in \cite{weihrauch2000computable,CCAtutorial,KWECTOP}.

\subsubsection*{Basic notions of computability theory}

Let $A,B$ be sets. A \emph{multi-function}\footnote{The definition of multi-function used in this paper is that of standard computable analisys, it must not be confused with the traditional use of this term. In fact, we use multi-functions only to define an encoding of compact subsets of a topological space $X$. We do not need it elsewhere.} from $A$ to $B$ is a triple $f=(A, B, R_f)$ such that
$R_f\subseteq A \times B$ (this is the \emph{graph} of $f$). We will denote it by $f\colon
A \rightrightarrows B$. The inverse of $f$ is the multi-function $f^{-1}= (B, A, R_{f^{-1}})$.
For $X\subseteq A$, let $\directimage{X}{f} = \{ b\in B \mid (\exists a\in X) (a,b) \in R_f\}$,
$\domp{f} = \inverseimage{B}{f}$ and $\rangep{f} = \directimage{A}{f}$. For $a\in A$, let
$f(a)= \directimage{\{ a \}}{f}$. If it happens that for every $a\in A$, $f(a)$ contains at
most one element, $f$ can be treated as  a usual partial function denoted by $f\colon \subseteq
A\to B$. In contrast to relational composition, for multi-functions $f\colon A \rightrightarrows B$
and $g\colon B \rightrightarrows C$, we define the composition $g\circ f \colon A \rightrightarrows
B$ by $a\in \domp{g\circ f} \Leftrightarrow f(a)\subseteq \domp{g}$ and $g\circ f(a) =
\directimage{f(a)}{g}$ \cite[Section 3]{Weihrauch:jucs_14_6:the_computable_multi_functions}.

\begin{definition}\label{defMultifuncRestrictionRange}
 For a multi-function $f\colon X \rightrightarrows Y$  and $Z\subseteq Y$, define $f|^Z\colon 
X\rightrightarrows Z$ by $f|^Z(x) = f(x) \cap Z$ for all $x\in X$.
\end{definition}

Notice that when $f$ is a function, then the multi-funcion $f|^Z$ from Definition 
\ref{defMultifuncRestrictionRange} is simply the usual restriction of $f$ to the set 
$\inverseimage{Z}{f}$.

An \emph{alphabet} is any non-empty finite set $\Sigma = \{ a_1, \ldots, a_n \}$. We assume that any
alphabet that we use contains at least the symbols 0, 1. We denote the set of finite words over
$\Sigma$ by $\sigmastringset$ and with $\sigmaomegastringset$ the set of infinite
sequences\footnote{It is customary to treat an element $p$ of $\sigmaomegastringset$ as an
``infinite word'' $p(0)p(1)p(2)\cdots$.} $p\colon \naturals\to \Sigma$ over $\Sigma$. We will be
using the ``wrapping function'' $\iota \colon \sigmastringset \to \sigmastringset$, defined by
$\iota(a_1a_2\cdots a_k)= 110a_10a_20\cdots0a_k011,$
for coding words in such a way that $\iota(u)$ and $\iota(v)$ cannot overlap. We will be
using standard functions for finite or countable tupling on $\sigmastringset$
and $\sigmaomegastringset$, denoted by $\langle \cdot \rangle$, in particular, $\langle u_1,
\ldots, u_n \rangle = \iota(u_1) \cdots \iota(u_n)$, $\langle u, p \rangle = \iota(u)p, \langle p,q
\rangle=p(0)q(0)p(1)q(1)\cdots$ and $\langle p_0, p_1, \ldots \rangle\langle i, j \rangle =
p_i(j)$ for all $u,u_1,\ldots, u_n \in \sigmastringset$ and $p,q,p_i \in \sigmaomegastringset$
$(i=0,1,\ldots)$.
For $u \in \sigmastringset$ and $w\in \sigmastringset \cup \sigmaomegastringset$, $u\ll w$ if and
only if $\iota(u)$ is a subword of $w$. 

\subsubsection*{Tarski's decision method for the elementary algebra of the reals}

In order to prove some results in the next sections, we will be using a celebrated result from A.
Tarski \cite{Tarski51}. We first give the definition of \emph{elementary expression}.

\begin{definition}
 An \emph{elementary expression} in the algebra of the real numbers, is an expression build with
the following objects: Variables over the real numbers; constants $c\in \naturals$; the symbols 
$+,-,\cdot, \div$ which denote sum, subtraction, multiplication and division of 
real numbers respectively; the symbols $>, =$ that denote the relations 
``greater that'' and ``equal to'' respectively, of real numbers; the logic connectives 
$\vee$ (disjunction), $\wedge$ (conjunction), $\neg$ (negation) and $\Rightarrow$ (implication); 
the universal ($\forall$) and existential ($\exists$) quantifiers.
\end{definition}

Notice that in general, it is impossible to say something about sets of real numbers with
elementary expressions. Although we can give expressions for sets like $\{ 5, 7, 10000 \}$ (with
the elementary expression $x=5 \vee x=7 \vee x=10000$), it is impossible to write down with an
elementary expression the statement ``$x$ is an integer'', that is, the expression $x\in \integers$
is not elementary, neither is the following expression about integer equations:
$(\exists x, y, z \in \integers)(x^3 + y^3 = z^3)$; if such expressions were to be elementary, 
that would imply that all sentences 
of elementary number theory, are elementary expressions, and that would mean that Theorem 
\ref{thmTarskiDecisionMethod} below is false \cite{Tarski51}, thus the previous expressions about 
elements of $\integers$ are not elementary. However, all basic properties of order, sum and 
multiplication of the field $\reals$ can be expressed as elementary expressions. Also, many useful 
properties of the ring $\reals\left[ x_1, \ldots, x_n \right]$ ($n\geqslant 1$) can be given with 
elementary expressions.
We are now ready to introduce Tarski's result.

\begin{theorem}\label{thmTarskiDecisionMethod}
 There exists an algorithm to decide, given an elementary expression in the algebra of the real
numbers, whether it is true or false.
\end{theorem}

More details on Tarski's method can be found in \cite{Tarski51}. From now on, in this 
paper, whenever we need to apply the Tarski's decision method, we assume that we 
have an appropriate encoding of polynomial functions with rational coefficients as strings of 
$\sigmastringset$.

\subsubsection*{Type-2 theory of effectivity and topology}

Let $Y_0,\ldots,Y_n\in \{ \sigmastringset, \sigmaomegastringset \}$ and $Y= \prod_{i=1} Y_i$. A
function $f\colon \subseteq Y \to Y_0$ is called \emph{(Turing) computable} if for some Type-2 
machine $M$, $f$ is the function $f_M$  computed by $M$. Informally, a Type-2 machine is a Turing 
machine which reads from input files (tapes) with finite or infinite inscriptions, operates on some 
work tapes and writes to an output-only tape. For $Y_0=\sigmastringset$, $f_M(y)=w$ if $M$, on 
input 
$y$, halts with the string $w$ on the output tape, and for $Y_0=\sigmaomegastringset$, $f_M(y)=q$ 
if 
$M$ on input $y$, computes forever and writes $q\in \sigmaomegastringset$ on the output-only tape. 
The computable functions on $\sigmastringset$ and $\sigmaomegastringset$ are closed under 
composition and even under programming \cite{Weihrauch:jucs_14_6:the_computable_multi_functions,GTMs}. The composition of computable functions has a computable extension. If $W,Z\subseteq Y$, the 
set $W$ is called \emph{computable enumerable (c.e.) in } $Z$ if there exists a Type-2 machine $M$ 
which halts on input $y\in Y$ if and only if $y\in W$ for all $y\in Z$. Equivalently, $W$ is c.e. 
in 
$Z$ if $W=Z\cap \dom f$ for some computable function $f\colon \subseteq Y \to \sigmastringset$). 
When $Z=Y$, we ommit ``in $Z$''. 

We equip $\sigmastringset$ with its discrete topology and $\sigmaomegastringset$ with the
topology generated by the base $\{ w\sigmaomegastringset \mid w\in \sigmastringset \}$ of open
sets. With these topologies, every computable function is continuous and every c.e. set is open.

\subsubsection*{Notations and representations}

In TTE, computability on finite or infinite sequences of symbols is transferred to other sets by
means of notations and representations, where elements of $\sigmastringset$ or
$\sigmaomegastringset$ are used as ``concrete name'' of abstract objects. We will need the more
general concept of realization via \emph{multi-representations}\footnote{The only multi-representation 
that we need to use in this paper is the multi-representation 
$\kappa \colon \sigmaomegastringset \to \mathcal{K}$ of compact subsets of a topological space $X$.} 
(see \cite[Section 6]{Weihrauch:jucs_14_6:the_computable_multi_functions} for a detailed discussion,
and also \cite{schroder2003admissible,KWECTOP}).

\begin{definition}\label{defMultirepresentation}
A \emph{multi-representation} of a set $M$ is a surjective multi-function
$\gamma\colon Y \rightrightarrows M$ where $Y\in \{ \sigmastringset, \sigmaomegastringset \}$. If
$\gamma$ is single-valued, it is called simply a \emph{representation} of $M$ and if additionally, $Y=\sigmastringset$, then $\gamma$ is called a \emph{notation} of
the set $M$.
\end{definition}

Examples of multi-representations are the canonical notations $\nu_{\naturals}\colon
\sigmastringset \to \naturals$ and $\nu_{\rationals} \colon$ $\sigmastringset \to \rationals$ of the
natural numbers and the rational numbers respectively, and the single-valued representation $\rho
\colon \subseteq \sigmaomegastringset \to \reals$ of the real numbers \cite{weihrauch2000computable}
which is defined by 
\begin{equation}\label{eqdefEuclideanCauchyRep}
 \rho(\# w_0\# w_1\# w_2\# \cdots) = x \Longleftrightarrow \left| x -  \nu_{\rationals}(w_i) \right|
< 2^{-i}, \text{ for all } i\in \naturals.
\end{equation}
This is called the \emph{Cauchy representation} of $\reals$. This idea can be easily generalized to
a representation $\rho^n \colon \subseteq \sigmaomegastringset \to \euclidean{n}$ of $n$-dimensional
euclidean space for all $n\geq 0$. Mathematical examples of multi-representations will be given
later. 

For multi-representations $\gamma_i\colon Y_i \rightrightarrows M_i$ $(0\leqslant i \leqslant n)$,
let $Y=\prod_{i=1} Y_i$, $M=\prod_{i=1} M_i$ and $\gamma\colon Y \rightrightarrows M$,
$\gamma(y_1, \ldots, y_n)=\prod_i \gamma_i(y_i)$. A partial function $h\colon \subseteq Y\to Y_0$
\emph{realizes} the multi-function $f\colon M \rightrightarrows M_0$ if $f(x) \cap \gamma_0 \circ
h(y)\neq \myemptyset$ whenever $x\in \gamma(y)$ and $f(x) \neq \myemptyset$. This means that $h(y)$
is a name of some $z\in f(x)$ if $y$ is a name of $x\in \dom f$. If $f\colon \subseteq M\to M_0$ is
single-valued, then $h(y)$ is a name of $f(x)$ if $y$ is a name of $x\in \dom f$. If only the
representations are single-valued, $\gamma_0 \circ h(y) \in f(x)$ if $\gamma(y)=x$.

The multi-function $f$ is called $(\gamma_1,\ldots,\gamma_n,\gamma_0)$-\emph{continuous
(-computable)} if it has a continuous (computable) realization. 
The 
continuous (computable) functions are closed under composition , even more, they are
closed under GOTO-programming with indirect addressing
\cite{Weihrauch:jucs_14_6:the_computable_multi_functions,GTMs}. 
 A
 point $x\in M_1$ is $\gamma_1$-computable if and only
if $x\in \gamma_1(p)$ for some computable $p\in \dom \gamma_1$. A set $S\subseteq M$ is
$(\gamma_1,\ldots,\gamma_n)$-c.e. 
if there is a c.e. 
set $W\subseteq Y$ such that
\[ x\in S \Leftrightarrow y\in W \]
for all $x,y$ with $x\in \gamma(y)$.
Therefore, $S\subseteq M$ is $(\gamma_1,\ldots, \gamma_n)$-c.e. if and only if there is a Type-2
machine 
that halts on input $y\in \dom \gamma$ if and only if
$y$ is a name of some $x\in S$. 

Finally, we say that $\gamma_1$ is \emph{reducible} to $\gamma_0$ ($\gamma_1 \leq \gamma_0$) if
$M_1\subseteq M_0$ and the inclusion $i_{M_1}\colon M_1 \hookrightarrow M_0$ is
$(\gamma_1,\gamma_0)$-computable. This means that some computable function $h$ translates
$\gamma_1$-names to $\gamma_0$-names, that is, $\gamma_1(p)\subseteq \gamma_0 \circ h(p)$.
\emph{Continuous reducibility} ($\gamma_1 \leq_t \gamma_0$) is defined analogously by means of
continuous functions. Computable and continuous equivalences are defined canonically: 
\[\gamma_1 \equiv \gamma_0\Leftrightarrow
\gamma_1 \leq \gamma_0 \wedge \gamma_0\leq \gamma_1 \quad\text{and}\quad \gamma_1\equiv_t \gamma_0
\Leftrightarrow \gamma_1 \leq_t \gamma_0 \wedge \gamma_0 \leq_t\gamma_1.\]
Two multi-representations induce the same computability (continuity) if and only if they are
computably equivalent (continuously equivalent). For $X\subseteq M_1$, if $X$ is $\gamma_0$-c.e. and
$\gamma_1\leq \gamma_0$, then $X$ is $\gamma_1$-c.e. 

From the representations $\gamma_1,\gamma_2$, a multi-representation $\left[ \gamma_1,\gamma_2
\right]$ of the product $M_1 \times M_2$ is defined by 
\[\left[ \gamma_1,\gamma_2 \right]\langle y_1,y_2 \rangle = \gamma_1(y_1) \times \gamma_2(y_2).\]
Since the function $(x_1,x_2) \mapsto
(x_1,x_2)$ is $(\gamma_1,\gamma_2, \left[ \gamma_1,\gamma_2 \right])$-computable and
$(x_1,x_2)\mapsto x_i$ is $(\left[ \gamma_1,\gamma_2 \right], \gamma_i)$-computable $(i=1,2)$, a
multi-function $g\colon M_1\times M_2 \rightrightarrows M_0$ is $(\gamma_1,\gamma_2,
\gamma_0)$-computable if and only if $g$ is $(\left[ \gamma_1,\gamma_2 \right],
\gamma_0)$-computable. A set is $(\gamma_1, \gamma_2)$-open if and only if it is $\left[
\gamma_1,\gamma_2 \right]$-open, etc. 

In this paper, we will be using the canonical notation given in \cite{KWECTOP} of finite
subsets and apply Lemma \ref{lemmaFiniteCountableSetsNotationsRepresentations} without
further mentioning.

\begin{definition}\label{defFSCSnotations}
 For the notation $\mu \colon \subseteq \sigmastringset \to M$ 
 define 
the notation $\mu^{\text{fs}}$
   of finite
 subsets of $M$
as follows\footnote{Remember that $v \ll w \Leftrightarrow$ $\iota(u)$ is a subword of $w$. See Section \ref{secComputabilityTheory}.} ($w\in \sigmastringset$) 
\begin{equation}\label{eqdefNotationFiniteSet}
 \mu^{\text{fs}}(w) = W \Longleftrightarrow \begin{cases}
                                (\forall v \ll w) v\in \domp{\mu}, \\
				W = \{ \mu(v) \mid v \ll w \};
                               \end{cases}
\end{equation}
\end{definition}

\begin{lemma}\label{lemmaFiniteCountableSetsNotationsRepresentations}
Let $\mu$ be as in Definition \ref{defFSCSnotations}.
\begin{enumerate}
 \item The set $\domp{\mu^{\text{fs}}}$ is computable if $\domp{\mu}$ is computable,
\item The function $(x,y)\mapsto \{ x, y\}$ is $(\gamma, \gamma, \gamma^{\text{fs}})$-computable,
\item $\gamma^\prime \leq \gamma^{\text{fs}}$, where $\gamma^\prime(w)=\{
\gamma(w) \}$,
\item $\beta^{\text{fs}} \leq \gamma^{\text{fs}}$ if $\beta \leq \gamma$.
\end{enumerate}
\end{lemma}

 {\it Remark.} In some cases, the notation $\mu^{\text{fs}}$ will be used to give abstract 
 names to finite unions
or intersections of a collection $\mathcal{C}$ of subsets of a set $X$, where $\mu \colon
\sigmastringset \to \mathcal{C}$. To avoid confusion about which set operation we refer to with the
notation $\mu^{\text{fs}}$, we will denote $\mu^{\text{fs}}$ as $\fsunionbasenotation{\mu}$ when we
want to encode the finite union of elements of $\mathcal{C}$ and when we want to use
$\mu^{\text{fs}}$ to describe finite intersections of the members of $\mathcal{C}$,
we write $\fsintersectionbasenotation{\mu}$ instead of $\mu^{\text{fs}}$.

\subsection{Computable topology}\label{secComputableTopology}

In this section, we introduce the basic concepts of computable topology that we need in order to
define computable manifolds. Our main reference is \cite{KWECTOP}. The most important definition
that we need is that of \emph{computable topological space}.

\begin{definition}[\cite{KWECTOP}]\label{defCompTOPSPACE}
 An \emph{effective topological space} is a 4-tuple 
$\mathbf{X} = (X,\tau, \beta, \nu)$ such that $(X,\tau)$ is a $T_0$-space and $\nu \colon \subseteq 
\sigmastringset \to \beta$ is a notation of a base $\beta\subseteq \tau$. $\mathbf{X}$ is a
\emph{computable topological space} if $\dom \nu$ is computable and there exists a c.e. set
$S\subseteq (\dom \nu)^3$ such that
\begin{equation}\label{eqdefCompTOPSPACE}
\nu(u) \cap \nu(v) = \bigcup \{ \nu(w) \mid (u, v, w) \in S \} \text{ for all }
u,v,w\in \dom \nu,
\end{equation}
\end{definition}
Equation \eqref{eqdefCompTOPSPACE}
says that in a computable topological space the intersection of base elements is
computable\footnote{Using the representation $\theta$ of open sets given in Definition
\ref{defPositiveRepresentationsCompTOPSPACE}, intersection of base elements is
$(\nu,\nu, \theta)$-computable.}. 

\begin{example}[\emph{Computable euclidean space}]\label{exampleCompTOPSPACES}
 Define $\mathbf{R}^n = (\euclidean{n}, \tau^n, \beta^n, \mu^n)$ such that $\tau^n$ is
the usual topology on $\euclidean{n}$ and $\mu^n$ is a canonical notation of the set of all open
balls with rational radii and center. The inclusion of a rational ball in the intersection of
two rational balls can be decided, therefore $\mathbf{R}^n$ is a computable topological space. Since this
computable space is very important throughout all the paper, we fix once and for all the notation
used in this example to denote the elements of $\mathbf{R}^n$. When $n=1$, $\mathbf{R}^1 =
\mathbf{R}$.
\end{example}

More examples of computable topological spaces can be found in \cite{weihrauch2000computable} and
\cite{KWECTOP}. The definition of computable topological space allows us to define 
representations of the points of $X$ and many classes of subsets (open, closed, compact, etc)
\cite{KWECTOP}. All these representations are an important piece to define computability inside
computable spaces. We will use the notations $\fsunionbasenu$ and $\fsintersectionbasenu$ of the
finite unions and finite intersections respectively, of the base sets of a computable topological 
space $\mathbf{X}=(X, \tau, \beta, \nu)$, see Definition \ref{defFSCSnotations} for details on the 
definition of the notations $\fsunionbasenu$ and $\fsintersectionbasenu$. As usual, we assume that 
$\bigcap \myemptyset = X$ and $\bigcup \myemptyset = \myemptyset$.

\begin{definition}
\label{defPositiveRepresentationsCompTOPSPACE}
 Let $\mathbf{X} = (X,\tau, \beta, \nu)$ be an effective to\-po\-lo\-gi\-cal space. Define a
representation $\delta\colon \subseteq \sigmaomegastringset \to X$ of $X$, a representation 
$\theta \colon \subseteq \sigmaomegastringset \to \tau$ of the set of open sets, a representation 
$\psi \colon \subseteq \sigmaomegastringset \to \mathcal{A}$ of the set of closed sets 
and a multi-representation $\kappa \colon \subseteq 
\sigmaomegastringset \rightrightarrows \mathcal{K}$ of the set of compact subsets of $X$ as follows:
\begin{equation}\label{eqdefPositiveRepresentationsCompTOPSPACE_deltaplus}
 x=\delta(p) \Longleftrightarrow (\forall w\in \sigmastringset)(w \ll p \Leftrightarrow w\in 
\dom \nu \text{ and } x\in \nu(w)),
\end{equation}
\begin{equation}\label{eqdefPositiveRepresentationsCompTOPSPACE_thetaplus}
 W=\theta(p) \Longleftrightarrow \begin{cases}
                                w \ll p \Rightarrow w\in \dom \nu, \\
				W = \bigcup_{w \ll p} \nu(w),
                               \end{cases}
\end{equation}
\begin{equation}\label{eqdefPositiveRepresentationsCompTOPSPACE_psiplus}
 A=\psi(p) \Longleftrightarrow (\forall w\in \sigmastringset)(w \ll p \Leftrightarrow w\in 
\dom \nu \text{ and } A \cap \nu(w) \neq \myemptyset),
\end{equation}
\begin{equation}\label{eqdefPositiveRepresentationsCompTOPSPACE_kappaplus}
 K \in \kappa(p) \Longleftrightarrow (\forall z\in \sigmastringset)(z \ll p \Leftrightarrow z\in 
\domp{\fsunionbasenu} \text{ and } K\subseteq \fsunionbasenu(z)),
\end{equation}
\end{definition}

The previous representations give us information about the represented object, that is, about its 
contents. There exists representations that complement the previous ones, in the sense that these 
representations can say something about the ``complements'' of the objects that the representations 
$\delta, \theta$ and $\psi$ are encoding. For our work, we will need only one of these representations.

\begin{definition}
\label{defNegativeRepresentationsCompTOPSPACE}
  Let $\mathbf{X} = (X,\tau, \beta, \nu)$ be an effective topological space. Define a representation
$\psi^-\colon \subseteq \sigmaomegastringset \to \mathcal{A}$ of the set of closed sets by 
$\psi^-(p) = X - \theta(p)$.
\end{definition}
We fix once and for all the following convention for all
the representations induced by the space $\mathbf{R}^n$: Each representation $\gamma$ from
Definitions \ref{defPositiveRepresentationsCompTOPSPACE} and
\ref{defNegativeRepresentationsCompTOPSPACE} will be denoted as $\gamma^n$.

\begin{definition}[\cite{weihrauch2000computable,theoryrepresentations,schroder2003admissible,MatthiasAdmisibility}]\label{defCompTOPSPACEAdmissibleReps}
A representation $\gamma\colon \sigmaomegastringset \to X$ of a topological space $(X, \tau)$ is
called \emph{admissible} (with respect to $\tau$) if it is continuous and $\gamma^\prime\leq_t 
\gamma$ for every con\-ti\-nuous function $\gamma^\prime\colon \sigmaomegastringset \to X$. 
\end{definition}

\begin{proposition}[\cite{weihrauch2000computable}]\label{propDeltaAdmissibleCompTOPSPACE}
If $\mathbf{X}=(X,\tau,\beta,\nu)$ is an effective topological space, then the representation 
$\delta$ is admissible with respect to the topology $\tau$. 
\end{proposition}

It can be proven that all the other (single-valued) representations of Definitions 
\ref{defPositiveRepresentationsCompTOPSPACE} and \ref{defNegativeRepresentationsCompTOPSPACE} are 
admissible with respect to appropriated topologies \cite{schroder2003admissible}. We now present a 
result which gives some nice properties of the representations of a computable topological space 
$\mathbf{X}$ and the union and intersection operation between subsets of $X$, the proof can be 
found in \cite{KWECTOP}.

\begin{theorem}[\cite{KWECTOP}, Theorem 11]\label{thmTheorem11ECT}
 Let $\mathbf{X}=(X,\tau,\beta,\nu)$ be a computable topological space.
\begin{enumerate}
 \item Finite intersection on open sets is $(\nu^{\text{fs}},\theta)$-computable and 
$(\theta^{\text{fs}},\theta)$-computable. 
\item On closed sets, finite union is $((\psi^-)^{\text{fs}},\psi^-)$-com\-pu\-ta\-ble.
\item On the compact sets, finite union is $(\kappa^{\text{fs}},\kappa)$-computable.
\end{enumerate}
\end{theorem}

The following result tell us something about the computability of some basic decision problems in a 
computable topological space and the representations given in Definition 
\ref{defPositiveRepresentationsCompTOPSPACE}. 
it 
will 
be useful in this paper.

\begin{lemma}[\cite{KWECTOP}, Corollary 14]\label{lemCorollary14ECT}
 Let $\mathbf{X}=(X,\tau,\beta,\nu)$ be a computable topological space. Then for all points $x\in 
X$ and open sets $W$
of $X$, the decision problem ``$x\in W$'' is $(\delta, \theta)$-c.e.
\end{lemma}

\subsubsection*{Predicate spaces}

An important class of computable topological spaces can be constructed from very simple
assumptions. Let $X$ be any set and $\sigma\subseteq 2^X$. We may say ``$x$ has property
$U$'' if $x\in U$. For each $x\in X$, let 
\[ \mathcal{P}_x(X)=\{ U\in \sigma \mid x\in U \}. \]

\begin{definition}\label{defAtomicPRedicate}
 Let $X$ be any set. An \emph{effective predicate space} is a triple $\mathbf{Z}=(X, \sigma, \lambda)$ 
 such that $\sigma \subseteq 2^X$ is countable and $\bigcup \sigma = X$, $\lambda \colon \subseteq
\sigmastringset \to \sigma$ is a notation of $\sigma$ and the following assertion holds
\begin{equation}\label{eqdefAtomicPRedicate}
 (\forall x,y \in X)(x=y \Longleftrightarrow \mathcal{P}_x(X)=\mathcal{P}_y(X)).
\end{equation}
$\mathbf{Z}$ is a \emph{computable predicate space} if $\dom \lambda$ is computable. Define the 
representation $\delta_{\mathbf{Z}}\colon \subseteq \sigmaomegastringset \to X$ of $X$ by
$\delta_{\mathbf{Z}}(p)=x \Longleftrightarrow (\forall w\in \sigmastringset)(w \ll p
\Leftrightarrow w\in \dom \lambda \wedge x \in \lambda(w))$. Let 
$T(\mathbf{Z})=(X, \tau_\lambda, \beta_\lambda, \nu_\lambda)$ where $\beta_\lambda$ is
the set of all finite intersections of sets from $\sigma$, $\nu_\lambda =
\fsintersectionbasenotation{\lambda} \colon \subseteq \sigmastringset \to \beta_\lambda$ and
$\tau_\lambda$ is the set of all unions of subsets from $\beta_\lambda$.
\end{definition}


\begin{lemma}[\cite{KWECTOP}]\label{lemmaCompPredicateSpaceISCompTOPSPACE}
  Let $\mathbf{Z}=(X, \sigma, \lambda)$ be an effective predicate space.
\begin{enumerate}
 \item $T(\mathbf{Z})$ is an effective topological space, which is computable if $\mathbf{Z}$ is
computable
(that is, if $\domp{\lambda}$ is computable).
\item Let $\delta_{T(\mathbf{Z})}$ be the inner representation of points for $T(\mathbf{Z})$. Then
$\delta_{T(\mathbf{Z})}
\equiv \delta_{\mathbf{Z}}$.
\item For every representation $\gamma_0$ of a subset $Y\subseteq X$, the set $\{ (x,U) \in
Y\times \sigma \mid x\in U \}$ is $(\gamma_0,\lambda)$-c.e. if and only if $\{ (x, V) \in Y\times
\beta_\lambda \mid x\in V \}$ is $(\gamma_0, \nu_\lambda)$-c.e.
\end{enumerate}
\end{lemma}

Roughly speaking, a $\delta_{\mathbf{Z}}$-name of a point is a list of all of its atomic
predicates, while a $\delta_{T(\mathbf{Z})}$-name is a list of all finite intersections of such sets.
Clearly, the two representations are equivalent.

\begin{example}\label{exampleCompPREDSPACES}
 Define $\mathbf{Z} \defeq (\reals, \sigma, \lambda)$ such that $\dom \lambda=\dom 
\nu_\rationals$, $\lambda(w) \defeq (\nu_\rationals(w),\nu_\rationals(w) + 1)\subset \reals$ and 
$\sigma \defeq \rangep{\lambda}$. For each $w\in \dom \lambda$, $\lambda(w)$ is the open interval 
in 
$\reals$ with the endpoints $\nu_\rationals(w)$ and $\nu_\rationals(w) + 1$. We claim that 
$\mathbf{Z}$ is a computable predicate space. By definition, the set $\dom \lambda$ is a computable 
subset of $\sigmastringset$, because $\dom\nu_\rationals$ is computable. Now we need to show that 
property \eqref{eqdefAtomicPRedicate} is satisfied by $\mathbf{Z}$, to do this, we will prove that 
for $x,y\in \reals$,
\[ x\neq y \Rightarrow \mathcal{P}_x(\reals) \neq \mathcal{P}_y(\reals). \]
So, assume that we have $x,y\in \reals$ such that $x\neq y$. Without loss of generality, suppose 
that $x < y$. Then we can find $q\in \rationals$ such that $x < q < y$ and $d(q,y) < \frac{1}{2}$. 
This implies that $y\in Q$, where $Q=(q, q+ 1)$, so that $Q\in \mathcal{P}_y(\reals)$. Also, as $x 
< q$, we have that $x\notin Q$, thus $Q\notin \mathcal{P}_x(\reals)$, therefore 
$\mathcal{P}_x(\reals) \neq \mathcal{P}_y(\reals)$. So that $\mathbf{Z}$ fulfills Definition 
\ref{defAtomicPRedicate}. By Lemma \ref{lemmaCompPredicateSpaceISCompTOPSPACE}, 
$T(\mathbf{Z})=(\reals, \tau^1, \beta_\lambda, \nu_\lambda)$ is a computable topological space. The 
computable space $T(\mathbf{Z})$ and the computable space $\mathbf{R}$ from Example 
\ref{exampleCompTOPSPACES} are equivalent as we shall see in Definition 
\ref{defEquivCompTOPSPACES}. 
This example can be generalized to show that for any $n\in \naturals$, $\euclidean{n}$ has the 
structure of a computable predicate space $\mathbf{Z}^n$ such that $T(\mathbf{Z}^n)$ is equivalent 
to $\mathbf{R}^n$.
\end{example}

\subsubsection*{Subspaces of computable topological spaces}

We need to consider restrictions and products of effective topological spaces \cite{KWECTOP}. Let
$\mathbf{X}=(X,\tau, \beta, \nu)$ be an effective topological space. For a subspace $B\subseteq X$,
define the restriction $\ctopcompsubspace{\mathbf{X}}{B}=(B,\tau_B, \beta_B, \nu_B)$ of $X$ to $B$ 
by $\domp{\nu_B} = \dom \nu,\nu_B(w) = \nu(w) \cap B, \beta_B= \rangep{\nu_B}$ and 
$\tau_B = \{W\cap B\mid W\in \tau \}$. Let $\delta_B,\theta_B,\ldots, \psi_B^{-}$ be the 
representations for $\ctopcompsubspace{\mathbf{X}}{B}$ from Definitions 
\ref{defPositiveRepresentationsCompTOPSPACE} and
\ref{defNegativeRepresentationsCompTOPSPACE}. Remember from Definition 
\ref{defMultifuncRestrictionRange} that for a multi-function $f\colon X \rightrightarrows Y$  and 
$Z\subseteq Y$, the multi-function $f|^Z\colon X\rightrightarrows Z$ is defined by $f|^Z(x) = 
f(x) \cap Z$ for all $x\in X$. The next result is proven in \cite[Lemma 26]{KWECTOP}.

\begin{lemma}\label{lemmaPropertiesEffectiveTOPSUBSPACE}
 $\ctopcompsubspace{\mathbf{X}}{B}$ is an effective topological space, which is computable if $\mathbf{X}$ is
computable. Also, the following properties are satisfied:
\begin{enumerate}
 \item $\delta_B=\delta|^B$,
\item $\theta_B(p) = \theta(p) \cap B$ for all $p\in\domp{\theta_B}=\dom \theta$,
\item $\psi^{-}_B(p) = \psi^{-}(p) \cap B$ for all $p\in\domp{\psi^{-}_B}=\dom \psi^{-}$,
\item $\psi_B|^{\mathcal{C}}=\psi|^{\mathcal{C}}$ for $\mathcal{C}=\{ C\subseteq B \mid C \text{
closed in } X \}$,
\item $\kappa_B|^{\mathcal{L}}=\kappa|^{\mathcal{L}}$ for $\mathcal{L}=\{ K\subseteq B \mid K \text{
compact in } X \}$.
\end{enumerate}
\end{lemma}

\subsubsection*{The product of computable topological spaces}

For $i=1,2$ let $\mathbf{X}_i=(X_i,\tau_i, \beta_i, \nu_i)$ be effective topological spaces with
representations $\delta_i,\theta_i,\ldots,\psi_i^{-}$ from Definitions
\ref{defPositiveRepresentationsCompTOPSPACE} and \ref{defNegativeRepresentationsCompTOPSPACE}. To
convert the space $X_1\times X_2$ with the product topology into an effective topological space
define the product $\overline{\mathbf{X}}=(X_1\times X_2,\overline{\tau}, \overline{\beta},
\overline{\nu})$ of $\mathbf{X}_1$ and $\mathbf{X}_2$ such that $\dom \overline{\nu} = \{ \langle
u_1, u_2 \rangle \mid (u_1,u_2) \in \dom \nu_1 \times \dom \nu_2 \}$, $\overline{\nu}(\langle u_1,
u_2 \rangle) = \nu_1(u_1) \times \nu_2(u_2)$, $\overline{\beta}=\rangep{\overline{\nu}}$ and
$\overline{\tau}$ is the product topology generated by $\overline{\beta}$. The basic properties of
$\overline{\mathbf{X}}$ are proven in \cite{KWECTOP} and are given by the following

\begin{lemma}\label{lemmaPropertiesEffectiveTOPPRODUCT}
$\overline{\mathbf{X}}$ is an effective topological space, which is computable if
$\mathbf{X}_1$ and $\mathbf{X}_2$ are computable. Let
$\overline{\delta},\overline{\theta},\ldots,\overline{\psi}^{-}$ be the representations for
$\overline{\mathbf{X}}$ from Definitions \ref{defPositiveRepresentationsCompTOPSPACE} and
\ref{defNegativeRepresentationsCompTOPSPACE}. Then
\begin{enumerate}
 \item $\overline{\delta}\equiv \left[ \delta_1,\delta_2 \right]$.
\item The function $(x_1,x_2)\mapsto (x_1,x_2)$ is $(\delta_1, \delta_2,
\overline{\delta})$-computable and each projection $(x_1,x_2)\mapsto x_i$
is $(\overline{\delta},\delta_i)$-computable.
\item For open sets, the product $(W_1,W_2)\mapsto W_1\times W_2$ is $(\theta_1, \theta_2,
\overline{\theta})$-computable. 
\item For open sets, the projection $W_1\times W_2\mapsto W_1$ is $(\overline{\theta},
\theta_1)$-computable if the set $Z_2 = \{ w\in\sigmastringset \mid \nu_2(w) \neq
 \myemptyset\}$ is c.e.
\item For closed sets, the product $(A_1,A_2)\mapsto A_1\times A_2$ is $(\psi_1, \psi_2,
\overline{\psi})$-computable and $(\psi_1^{-}, \psi_2^{-},$ $\overline{\psi}^{-})$-computable.
\item For compact sets, the operation $(K_1,K_2)\mapsto K_1\times K_2$ is $(\kappa_1, \kappa_2,
\overline{\kappa})$-computable and the projection $K_1\times K_2\mapsto K_i$ is $(\overline{\kappa},
\kappa_i)$-computable.
\end{enumerate}
\end{lemma}

The generalization to finite products is straightforward. Some examples of computable topological
spaces which are subspaces or products of other computable topological spaces will be given later.

\subsection{Computable functions between computable topological
spaces}\label{subsecCompFunctionsCompTOPSPACES}

In this section, we define \emph{computable functions} between computable topological spaces and prove 
some useful results about them. We also introduce computable embeddings and equivalences of 
computable topological spaces defined on the same topological space.

A partial function $f\colon \subseteq X\to Y$ is continuous if and only if for every $W \in
\tau^\prime$, $\inverseimage{W}{f}$ is open in $\dom f$, that is, $\inverseimage{W}{f} = V \cap
\dom f$ for some $V ∈ \tau$. 
Type-2 theory gives us a surprising connection 
between topology and computability, this is an equivalence for continuity in terms of 
continuous functions from $\sigmaomegastringset$ to itself \cite[Theorem 3.2.11]{weihrauch2000computable}.

\begin{theorem}\label{thmMainTheoremAdmissibleReps}
 Let $\mathbf{X} = (X, \tau, \beta, \nu)$ and $\mathbf{Y} = (Y, \tau^\prime, \beta^\prime, 
\nu^\prime)$ be effective topological spaces. Then a map $f\colon \subseteq X\to Y$ is continuous 
if and only if $f$ has a continuous $(\delta_X, \delta_Y)$-realization.
\end{theorem}

This is the ``main theorem'' for admissible representations, since for an effective topological 
space $\mathbf{X} = (X, \tau, \beta, \nu)$, by Proposition \ref{propDeltaAdmissibleCompTOPSPACE}, 
the representation of points $\delta$ is admissible with respect to the topology $\tau$. 
\begin{definition}\label{defCompFunctionCompTopSpace}
Let $\mathbf{X}$ and $\mathbf{Y}$ be effective topological spaces. A 
function $f\colon \subseteq X \to Y$ 
is \emph{computable} if $f$ has a computable $(\delta_X, \delta_Y)$-realization.
\end{definition}

We also say that $f$ is $(\delta_X,\delta_Y)$-computable. The following result is an immediate 
consequence of Theorem \ref{thmMainTheoremAdmissibleReps}.

\begin{proposition}\label{propCompFunctContFunc}
  if $f$ is a computable function between the computable spaces $\mathbf{X}$ and $\mathbf{Y}$, 
  then $f$ is continuous.
\end{proposition}

The definition of computable functions is given in terms of the representations of points of 
$\mathbf{X}$ and $\mathbf{Y}$. Characterizations given in terms of the other representations can be 
derived.

\begin{theorem}[\cite{KWECTOP}]\label{thmequivCompFunctions}
Let $f$ be a function between the computable spaces $\mathbf{X}$ and $\mathbf{Y}$. The following are equivalent:
 a) $f$ is $(\delta_X,\delta_Y)$-computable; 
b) the function $W\mapsto \inverseimage{W}{f}$ is $(\theta_Y, \theta_X)$-computable; 
c) the function $B \mapsto \inverseimage{B}{f}$ is $(\nu_Y, \theta_X)$-computable; 
d) the function $\overline{C} \mapsto \overline{\directimage{C}{f}}$ is $(\psi_X,
\psi_Y)$-computable; 
e) the function $K \mapsto \directimage{K}{f}$ is $(\kappa_X, \kappa_Y)$-computable.
\end{theorem}

Many examples of computable functions exists. Almost all known real functions are computable, like
sum, multiplication, the trigonometric functions and their inverses, the square root, exponential
and logarithm functions (with computable bases) \cite{weihrauch2000computable}. We will present many
more examples in
Section \ref{secComputableManifolds}.

\subsubsection*{Equivalences between computable topological spaces}

We will be using two types of equivalences between computable topological spaces. The first one is
that of \emph{computable homeomorphism}.

\begin{definition}\label{defComputableHomeomorphism}
  Let $\mathbf{X}=(X,\tau, \beta, \nu)$ and $\mathbf{X}^\prime=(X^\prime,\tau^\prime, \beta^\prime,
\nu^\prime)$ be computable topological spaces 
. A 
\emph{computable homeomorphism}
 is a map
$h\colon X \to X^\prime$ such that $h$ is a homeomorphism and also $h,h^{-1}$ are computable
functions
. When such a map $h$ exists, we say that $\mathbf{X}$
and $\mathbf{X}^\prime$ are 
\emph{computably homeomorphic}
. This fact is denoted by
$\mathbf{X} \chomeom \mathbf{X}^\prime$.
\end{definition}

Sometimes, if there is no confusion about which (computable) topologies are using the sets
$X,X^\prime$, we will just say that $X$ and $X^\prime$ themselves are computably homeomorphic.

\begin{example}\label{exampleBallCTEuclideanSpace}
 Let $B=B(0, 1)\subset \euclidean{n}$ be the open unit ball with the (computable) subspace topology.
The function $h\colon B \to \euclidean{n}$ defined by $h(x)=\frac{x}{1 - \lVert x \rVert^2}$ is a
$(\delta_B,\delta^n)$-computable map and its inverse is the function $h^{-1} \colon \euclidean{n}
\to B$, given by 
\[h^{-1}(y)=\frac{2y}{1 + \sqrt{1 + 4\lVert y \rVert^2}},\] 
which is $(\delta^n,\delta_B)$-computable, so that $h$ is a computable homeomorphism between
$\ctopcompsubspace{\mathbf{R}^n}{B}$ and $\mathbf{R}^n$. 
\end{example}

\begin{example}\label{exampleStereographicProjection}
 Let $\mathbf{S}^n = \ctopcompsubspace{\mathbf{R}^{n+1}}{\sphere{n}}$ denote the $n$-sphere $\sphere{n}$ equipped
with the computable subspace topology induced by $\mathbf{R}^{n+1}$. An important example of a
computable homeomorphism is the stereographic projection $s\colon
\puncturedsphere{n} \to \euclidean{n}$ ($\puncturedsphere{n}=\sphere{n} - \{ P \}, P=(0, \ldots, 0, 1)$) onto euclidean space
$\euclidean{n}$ given by the equation
\[ s(x,t) = \frac{x}{1 - t},\quad x= (x_1,\ldots, x_n).\]
Thus, we have that $\ctopcompsubspace{\mathbf{S}^{n}}{\puncturedsphere{n}} \chomeom \mathbf{R}^{n}$.
\end{example}

The homeomorphism $s$ has been
fundamental for the work in \cite{zhoucompfuncopenclosed}. 
Our second equivalence between computable topological spaces is as follows.

\begin{definition}[\cite{KWECTOP}]\label{defEquivCompTOPSPACES}
 The computable topological spaces $\mathbf{X}=(X,\tau, \beta, \nu)$ and $\mathbf{X}^\prime=(X,\tau,
\beta^\prime,$ $\nu^\prime)$ are \emph{equivalent} if and only if  $\nu\leq \theta^\prime$ and
$\nu^\prime \leq \theta$, where $\theta$ is the representation of $\tau$ defined by $\nu$ and $\theta^\prime$ is the representation of $\tau$ defined by $\nu^\prime$.
\end{definition}

The equivalence given in the previous definition identifies when the two notations $\nu,\nu^\prime$
of the bases $\beta$ and $\beta^\prime$ respectively, induce the same computability on the space
$(X,\tau)$. 
Example \ref{exampleCompPREDSPACES} 
shows two equivalent computable topological spaces. 

The computability concepts introduced in Definitions \ref{defPositiveRepresentationsCompTOPSPACE}
and \ref{defNegativeRepresentationsCompTOPSPACE} can be called ``computationally robust'', since
they are the same for equivalent computable topological spaces. Usually, non-robust concepts
\cite{weihrauch2000computable,BrattkaW99,BrattkaP03} have only few applications. The next result, 
which can be found in \cite[Theorem 22.2]{KWECTOP} is very useful.

\begin{theorem}\label{thmRobustnessEquivTOPSPACES}
 Let $\mathbf{X}=(X,\tau, \beta, \nu)$ and $\mathbf{X}^\prime=(X,\tau, \beta^\prime, \nu^\prime)$
be computable topological spaces. The following statements are equivalent: 
 1) $\mathbf{X}$ and $\mathbf{X}^\prime$ are equivalent;
 2) $\delta \equiv \delta^\prime$;
 3) $\theta \equiv \theta^\prime$.
Moreover, if $\mathbf{X}$ and $\mathbf{X}^\prime$ are equivalent, then $\gamma \equiv
\gamma^\prime$ for every representation $\gamma$ from Definitions
\ref{defPositiveRepresentationsCompTOPSPACE} and \ref{defNegativeRepresentationsCompTOPSPACE},
where $\gamma^\prime$ is the representation for $\mathbf{X}^\prime$ corresponding to $\gamma$.
\end{theorem}

\begin{corollary}\label{corEquiTOPSPACEIDENTITYCOMPHOM}
 Let $\mathbf{X}=(X,\tau, \beta, \nu)$ and $\mathbf{X}^\prime=(X,\tau, \beta^\prime, \nu^\prime)$
be computable topological spaces. Then $\mathbf{X}$ and $\mathbf{X}^\prime$ are equivalent if
and only if $1_X \colon X \to X$ is a computable homeomorphism.
\end{corollary}

\subsubsection*{The induced effective topological space}

Let $\mathbf{X}=(X,\tau, \beta, \nu)$ be an effective topological space and $Y$ a set and suppose
that there exists a bijective function $f\colon X \to Y$. Since $(X,\tau)$ is a topological space,
$f$ induces a topology $\tau_f$ on $Y$ given by the base 
\[\beta_f = \{ \directimage{V}{f} \mid V \in \beta \}.\] 
The topology defined by the base $\beta_f$ is called the \emph{push-forward} topology. We define a notation $\nu_f \colon \subseteq \sigmastringset \to \beta_f$ of the base
$\beta_f$ as $\nu_f(w) = \directimage{\nu(w)}{f}$ ($\dom \nu_f = \dom \nu)$. The 4-tuple $(Y,
\tau_f, \beta_f, \nu_f)$ will be denoted by $\mathbf{Y}_f$. It is easy to see that $f$ becomes a
homeomorphism between the two spaces $(X,\tau)$ and $(Y,\tau_f)$. The following results tell us that
this fact still holds in a computable way.

\begin{lemma}\label{lemmaTopSpaceCompFuncRepSet}
 Let $\mathbf{X} = (X, \tau, \beta,\nu)$ be an effective topological space and
$Y$ a set. If $f\colon X \to Y$ is a bijective function, then 
\begin{itemize}
 \item [a)] $\mathbf{Y}_f$ is an effective topological space, which is computable if $\mathbf{X}$ is
computable. Also, for any $w\in \dom \nu_f$ $\nu_f(w) = \directimage{\nu(w)}{f}$.
\item [b)] $f$ becomes a computable homeomorphism between $\mathbf{X}$ and $\mathbf{Y}_f$.
\item [c)] Let $\gamma$ be any of the representations given in Definitions
\ref{defPositiveRepresentationsCompTOPSPACE} and \ref{defNegativeRepresentationsCompTOPSPACE}
induced by $\mathbf{X}$ and let $\gamma_f$ be the respective representation induced by
$\mathbf{Y}_f$. Then
\begin{enumerate}
 \item $\dom \gamma_f = \dom \gamma$.
\item if $\gamma=\delta$ then $\gamma_f = f \circ \gamma$.
 \item If $\gamma \in \{ \theta, 
 \psi^- \}$, then $\gamma_f(p) =
\directimage{\gamma(p)}{f}$.
\item If $\gamma = \kappa$, then $\gamma_f(p) = \{ \directimage{D}{f} \mid D\in \gamma(p) \}$.
\end{enumerate}
\end{itemize}
\end{lemma}

\begin{proof}
$a)$ It is clear that $\mathbf{Y}_f$ is an effective topological space. If $\mathbf{X}$ is
computable, then $\dom \nu$ is computable and there exists a c.e. set $S \subseteq \dom \nu \times
\dom \nu \times \dom \nu$ such that equation \eqref{eqdefCompTOPSPACE} of Definition
\ref{defCompTOPSPACE} is satisfied. Since $\dom \nu_f = \dom \nu$, $\dom \nu_f$ is computable and the
same c.e. set $S$ can be used to fulfill \eqref{eqdefCompTOPSPACE} for the effective space
$\mathbf{Y}_f$ in Definition \ref{defCompTOPSPACE}. Thus $\mathbf{Y}_f$ is a computable topological
space.

$b)$ We can show that $f$ is a computable homeomorphism between $\mathbf{X}$ and $\mathbf{Y}_f$ by
proving that the functions
\[ B \mapsto \inverseimage{B}{f}\quad\text{and}\quad A \mapsto \inverseimage{A}{(f^{-1})},\quad (B
\in \beta, A \in \beta_f), \]
are $(\nu_f,\theta)$-computable and $(\nu,\theta_f)$-computable respectively. The proof is 
straightforward, thus we omit it.

$c)$ We argue by cases on the induced representations $\gamma$ and $\gamma_f$:
\begin{itemize}
 \item [] {\it Case $\gamma=\delta$}. Let $p\in \sigmaomegastringset$, if $\delta(p)=x\in X$,
then $f(x)\in Y$ and
\begin{eqnarray*}
 \delta(p) = x & \Leftrightarrow & (\forall w \in \sigmastringset) (w \ll p \Leftrightarrow w \in 
\dom
\nu \wedge x \in \nu(w)) \\
& \Leftrightarrow & (\forall w \in \sigmastringset) (w \ll p \Leftrightarrow w \in \dom \nu_f \wedge
f(x) \in \directimage{\nu(w)}{f}) \\
& \Leftrightarrow & (\forall w \in \sigmastringset) (w \ll p \Leftrightarrow w \in \dom \nu_f \wedge
f(x) \in \nu_f(w)) \\
& \Leftrightarrow & \delta_f(p) = f(x).
\end{eqnarray*}
Thus $\dom \delta = \dom \delta_f$ and also $\delta_f(p) = f(x)=f(\delta(p))=f \circ \delta (p)$
for any $p\in \dom \delta_f$. Therefore, $\delta_f = f \circ \delta$.
\item  [] {\it Case $\gamma=\theta$}. Given $p\in \sigmaomegastringset$, we have that
\begin{eqnarray*}
 \theta(p)=W & \Leftrightarrow & \begin{cases}
                                w \ll p \Rightarrow w\in \dom \nu, \\
				W = \bigcup \{ \nu(w) \mid w \ll p \},
                               \end{cases} \\
& \Leftrightarrow & \begin{cases}
                                w \ll p \Rightarrow w\in \dom \nu_f, \\
				\directimage{W}{f} = \bigcup \{ \directimage{\nu(w)}{f} \mid w \ll p
\},
                               \end{cases} \\
& \Leftrightarrow & \begin{cases}
                                w \ll p \Rightarrow w\in \dom \nu_f, \\
				\directimage{W}{f} = \bigcup \{ \nu_f(w) \mid w \ll p
\},
                               \end{cases} \\
& \Leftrightarrow & \theta_f(p)=\directimage{W}{f},
\end{eqnarray*}
so that $\dom \theta = \dom \theta_f$ and
$\theta_f(p)=\directimage{W}{f}=\directimage{\theta(p)}{f}$ for any $p\in \dom \theta_f$.
\item  [] {\it Case $\gamma=\psi$}. For $p\in \sigmaomegastringset$,
\begin{eqnarray*}
 \psi(p)=A & \Leftrightarrow & (\forall w \in \sigmastringset) (w \ll p \Leftrightarrow w \in \dom
\nu \wedge A \cap \nu(w) \neq \myemptyset) \\
& \Leftrightarrow & (\forall w \in \sigmastringset) (w \ll p \Leftrightarrow w \in \dom
\nu_f \wedge \directimage{A}{f} \cap \directimage{\nu(w)}{f} \neq \myemptyset) \\
& \Leftrightarrow & (\forall w \in \sigmastringset) (w \ll p \Leftrightarrow w \in \dom
\nu_f \wedge \directimage{A}{f} \cap \nu_f(w) \neq \myemptyset) \\
& \Leftrightarrow & \psi_f(p)=\directimage{A}{f}.
\end{eqnarray*}
And we conclude this case with the same arguments we use in all previous cases.
\item [] 
{\it Case $\gamma=\psi^-$}.  
This case follows immediately from the definitions
of the representation 
$\psi^-$. and the 
previous cases.
\item [] {\it Case $\gamma=\kappa$}. Let $p\in \sigmaomegastringset$, notice first that if
$\kappa(p)$ is defined, then $\kappa(p)$ is in general a set of compact subsets of $X$ and
accordingly to \eqref{eqdefPositiveRepresentationsCompTOPSPACE_kappaplus},
\[ K \in \kappa(p) \Longleftrightarrow (\forall z\in \sigmastringset)(z \ll p \Leftrightarrow
K\subseteq \fsunionbasenu(z)), \]
and $\fsunionbasenu(z)=\nu(w_1) \cup \cdots \cup \nu(w_s)$ with $w_i\in \dom \nu \wedge w_i
\ll z$ for all $i=1,\ldots,s$ (see Definition \ref{defFSCSnotations} and the remark after Lemma
\ref{lemmaFiniteCountableSetsNotationsRepresentations}). As each $w_i \in \dom \nu_f$, we have that
$\directimage{\fsunionbasenu(z)}{f}=\directimage{\nu(w_1) \cup \cdots \cup
\nu(w_s)}{f}=\directimage{\nu(w_1)}{f} \cup \cdots \cup \directimage{\nu(w_s)}{f}=\nu_f(w_1) \cup
\cdots \cup \nu_f(w_s)=\fsunionbasenu_f(z)$. Thus for a compact set $K \in \kappa(p)$,
\begin{eqnarray*}
 K \in \kappa(p) & \Leftrightarrow & (\forall z\in \sigmastringset)(z \ll p \Leftrightarrow
K\subseteq \fsunionbasenu(z)), \\
& \Leftrightarrow & (\forall z\in \sigmastringset)(z \ll p \Leftrightarrow
\directimage{K}{f} \subseteq \directimage{\fsunionbasenu(z)}{f}), \\
& \Leftrightarrow & (\forall z\in \sigmastringset)(z \ll p \Leftrightarrow
\directimage{K}{f} \subseteq \fsunionbasenu_f(z)), \\
& \Leftrightarrow & \directimage{K}{f} \in \kappa_f(p).
\end{eqnarray*}
This shows that $\dom \kappa = \dom \kappa_f$ and also that $\kappa_f(p)= \{ \directimage{K}{f}
\mid K \in \kappa(p) \}$.
\end{itemize}
This concludes all the cases for the induced representations on $\mathbf{X}$ and $\mathbf{Y}_f$ and
finishes the proof.
\end{proof}

This result tell us that the only essential difference between $\mathbf{X}$ and $\mathbf{Y}_f$ is
the ``abstract symbol'' $f$. But more can be said if $Y$ already possesses a computable topology.

\begin{corollary}\label{corCompHomeoTopSpaces}
 Let $\mathbf{X} = (X, \tau_X, \beta_X,\nu_X)$, $\mathbf{Y} = (Y, \tau_Y, \beta_Y,\nu_Y)$ be
computable topological spaces. If $f\colon X \to Y$ is a computable homeomorphism, then
$\mathbf{Y}$ and $\mathbf{Y}_f$ are equivalent.
\end{corollary}

\begin{proof}
Let $\delta_X,\delta_Y$ be the inner representations induced by $\mathbf{X}$ and $\mathbf{Y}$
respectively. 
We will show that $\delta_f \equiv \delta_Y$. Since $f$ and $f^{-1}$ are computable functions 
between the
computable spaces $\mathbf{X}$ and $\mathbf{Y}$, they are computable with respect to $\delta_X$
and $\delta_Y$, thus there exist computable functions $F,G\colon \subseteq \sigmaomegastringset \to
\sigmaomegastringset$ such that the two inner squares in the following diagram
\begin{equation}\label{dialemmaTopSpaceCompFuncRepSet}
\begin{CD}
 \sigmaomegastringset @> F >> \sigmaomegastringset @> G >> \sigmaomegastringset \\
@V \delta_X VV @V \delta_Y VV @V \delta_X VV \\
 X @> f >> Y @> f^{-1} >> X 
\end{CD}
\end{equation}
commute. By part (c) of Lemma \ref{lemmaTopSpaceCompFuncRepSet} and the first square, $\delta_f = f 
\circ \delta_X = \delta_Y \circ F$, that
is, $\delta_f \leq \delta_Y$. By the second square, $f^{-1} \circ \delta_Y = \delta_X \circ
G$, so that $ \delta_Y = f \circ \delta_X \circ G = \delta_f \circ G$, thus $\delta_Y \leq
\delta_f$. This proves that $\delta_Y \equiv \delta_f$ and then we can apply Theorem
\ref{thmRobustnessEquivTOPSPACES} to deduce that $\mathbf{Y}$ and $\mathbf{Y}_f$ are equivalent
\end{proof}

For a computable predicate space $\mathbf{Z} = (X, \sigma, \lambda)$, we have the computable space 
$T(\mathbf{Z})=(X, \tau, \beta_\lambda, \nu_\lambda)$ of Lemma
\ref{lemmaCompPredicateSpaceISCompTOPSPACE}, 
where $\nu_\lambda(\iota(u_1)\cdots \iota(u_k))=\lambda(u_1) \cap \cdots \cap \lambda(u_k)$, $\tau$
is 
the topology generated by the
subbase $\sigma$ and $\delta_\lambda \equiv \delta_{\mathbf{Z}}$. If $\sigma$ is not only a subbase,
but
a base of $\tau$, then we can construct the effective space $\mathbf{Y}= (X,\tau, \sigma,
\lambda)$,
which may be computable. For the topology $\tau$ we have on the one hand, the basis $\beta_\lambda$ 
with
the 
notation $\nu_\lambda$ (defined via formal intersection of subbase elements) and on the other hand,
the 
basis $\sigma$ with notation $\lambda$. The question is: Are $T(\mathbf{Z})$ and $\mathbf{Y}$
equivalent ? 
The answer is right here \cite[Lemma 23]{KWECTOP}.

\begin{lemma}\label{lemmaEquivalentCompTOPSPACEBaseSubbase}
Let $\mathbf{Y}= (X,\tau, \sigma, \lambda)$ be an effective topological space such that 
$\mathbf{Z} = (X, \sigma, \lambda)$ is a computable predicate space. Then $T(\mathbf{Z})$ and
$\mathbf{Y}$ 
are equivalent if and only if $\mathbf{Y}$ is a computable topological space.
\end{lemma}

\subsubsection*{Computable Embeddings}

We now introduce the definition of \emph{computable embedding} for computable topological spaces.

\begin{definition}\label{defCompImbedding}
 Let $\mathbf{X} = (X, \tau_X, \beta_X,\nu_X)$, $\mathbf{Y} = (Y, \tau_Y, \beta_Y,\nu_Y)$ be
computable topological spaces. A \emph{computable embedding} of $\mathbf{X}$ into $\mathbf{Y}$ is a
topological embedding $h\colon X \hookrightarrow Y$ such that $h$ is a computable
homeomorphism of $\mathbf{X}$ onto the computable subspace $\ctopcompsubspace{\mathbf{Y}}{X^\prime}$, where
$X^\prime=\directimage{X}{h}$. 
\end{definition}

By Corollary \ref{corCompHomeoTopSpaces}, the computable topological spaces $\mathbf{X}^\prime_h =
(\directimage{X}{h}, \tau_h, \beta_h, \nu_h)$ and $\ctopcompsubspace{\mathbf{Y}}{X^\prime}$ are equivalent.


\begin{example}\label{exampleImbeddingBallInclusionEuclideanN}
Let $h \colon B \to \euclidean{n}$ ($B = B(0,1)$) be the computable homeomorphism of example
\ref{exampleBallCTEuclideanSpace} and $\iota \colon \euclidean{n} \to \euclidean{n+1}$ be such that
$\iota(x_1, \ldots, x_n) = (0, x_1, \ldots, x_n)$. Clearly $\iota$ is a computable embedding of
$\mathbf{R}^n$ into $\mathbf{R}^{n+1}$. We can compose $h$ with $\iota$ to obtain a computable
embedding $B(0,1) \hookrightarrow \euclidean{n+1}$ of $\ctopcompsubspace{\mathbf{R}^n}{B}$ into $\mathbf{R}^{n+1}$.
\end{example}

\subsection{Computably Hausdorff spaces}

In the theory of computable topology, computable versions of the standard separation axioms
$T_i$ ($i=0,1,2$) has been proposed \cite{GrubbaEffectiveStoneWeierstrass,Weihrauch:jucs_16_18:computable_separation_in_topology}. Some of the relationships between these
computable axioms and many examples are studied in
\cite{Weihrauch:jucs_16_18:computable_separation_in_topology}. One of this axioms will be used in
this paper, which is one of the computable variants of the $T_2$ (Hausdorff) separation axiom. A
space $(X,\tau)$ is $T_2$ or \emph{Hausdorff} if and only if the following condition holds:
\[ (\forall x,y \in X)(x\neq y \Rightarrow (\exists U,V \in \tau) U \cap V = \myemptyset \wedge
x\in U \wedge y \in V). \]
Many examples of Hausdorff spaces exist in the literature 
\cite{munkres2000topology,armstrong,engelking1989general}.
The next definition is a constructive version of the Hausdorff property.

\begin{definition}\label{defSCT2}
 A computable topological space $\mathbf{X}=(X,\tau, \beta, \nu)$ is called \emph{computably
Hausdorff} if there exists a c.e. set $H\subseteq \dom \nu \times \dom \nu$ such that
\begin{eqnarray}\label{eqDefSCT2}
  (\forall (u,v)\in H) (\nu(u) \cap \nu(v) = \myemptyset),\qquad\qquad \\ \label{eqDefSCT21}
  (\forall x,y \in X) (x\neq y \Rightarrow (\exists (u,v)\in H) (x\in \nu(u) \wedge y \in
\nu(v))).
\end{eqnarray}
\end{definition}

\begin{lemma}\label{exampleSCT2EuclideanSpace}
The computable euclidean space $\mathbf{R}^n$ is a computably Hausdorff space.
\end{lemma}

\begin{proof}
We claim that the 
set $H$ defined as
\[ H= \{ (u,v)\in \sigmastringset \times \sigmastringset \mid \mu^n(u) \cap 
\mu^n(v)=\myemptyset \} \]
is c.e. and satisfies Equations \eqref{eqDefSCT2} and \eqref{eqDefSCT21}. By definition, $H$ 
fulfills \eqref{eqDefSCT2} and $H$ is c.e. by Theorem 
\ref{thmTarskiDecisionMethod}. To see that $H$ satisfies \eqref{eqDefSCT21}, we proceed as follows. 
Since $\euclidean{n}$ is Hausdorff, then given any two points $x,y \in \euclidean{n}$ with $x\neq y$, 
there exist open sets $U,V\in \tau^n$ such that $U\cap V = \myemptyset$ and $x \in U, y\in V$. Thus 
we can find two base elements $B_1,B_2 \in \beta^n$ such that $x\in B_1 \subseteq U$ and $y\in B_2 
\subseteq V$. As $\mu^n$ is a notation, it is surjective, so that there exist $u_1,u_2 \in \dom 
\mu^n$ with $\mu^n(u_i)=B_i$, $i=1,2$. Since $U \cap V = \myemptyset$, $\mu^n(u_1) \cap 
\mu^n(u_2) 
= \myemptyset$, so that $(u_1,u_2)\in H$. Therefore $H$ satisfies Definition \ref{defSCT2} and 
$\mathbf{R}^n$ is a computably Hausdorff space.
\end{proof}

The computable Hausdorff axiom of Definition \ref{defSCT2} is the strongest of all the computable
separation axioms of \cite{Weihrauch:jucs_16_18:computable_separation_in_topology}. Of course, A
computably Hausdorff space is a Hausdorff space in the usual sense. We list some properties of
computably Hausdorff topological spaces.

\begin{theorem}[\cite{GrubbaEffectiveStoneWeierstrass,Weihrauch:jucs_16_18:computable_separation_in_topology}]\label{thmPropertiesSCT2}
 Let $\mathbf{X}=(X,\tau, \beta, \nu)$ be a computable topological space.
\begin{enumerate}
\item If $\mathbf{X}$ is computably Hausdorff and $A\subseteq X$, then the computable subspace
$\ctopcompsubspace{\mathbf{X}}{A}$ is computably Hausdorff.
 \item If $\mathbf{X}$ is computably Hausdorff then $\kappa \leq \psi^-$.
\item If $\mathbf{Y}=(Y,\tau^\prime, \beta^\prime, \nu^\prime)$ is another computable topological
space and $\mathbf{X}, \mathbf{Y}$ are computably Hausdorff, then the computable product space
$\overline{\mathbf{X}}=(X\times Y,\overline{\tau}, \overline{\beta}, \overline{\nu})$ is computably
Hausdorff.
\end{enumerate}
\end{theorem}

More information about Computably Hausdorff spaces and many other computable separation axioms,
results, examples and counterexamples can be found in
\cite{Weihrauch:jucs_16_18:computable_separation_in_topology}. The computable Hausdorff property
will be very important for our work on computable manifolds, as it is important the standard
Hausdorff property for topological manifolds.

\section{Computable manifolds}\label{secComputableManifolds}

In this section, we will construct the concept of \emph{computable manifold}. 
The first thing to do is to give the definition of what we call
\emph{computable
structures}. 
Once we have defined computable
manifolds, we will talk about computable functions between computable manifolds.

\subsection{The computable predicate space induced by a topological
atlas}\label{subsecCompPredicateSpacebyAtlas}

Assume that $X$ is a non-empty set and that $\Phi = \{ (\varphi_i, U_i)\}_{i \in I}$ is an
$n$-dimensional atlas on $X$ with index set $I\subseteq \sigmastringset$ (so that $\Phi$ is 
countable). If $\varphi_i \colon U_i \to \directimage{U_i}{\varphi_i}$ is a chart, then
by Definition \ref{defChartAtlas}, $\directimage{U_i}{\varphi_i} \subseteq 
\euclidean{n}$ is an 
open
set,thus
\[ \directimage{U_i}{\varphi_i}=\bigcup_{w\in S_i} \mu^n(w), \]
where $S_i = \{ w \in \domp{\mu^n} \mid \mu^n(w) \subset \directimage{U_i}{\varphi_i} \}$. 
Therefore $U_i=\bigcup_{w\in S_i}
\inverseimage{\mu^n(w)}{\varphi_i}$ $\subset X$. Let $A_\Phi = \sigmastringset \times 
\domp{\mu^n} \subseteq \sigmastringset \times \sigmastringset$ and define the set of strings 
$\Lambda_\Phi\subseteq 
\sigmastringset$ as
\[ \Lambda_\Phi=\{ \langle i, w \rangle \mid (i, w) \in A_\Phi \}, 
\] 
that is, $\Lambda_\Phi$ is just the image of the set $A_\Phi$ under the function 
$\langle \cdot , \cdot \rangle$ 
(see Section \ref{secComputabilityTheory}).
For each $\langle i, w \rangle \in \Lambda_\Phi$, 
define the set $B_{\langle i, w \rangle}\subseteq X$ by
\begin{equation}\label{eqDefBiw}
B_{\langle i, w \rangle} = \begin{cases}
                              \inverseimage{\mu^n(w)}{\varphi_i} & \text{ if } i \in I \text{ and 
}
\mu^n(w) \subseteq \directimage{U_i}{\varphi_i}, \\
			      \myemptyset & \text{ otherwise. }
                             \end{cases}
\end{equation}
Each set $B_{\langle i, w \rangle}$ with $\langle i, w \rangle \in \Lambda_\Phi$ is called a 
\emph{computable ball}. Notice 
that unlike an ordinary ball, a computable ball can be empty. Because the set $\{ 
U_i \}_{i\in I}$ covers $X$, we can see that the set $\mathcal{B}_\Phi=\{ B_{j} \mid j \in 
\Lambda_\Phi \}$ also covers $X$. 

\begin{lemma}\label{lemBallsBaseAtlas}
The set $\mathcal{B}_\Phi$ is a base of the topology induced by $\Phi$.
\end{lemma}

\begin{proof}
Let $x \in B_{\langle i, w \rangle} \cap B_{\langle j, v \rangle} = \inverseimage{\mu^n(w)}{\varphi_i} \cap \inverseimage{\mu^n(v)}{\varphi_j}$, then we have that 
\[\varphi_i(x) \in \directimage{\inverseimage{\mu^n(w)}{\varphi_i}}{\varphi_i} \cap 
\directimage{\inverseimage{\mu^n(v)}{\varphi_j}}{\varphi_i} = \mu^n(w) \cap \directimage{\inverseimage{\mu^n(v)}{\varphi_j}}{\varphi_i}.\]
Since the rational open balls are a base of $\euclidean{n}$, there is a ball 
$\mu^n(z)$ such that 
$\varphi_i(x) \in \mu^n(z) \subset \mu^n(w) \cap \directimage{\inverseimage{\mu^n(v)}{\varphi_j}}{\varphi_i}$. 
Therefore $x \in \inverseimage{\mu^n(z)}{\varphi_i} \subset \inverseimage{\mu^n(w)}{\varphi_i} \cap \inverseimage{\mu^n(v)}{\varphi_j} = B_{\langle i, w \rangle} \cap B_{\langle j, v \rangle}$. Since each $\varphi_i$ is continuous in the topology $\tau_\Phi$, hence each $B_{\langle i, w \rangle} =  \inverseimage{\mu^n(w)}{\varphi_i}$ is open
 in $\tau_\Phi$, so that $\mathcal{B}_\Phi$ is a base for $\tau_\Phi$.
\end{proof}

Let the notation $\lambda_\Phi \colon \subseteq \sigmastringset 
\to \mathcal{B}_\Phi$ be given as follows: $\dom \lambda_\Phi = \{ z \in \sigmastringset \mid z = 
\langle i, w \rangle \text{ and } \langle i, w \rangle$ $ \in \Lambda_\Phi \}$ and 
\[ \lambda_\Phi(\langle i,w \rangle)= B_{\langle i, w \rangle}.\]

\begin{lemma}\label{lemmaAbstractManifoldCompPredicateSpace}
 The triple $\mathbf{Z}_\Phi=(X,\mathcal{B}_\Phi,\lambda_\Phi)$ is a computable predicate space.
\end{lemma}

\begin{proof}
 Clearly, the set $\dom \lambda_\Phi=\Lambda_\Phi$ is a computable set, because $\sigmastringset
\times \domp{\mu^n}$ is computable, so we only need to prove that for $x,y\in X$
\[ \mathcal{P}_x(X)= \mathcal{P}_y(X) \Rightarrow x = y.\]
Assume that $\mathcal{P}_x(X)= \mathcal{P}_y(X)$, so that 
there is a string $\langle j, u\rangle \in
\sigmastringset$ such that $x,y\in B_{\langle j, u \rangle}$. For $z\in \{ x,y \}$, define the set
\[\mathcal{P}_z^j=\{ B_{\langle j, w \rangle} \mid B_{\langle j, w \rangle} \in \mathcal{P}_z(X) 
\}\] 
and using $\varphi_j$ we can define the set
\[ \varphi_j(\mathcal{P}_z^j)= \{ \mu^n (w) \mid \mu^n (w)\subset  \directimage{U_j}{\varphi_j}
\wedge \varphi_j(z)\in \mu^n (w) \} \subset \beta^{n},\]
then using our hypothesis, we can prove that $\mathcal{P}_x^j=\mathcal{P}_y^j$ and clearly this
implies that $\varphi_j(\mathcal{P}_x^j)=\varphi_j(\mathcal{P}_y^j)$. 

Let $\mathbf{Z}^n=(\euclidean{n}, \beta^n, \mu^n)$ be the standard computable predicate space
associated with $\euclidean{n}$ (see Example \ref{exampleCompPREDSPACES}) and $V\in
\mathcal{P}_{\varphi_j(x)}(\euclidean{n})$. $V$ is an open neighborhood of $\varphi_j(x)$ and for a
suitable $u\in \domp{\mu^n}$, we can choose $\mu^n (u) \in \varphi_j(\mathcal{P}_x^j)$ such that
$\varphi_j(x), \varphi_j(y) \in \mu^n (u) \subseteq V \cap \directimage{U_j}{\varphi_j}$, thus $V\in
\mathcal{P}_{\varphi_j(y)}(\euclidean{n})$. This shows
that $\mathcal{P}_{\varphi_j(x)}(\euclidean{n})\subset \mathcal{P}_{\varphi_j(y)}(\euclidean{n})$
and a similar argument can be used to check that the
other inclusion holds. Therefore in the predicate space $\mathbf{Z}^{n}$,
\[\mathcal{P}_{\varphi_j(x)}(\euclidean{n})=\mathcal{P}_{\varphi_j(y)}(\euclidean{n}),\] 
implying that $\varphi_j(x)=\varphi_j(y)$ and as $\varphi_j$ is injective, $x=y$. We have proven
that the triple $\mathbf{Z}_\Phi=(X,\mathcal{B}_\Phi,\lambda_\Phi)$ is a computable predicate space.
\end{proof}

{\it Remark.} 
Since an atlas $\Phi$ on a set $X$ induces a computable predicate space $\mathbf{Z}_\Phi$ on $X$, 
$\Phi$ induces a structure of computable topological space on $X$.
By Lemma \ref{lemmaCompPredicateSpaceISCompTOPSPACE}, the effective space
\[\mathbf{T}_\Phi(X) \defeq T(\mathbf{Z}_\Phi)\]  
is a computable topological space, such that the topology induced by $\mathbf{Z}_\Phi$ is precisely
the topology $\tau_\Phi$ induced by $\Phi$. By Lemma \ref{lemBallsBaseAtlas}, the set $\mathcal{B}_\Phi$ of computable 
balls is not only a subbase, but a base of $\tau_\Phi$, so that we have another effective
topological space $\mathbf{X}_\Phi=(X, \tau_\Phi, \mathcal{B}_\Phi, \lambda_\Phi)$ associated with
$X$. We define the \emph{computable topological space associated to $\Phi$} (and induced on $X$) as
the space  $\mathbf{T}_\Phi(X)$. In general, we cannot use the effective space $\mathbf{X}_\Phi$,
because it could happen that it is not computable. But if it is the case that $\mathbf{X}_\Phi$ is
computable, then by Lemma \ref{lemmaEquivalentCompTOPSPACEBaseSubbase}, $\mathbf{X}_\Phi$ 
is equivalent to 
$\mathbf{T}_\Phi(X)$. If $\delta_{T(\mathbf{Z}_\Phi)} \colon \subseteq \sigmaomegastringset \to X$
is the inner representation of $X$ induced by $\mathbf{T}_\Phi(X)$, then by part 2. of Lemma
\ref{lemmaCompPredicateSpaceISCompTOPSPACE}, $\delta_{T(\mathbf{Z}_\Phi)} \equiv \delta_{\mathbf{Z}_\Phi}$, hence
from the point of view of computability, we can interchange $\delta$ with $\delta_{\mathbf{Z}_\Phi}$. 
Moreover, we
will use the symbol $\delta_\Phi$ to denote any of the representations $\delta_{T(\mathbf{Z}_\Phi)}$ or 
$\delta_{\mathbf{Z}_\Phi}$. In fact, because the computable space $\mathbf{T}_\Phi(X)$ depends on
the atlas $\Phi$, all the elements of $\mathbf{T}_\Phi(X)$ and all notations and
representations induced by this space will be denoted with a ``$\Phi$'' subindex, so that
$\mathbf{T}_\Phi(X)=(X, \tau_\Phi, \beta_\Phi, \nu_\Phi)$, where $\tau_\Phi$ is the topology induced
by $\Phi$ on $X$; $\beta_\Phi$ is the set of all finite intersections of elements of
$\mathcal{B}_\Phi$; $\nu_\Phi$ is the notation of $\beta_\Phi$ induced by $\lambda_\Phi$ and
$\delta_\Phi,\theta_\Phi,\psi_\Phi,\kappa_\Phi,$ etc., are the representations
given in Definitions \ref{defPositiveRepresentationsCompTOPSPACE} and
\ref{defNegativeRepresentationsCompTOPSPACE} respectively, for the computable space
$\mathbf{T}_\Phi(X)$.

\subsection{Computable structures}

We are now ready to formulate our definition of \emph{computable manifold}.
 We start by defining what we call \emph{computable atlas}. Remember that the computable topological
space $\mathbf{R}^n=(\euclidean{n}, \tau^n, \beta^n, \mu^n)$ of Example \ref{exampleCompTOPSPACES}
is computable euclidean space of dimension $n$ and that each standard representation induced by
$\mathbf{R}^n$ is denoted by $\gamma^n$.

\begin{definition}\label{defAbstractComputableAtlas}
An $n$-dimensional \emph{computable atlas} on a set $X$
is a topological atlas (See Definition \ref{defChartAtlas}) $\Phi = \{ (\varphi_i, $ $U_i) \}_{i \in I}$ ($I\subseteq \sigmastringset$) 
such that the following properties are satisfied:
\begin{itemize}
\item [(a)] For $i\in I$, the map $\varphi_i\colon U_i\subseteq X \to \euclidean{n}$ is a
$(\delta_{\mathbf{Z}_\Phi}, \delta^n)$-computable function  and  the inverse 
$\varphi_i^{-1}\colon \directimage{U_i}{\varphi} \to U_i$ is 
$(\delta^n, \delta_{\mathbf{Z}_\Phi})$-computable. 
\item [(b)] Each set $\directimage{U_i}{\varphi_i}\subseteq \euclidean{n}$ is a
$\theta^n$-computable subset of $\euclidean{n}$.
\end{itemize}
\end{definition}

We now present examples of computable atlases.

\begin{example}\label{exampleCompAtlasEuclideanSpace}
 Let $n\geqslant 0$. The identity $1_{\euclidean{n}}\colon \euclidean{n} \to \euclidean{n}$ is a
$(\delta^n, \delta^n)$-computable function which covers $\euclidean{n}$ and this is clearly
a $\theta^n$-computable open set, thus it determines an $n$-dimensional computable atlas $\Phi = \{
(1_{\euclidean{n}},\euclidean{n}) \}$ on $\euclidean{n}$.
\end{example}

\begin{example}\label{exampleUnusualLine}
 The map $\varphi_1\colon \reals \to \reals$ defined by $\varphi_1(x)=x + a$ $(a\in \reals)$ is a 
homeomorphism of the line onto itself. We now prove that $\Psi = \{ (\varphi_1, \reals) \}$ is a 
computable atlas on $\reals$. Clearly, $\reals$ is a $\theta^1$-computable open set in $\reals$, 
now we need to show that $\varphi_1$ and its inverse are computable with respect to $\delta_\Psi$ 
and $\delta^1$. $\varphi_1$ is $(\delta_\Psi, \delta^1)$-computable, because given $x \in \reals$ 
and $p\in \dom \delta_\Psi$ such that $\delta_\Psi(p)=x$, we have that 
\[ \delta_\Psi(p)=x \Longleftrightarrow (\forall w \in \sigmastringset) (w \ll p \Leftrightarrow w 
\in \dom \lambda_\Psi \text{ and } x \in \lambda_\Psi(w)), \]
So, for each $w \ll p, x \in \lambda_\Psi(w)=\inverseimage{\mu^1(z)}{\varphi_1}$ with $w= 
\langle 1, z \rangle$. Therefore, $\varphi_1(x) \in \mu^1(z)$ for all $w \ll p$. Let $\mathfrak{M}$ 
be a TTE machine with the following program. On the input $p \in \sigmaomegastringset$
\begin{enumerate}
 \item For each $w = \langle 1, z \rangle \ll p$
\begin{itemize}
 \item [1.1] output $\iota(z)$
\end{itemize}
\end{enumerate}
By the previous argument, the machine $\mathfrak{M}$ outputs a string $q\in \sigmaomegastringset$ 
such that $(\forall z \in \sigmastringset) (z \ll q \Leftrightarrow z \in \domp{\mu^1} \text{ and 
} \varphi_1(x) \in \mu^1(z))$. By definition \ref{defPositiveRepresentationsCompTOPSPACE}, 
$\delta^1(q)=\varphi_1(x)$, thus $\mathfrak{M}$ computes a function which realizes $\varphi_1$ with 
respect to $\delta_\Psi$ and $\delta^1$. To show that $\varphi_1^{-1}$ is $(\delta^1, 
\delta_\Psi)$-computable, the argument is very similar, we omit it. We have shown that $\Psi$ is a 
1-dimensional computable atlas on $\reals$.
\end{example}

Notice that the computability of the atlas $\Psi$ is independent of the computability of the real 
number $a$ (with respect to $\delta^1$), given in the definition of $\varphi_1$. We will come back 
to this example later.

\begin{example}\label{exampleCompAtlasCircle}
We construct a computable atlas for the 1-sphere $\sphere{1}\subset \euclidean{2}$. Let $U_+,U_-$
be defined as
\[ U_+ = \{ (x,y) \in \sphere{1} \mid y > 0 \}\quad\text{ and } \quad U_- = \{ (x,y) \in \sphere{1}
\mid y < 0 \}, \]
and let $f_+ \colon U_+ \to \reals, f_- \colon U_- \to \reals$ be given by $f_+(x,y)= f_-(x,y) =
x$, these two functions are injections of $U_+,U_-$ onto $(-1,1)\subset \reals$. Now let
\[ V_+ = \{ (x,y) \in \sphere{1} \mid x > 0 \}\quad\text{ and } \quad V_- = \{ (x,y) \in \sphere{1}
\mid x < 0 \}, \]
and define $g_+\colon V_+ \to \reals,g_-\colon V_- \to \reals$ as $g_+(x,y)= g_-(x,y) = y$. The set
\[\Phi = \{ (f_+, U_+),(f_-, U_-),(g_+, V_+),(g_-, V_-) \}\] 
is an 1-dimensional atlas on $\sphere{1}$. To check that $\Phi$ is computable, we need to show that 
each chart $\varphi\colon U \to \reals$ in 
$\Phi$ is a $(\delta_{\mathbf{Z}_\Phi}, \delta^1)$-computable function with  $(\delta^1,
\delta_{\mathbf{Z}_\Phi})$-computable inverse and the sets $\directimage{U}{\varphi}$ are 
$\theta^1$-computable open subsets of $\reals$. Each set $\directimage{U}{\varphi}$ is clearly a 
$\theta^1$-computable open subset of $\reals$, because we have that $\directimage{U}{\varphi}=(-1, 
1)$ and the latter set is $\theta^1$-computable. Now we give the full proof of the computability of 
$f_+$ and its inverse, the other cases $f_-, g_+$ and $g_-$ are very similar.

Let $f_1=f_+,f_2=f_-,f_3=g_+$ and $f_4=g_-$ To see that $f_1$ is $(\delta_{\mathbf{Z}_\Phi}, 
\delta^1)$-computable, consider the following Type-2 machine $\mathfrak{F}_1$. On input $p\in \dom 
\delta_{\mathbf{Z}_\Phi}$ ($\delta_{\mathbf{Z}_\Phi}(p)=(x,y)\in \sphere{1}$):
\begin{enumerate}
 \item For each $z = \langle i, w \rangle \ll p$,
\begin{itemize}
 \item [1.1] if $i=1$, then output $\iota(w)$;
\item [1.2] Otherwise, execute the following:
\begin{itemize}
 \item [1.2.1] Compute $a_y,b_y$ such that $y\in (a_y,b_y)\subset (-1,1)$;
 \item [1.2.2] compute $a_x=\sqrt{1 - a_y^2}$;
 \item [1.2.3] compute $b_x=\sqrt{1 - b_y^2}$;
 \item [1.2.4] compute $w^\prime$ such that $(b_x, a_x)\subseteq \mu^1(w^\prime)\subset (-1,1)$;
 \item [1.2.5] output $\iota(w^\prime)$.
\end{itemize}
\end{itemize}
\end{enumerate}
We claim that $\mathfrak{F}_1$ computes a function $F_1\colon \subseteq \sigmaomegastringset \to 
\sigmaomegastringset$ which realizes $f_1$. We first check that each step of $\mathfrak{F}_1$ can 
be done in finite time. Step 1.1 is clearly computable, and for step 1.2, we only need to check 
that steps 1.2.1-1.2.5 can be calculated in finite time by $\mathfrak{F}_1$. First of all, if 
$\mathfrak{F}_1$ is executing step 1.2, then we have that $i\neq 1$, so that the string $z=\langle 
i, w \rangle$ represents a computable ball $\inverseimage{\mu^1(w)}{f_i}$ such that 
$f_i \neq f_1$. Remember that $f_1$ is defined on the set $U_+=\{ (x,y) \in \sphere{1} \mid y 
> 0 \}$, so that if $f_i\neq f_1=f_+$, then $f_i=f_3$ or $f_i=f_4$. Without loss 
of generality, assume that $f_i=f_3$.

Step 1.2.1 can be computed in finite time, because since $(x,y)=\delta_{\mathbf{Z}_\Phi}(p) \in 
\inverseimage{\mu^1(w)}{f_3}$, then $f_3(x,y)=g_+(x,y)=y\in \mu^1(w)\subset (-1, 1)$. Using the 
string $w\in \domp{\mu^1}$, $\mathfrak{F}_1$ can compute $a_y,b_y\in \rationals$ with $y\in 
(a_y,b_y)\subset \mu^1(w)$, thus step 1.2.1, can be done in finite time by $\mathfrak{F}_1$. Steps 
1.2.2 and 1.2.3 are computable, because $a_y,b_y$ are rationals and 
the square root is a computable function. 
Step 1.2.4 can be 
calculated by $\mathfrak{F}_1$, because by the previous steps, $a_x$ and $b_x$ are computable point 
in $\reals$ and the given set inclusions can be tested by using Theorem 
\ref{thmTarskiDecisionMethod}. Step 1.2.5 is clearly computable.

Now we prove the correctness of $\mathfrak{F}_1$. Let $(x,y)=\delta_{\mathbf{Z}_\Phi}(p) \in \dom 
f_1$. and take $\langle i, w \rangle \ll p$, so that $(x,y)\in \inverseimage{\mu^1(w)}{f_i}$. If 
$i=1$ then $f_1(x,y) \in \mu^1(w)$ and in this case, $\mathfrak{F}_1$ outputs $\iota(w)$ in the 
output tape. When $i\neq 1$, $i=3$ or $i=4$, thus $f_i(x,y)=y \in \mu^1(w)$. There exist 
$a_y,b_y\in \rationals$ such that $y\in \mu^1(w)=(a_y,b_y)$. Since $f_i$ is bijective, 
$a_x=f_i^{-1}(a_y)$ and $b_x=f_i^{-1}(b_y)$ are defined and they satisfy the equations
 \[ a_x=\sqrt{1 - a_y^2} \quad\text{and}\quad b_x=\sqrt{1 - b_y^2}; \]
(notice that $b_x < a_x$, because $a_y < b_y$) and it is immediate to show that $b_x < f_1(x,y) < 
a_x$. There exist rational numbers $a,b$ such that  $f_1(x,y)\in (b_x, a_x)\subseteq (a,b)\subset 
(-1,1)$ and the set $(a,b)$ is represented by a string $w^\prime\in \domp{\mu^1}$, this string is 
computed by $\mathfrak{F}_1$ in step 1.2.5, thus the output $\iota(w^\prime)$ is correct. Therefore 
the machine $\mathfrak{F}_1$ computes a function $F$ such that on the input $p\in \dom 
\delta_{\mathbf{Z}_\Phi}$, $F(p)$ satisfies
\[ (\forall w^\prime \in \sigmastringset) (w^\prime \ll F(p) \Leftrightarrow w^\prime \in 
\domp{\mu^1} \text{ and } f_1(x,y) \in \mu^1(w^\prime)).\]
Therefore $\delta^1(F(p))=f_1(x,y)$ and $F$ realizes $f_1$ with respect to 
$\delta_{\mathbf{Z}_\Phi}$ and $\delta^1$, so that $f_1$ is $(\delta_{\mathbf{Z}_\Phi}, 
\delta^1)$-computable.

It only remains to prove that $f_1^{-1}$ is $(\delta^1,\delta_{\mathbf{Z}_\Phi})$-computable. There 
exists a Type-2 machine $\mathfrak{G}_1$ that, on input $q\in \dom \delta^1$, does the following: 
For each $z \ll q$, $\mathfrak{G}_1$ checks if $\mu^1(z) \subseteq (-1,1)$, if so, then it writes 
$\iota(\langle 1, z \rangle)$ on the output tape; otherwise $\mathfrak{G}_1$ ignores the string 
$z$. It is easy to see that $\mathfrak{G}_1$ computes a function $G\colon \subseteq 
\sigmaomegastringset \to \sigmaomegastringset$ which realizes $f_1^{-1}$ with respect to 
$\delta^1$ and $\delta_{\mathbf{Z}_\Phi}$.

Therefore, the atlas $\Phi$ fulfills Definition \ref{defAbstractComputableAtlas} so that it is a 
$1$-dimensional computable atlas for the circle.
\end{example}

We now show that with a computable atlas, the transition functions satisfy the expected
computability properties inside the induced computable topological space $\mathbf{T}_\Phi(X)$.

\begin{lemma}\label{lemmaCompPropertiesAtlas}
 Let $\Phi= \{ (\varphi_i, U_i) \}_{i\in I}$ be an $n$-dimensional computable atlas on 
$X$. 
Then for the computable spaces $\mathbf{T}_\Phi(X)$ and $\mathbf{R}^n$:
\begin{itemize}
 \item [(i)] For all $i\in I$, the chart $\varphi_i\colon U_i \to \directimage{U_i}{\varphi_i}$ is 
a computable
homeomorphism and $U_i$ is a $\theta_\Phi$-computable open set in $X$.
\item [(ii)] For each $i,j\in I$, $U_i\cap U_j$ is $\theta_\Phi$-computable open in $X$ and
$\varphi_i(U_i\cap U_j)$ is $\theta^n$-computable open in $\euclidean{n}$.
\item [(iii)] Each transition function $\varphi_i \varphi_j^{-1}$ is a computable homeomorphism
between $\theta^n$-computable subsets of $\euclidean{n}$.
\end{itemize}
\end{lemma}

\begin{proof}
 (i) That each $\varphi_i$ $(\varphi_i^{-1})$ is computable is true because $\varphi_i$
$(\varphi_i^{-1})$ is $(\delta_{\mathbf{Z}_\Phi}, \delta^n)$-computable ($(\delta^n,
\delta_{\mathbf{Z}_\Phi})$-computable) and by Lemma \ref{lemmaCompPredicateSpaceISCompTOPSPACE},
$\delta_{\mathbf{Z}_\Phi} \equiv \delta_{\Phi}$. $U_i$ is $\theta_\Phi$-computable
because since $\varphi_i$ is computable, the map $W \mapsto \inverseimage{W}{\varphi_i}$ is
$(\theta^n,\theta_\Phi)$-computable (Theorem \ref{thmequivCompFunctions}); (ii) Follows by
combining (i) and part 1. of Theorem \ref{thmTheorem11ECT} with Theorem 
\ref{thmequivCompFunctions}; (iii) is immediate because composition of computable functions
between computable topological spaces is again,
computable.
\end{proof}

We continue with more examples of computable atlases.

\begin{example}\label{exampleCompAtlasEuclideanSpaceCicleWithTOP}
The topology induced by the computable atlas
$\{ (1_{\euclidean{n}},\euclidean{n}) \}$ on $\euclidean{n}$ of Example
\ref{exampleCompAtlasEuclideanSpace} is of course, the usual euclidean topology and clearly
$\mathbf{T}_\Phi(\euclidean{n})$ is equivalent to $\mathbf{R}^n$.
Consider now the atlas $\Phi=\{ (f_+, U_+),$ $(f_-, U_-),(g_+, V_+),(g_-, V_-) \}$ on $\sphere{1}$ of
Example \ref{exampleCompAtlasCircle}. We have proven that it is computable. But which is the
topology that $\Phi$ induces on $\sphere{1}$ ? Notice that for suitable $a,b,c,d,e,f\in \rationals$
and $\varphi,\psi,\alpha \in \{ f_+,f_-,g_+,g_- \}$ the condition
\[ \inverseimage{(a, b)}{\varphi} \subseteq \inverseimage{(c, d)}{\psi} \cap
\inverseimage{(e, f)}{\alpha}, \]
can be verified algorithmically. For example, $\inverseimage{(a, b)}{g_-} \subseteq
\inverseimage{(c, d)}{f_+} \cap \inverseimage{(e, f)}{g_+}$ is equivalent to the condition
\begin{multline*}
 (\forall (x,y) \in \euclidean{2}) x^2+y^2=1 \wedge x < 0 \wedge y \in (a,b) \wedge
(a,b)\subset (-1,1) \Rightarrow y > 0 \wedge x \in (c,d) \\ 
\wedge (c,d)\subset (-1,1) \wedge x> 0 \wedge y \in (e,f) \wedge (e,f)\subset (-1,1)
\end{multline*}
and by Theorem \ref{thmTarskiDecisionMethod}, all expressions of this kind can always be checked in 
finite time by a Turing machine\footnote{Whenever we need to apply the Tarski's decision method 
\cite{Tarski51}, we
assume that we have an appropriate encoding of polynomial functions with rational coefficients as
strings of $\sigmastringset$.}. Hence, the effective space $\mathbf{S}^1_\Phi = (\sphere{1},
\tau_\Phi, \mathcal{B}_\Phi, \lambda_\Phi)$ becomes a computable topological space, equivalent to
$\mathbf{T}_\Phi(\sphere{1})$. It is not hard to show that $\mathbf{S}^1_\Phi$ is equivalent to the
computable subspace $\mathbf{S}^1$ of $\mathbf{R}^2$, therefore the induced computable topology on
$\sphere{1}$ by $\Phi$ is the computable subspace topology.
\end{example}

\begin{example}\label{exampleCompAtlasnSphereSP}
An atlas can be defined for the sphere $\sphere{n} \subset \euclidean{n+1}$
with the stereographic projection. Let $P_{1}=(0,\ldots, 0, 1), P_{-1}=(0,\ldots,0, -1)$ and define
$U_r = \sphere{n} - \{ P_r \}$ with $r=1,-1$. $U_1 \cup U_{-1}=\sphere{n}$ and if we define $s_1 
\colon U_1 \to \euclidean{n}$ by $s_1(x,t)=\frac{x}{1-t}$ and $s_{-1} \colon U_{-1} \to 
\euclidean{n}$, $s_{-1}(x,t)=\frac{x}{1+t}$ with $x=(x_1, \ldots, x_n)$, then $\Psi = \{ (s_1, 
U_1), 
(s_{-1}, U_{-1}) \}$ is a topological atlas for $\sphere{n}$. We now prove that $\Psi$ is a 
computable atlas for $\sphere{n}$ as follows: The atlas $\Psi$ induces the effective 
spaces $\mathbf{S}^n_\Psi = (\sphere{n}, \tau_\Psi, \mathcal{B}_\Psi,
\lambda_\Psi)$ and $\mathbf{T}_\Psi(\sphere{n})$, being the latter computable. Now, since the maps 
$s_1,s_{-1}$ are rational functions with coefficients in $\rationals$, we can apply Theorem 
\ref{thmTarskiDecisionMethod} to deduce that the decision problem
\[ \inverseimage{B_1}{s_r} \subseteq \inverseimage{B_2}{s_t} \cap \inverseimage{B_3}{s_u},\quad 
B_i \in \beta^n, r,t,u\in \{ 1, -1 \}, \]
is computable, thus $\mathbf{S}^n_\Psi = (\sphere{n}, \tau_\Psi, \mathcal{B}_\Psi, \lambda_\Psi)$ 
is 
a computable topological space and by Lemma \ref{lemmaEquivalentCompTOPSPACEBaseSubbase}, it is 
equivalent to $\mathbf{T}_\Psi(\sphere{n})$. Our next step is to show that $\mathbf{S}^n_\Psi$ and 
$\mathbf{S}^n=\ctopcompsubspace{\mathbf{R}^{n+1}}{\sphere{n}}$ are equivalent computable spaces. 
This can be done easily using Definition \ref{defEquivCompTOPSPACES} and Theorem 
\ref{thmTarskiDecisionMethod}
. If 
$\mu^{n+1}_{\sphere{n}}$ is the notation for base elements of $\mathbf{S}^n$ and $\theta_\Psi$ is 
the representation for open sets of $\mathbf{S}^n_\Psi$, then to prove that $\mu^{n+1}_{\sphere{n}} 
\leq \theta_\Psi$, we have that $\inverseimage{\mu^n(z)}{s_r}\subseteq \mu^{n+1}_{\sphere{n}}(w)$ 
($z\in \domp{\mu^n}$ and $w \in \domp{\mu^{n+1}}$) is equivalent to the expression
\begin{equation}\label{eq1exampleCompAtlasnSphereSP}
 (\forall y \in \euclidean{n+1}) (y\in \sphere{n} \text{ and } s_r(y)\in \mu^n(z) 
\Rightarrow y \in \sphere{n} \cap \mu^{n+1}(w))
\end{equation}
and this expression can be easily translated into an elementary expression. With all this data, a 
Type-2 machine can be constructed such that, on input $w\in \domp{\mu^{n+1}_{\sphere{n}}}$, 
enumerates all pairs $(r, z)$ ($r\in \{ -1, 1 \}, z\in \domp{\mu^n}$) and tests if $w,r$ and $z$ 
satisfy \eqref{eq1exampleCompAtlasnSphereSP}, if this is the case, then the machine outputs 
$\iota(z)$. This machine computes a function which translates 
$\mu^{n+1}_{\sphere{n}}$-names into $\theta_\Psi$-names, that is, $\mu^{n+1}_{\sphere{n}} \leq 
\theta_\Psi$. To prove that $\lambda_\Psi \leq \theta^{n+1}_{\sphere{n}}$, the argument is almost 
the same. Thus $\mathbf{S}^n_\Psi$ and $\mathbf{S}^n$ are equivalent, therefore 
$\mathbf{T}_\Psi(\sphere{n})$ and $\mathbf{S}^n$ are equivalent computable spaces.

Now to show that the atlas $\Psi$ is computable, the argument is the following: Since  
$\mathbf{T}_\Psi(\sphere{n})$ and $\mathbf{S}^n$ are equivalent, $\delta_\Psi \equiv 
\delta^{n+1}_{\sphere{n}}$, so that to check that $s_1,s_{-1}$ and their inverses are computable 
with respect to $\delta_\Psi$ and $\delta^n$, it is enough to show that they are computable with 
respect to $\delta^{n+1}_{\sphere{n}}$ and $\delta^n$. But $\delta^{n+1}_{\sphere{n}}$ is simply 
the representation $\delta^{n+1}$ of $\euclidean{n+1}$, restricted to $\sphere{n}$ (see Definition 
\ref{defMultifuncRestrictionRange}). In other words, we only need to show that the charts are 
computable 
with respect to $\delta^{n+1}$ and $\delta^{n}$. But the maps $s_1, s_{-1}$ and their inverses are 
defined in terms of sums, multiplications and square roots, thus they are easily seen to be  
computable.

To finish the proof, we need to show that the sets $\directimage{U_i}{s_i}$ are 
$\theta^n$-computable in $\euclidean{n}$. But this is immediate, because 
$\directimage{U_i}{s_i}=\euclidean{n}$. Therefore $\Psi$ is a computable atlas for 
$\sphere{n}$.
\end{example}

\begin{example}\label{exampleCompAtlasRealProjectiveSpace}
 Let $\realprojectivespace{n}$ be $n$-dimensional real projective space, the set of all
$1$-dimensional vector subspaces of $\euclidean{n+1}$. Each subspace is spanned by a non-zero
vector $v=(x_1,\ldots,x_{n+1})$. In other words
\[ \realprojectivespace{n}= (\euclidean{n+1} - \{ 0 \}) / \sim, \]
where $x\sim y \Leftrightarrow (\exists \alpha \in \reals - \{ 0 \})(x=\alpha y)$. We define for
$i=1,\ldots,n+1$ the set $U_i=\{ \left[ v \right] \in \realprojectivespace{n} \mid x_i \neq 0 \}$.
Clearly $X$ is covered by $U_1,\ldots,U_{n+1}$. Let $\varphi_i \colon U_i \to \euclidean{n}$ and 
$\varphi_i^{-1} \colon \euclidean{n} \to U_i$ be defined by 
\begin{itemize}
\item [] $\varphi_i (\left[ (x_1,\ldots,x_{n+1}) \right]) = \bigl(\frac{x_1}{x_i},\ldots,\frac{x_{i-1}}{x_i},\frac{x_{i+1}}{x_i},\ldots,\frac{x_{n+1}}{x_i}\bigr)$;
\item [] $\varphi_i^{-1} \bigl(y_1,\ldots,y_n\bigr) = \left[ (y_1,\ldots,y_{i-1},1,y_{i+1},\ldots,y_{n+1}) \right]$.
\end{itemize}
This is therefore a coordinate chart $(\varphi_i, U_i)$ and the set of all these charts is an atlas $\Gamma$ on
$\realprojectivespace{n}$. It is easy to show that each function $\varphi_i$,
$\varphi_i^{-1}$ is computable with respect to $\delta_\Gamma$ and $\delta^n$. We conclude that
$\Gamma$ is a computable atlas on $\realprojectivespace{n}$.

The set $\mathcal{B}_\Gamma =\{ \inverseimage{\mu^n(w)}{\varphi_i} \mid i=1,\ldots,n+1; w \in
\domp{\mu^n} \}$ is a base for the topology induced by $\Gamma$ and also the property 
\[\inverseimage{\mu^n(w_1)}{\varphi_i} \subseteq \inverseimage{\mu^n(w_2)}{\varphi_j} \cap
\inverseimage{\mu^n(w_3)}{\varphi_k}\] 
is equivalent to 
\[(\forall x\in \euclidean{n+1}) ( x_i \neq 0 \wedge p_i(x)\in \mu^n(w_1) \Rightarrow x_j \neq 0
\wedge p_j(x)\in \mu^n(w_2) \wedge x_k \neq 0 \wedge p_k(x)\in \mu^n(w_3)), \]
and using Theorem \ref{thmTarskiDecisionMethod}, the latter expression can be checked 
algorithmically. The effective space
$\mathbf{RP}^n_\Gamma = (\realprojectivespace{n}, \tau_\Gamma, \mathcal{B}_\Gamma,$
$\lambda_\Gamma)$ is a computable topological space, which is equivalent to the computable space
$\mathbf{T}_\Gamma(\realprojectivespace{n})$ induced  by $\Gamma$.
\end{example}

All the previous examples are cases in which we can replace the canonical computable space 
$\mathbf{T}_\Phi(X)$ with the somewhat simpler effective space $\mathbf{X}_\Phi$, but according
to Lemma \ref{lemmaEquivalentCompTOPSPACEBaseSubbase}, this can be done only when the latter is
computable. 

As with standard topological manifolds, it can happen that a set $X$ has more that one computable
atlas defined on it. We would like to consider two computable atlases that define the same
computable topology as equivalent.

\begin{definition}\label{defComputablyCompatibleAtlases}
 Two $n$-dimensional computable atlases $\Phi, \Psi$ on
$X$ are \emph{computably compatible} if and only if $\Phi$ and $\Psi$ are compatible (Definition
\ref{defCompatibleAtlases}) and $\delta_\Phi\equiv \delta_\Psi$.
\end{definition}

Computable compatibility of computable atlases is an equivalence relation on the set of all
computable atlases on $X$.

\begin{lemma}\label{lemmaCompatibleCompAtlasesIsSameCompTOP}
  Let $\Phi, \Psi$ be two computable atlases on $X$. 
Then $\Phi$ and $\Psi$ are computably compatible if and only if $\mathbf{T}_\Phi(X)$ and 
$\mathbf{T}_\Psi(X)$
are equivalent computable topological spaces.
\end{lemma}

\begin{proof}
  $(\Rightarrow)$ That $\Phi$ and $\Psi$ are computably compatible means that the following holds: 
\begin{itemize}
 \item [(1)] $\Phi$ and $\Psi$
are topologically compatible;
\item [(2)] $\delta_\Phi\equiv \delta_\Psi$.
\end{itemize}
 By (1) we have that $\tau= \tau_\Phi = \tau_\Psi$, thus $\mathbf{T}_\Phi(X) = (X, \tau, 
\beta_\Phi, 
\nu_\Phi)$ and $\mathbf{T}_\Psi(X) = (X, \tau, \beta_\Psi, \nu_\Psi)$. Combining (2) with part 2. 
of 
Lemma \ref{lemmaCompPredicateSpaceISCompTOPSPACE} and Theorem \ref{thmRobustnessEquivTOPSPACES}, 
$\mathbf{T}_\Phi(X)$ and $\mathbf{T}_\Psi(X)$ are equivalent. 

$(\Leftarrow)$ If $\mathbf{T}_\Phi(X)$ and $\mathbf{T}_\Psi(X)$ are equivalent computable 
topological spaces, then $\tau_\Phi=\tau_\Psi$, thus $\Phi$ and $\Psi$ induce the same topology on 
$X$, so that 
$\Phi$ and $\Psi$ are topologically 
compatible. Also, by Theorem \ref{thmRobustnessEquivTOPSPACES}, $\delta_\Phi\equiv \delta_\Psi$. 
Therefore, by Definition \ref{defComputablyCompatibleAtlases}, $\Phi$ and $\Psi$ are computably 
compatible.
\end{proof}

We now show two 
computable atlases which are compatible, but not computably compatible.

\begin{example}\label{exampleLineAtlasesNotCC}
 Let $\Phi = \{ (1_\reals, \reals) \}$ and $\Psi=\{ (\varphi_1, \reals) \}$, where $\varphi_1\colon 
\reals \to \reals$ is defined as $\varphi_1(x)=x+a$. $\Psi$ is the atlas of Example 
\ref{exampleUnusualLine}, where we saw that the computability of $\Psi$ does not depend on the 
computability of the number $a$ with respect to $\delta^1$. We now prove that if $a$ is not 
$\delta^1$-computable, then $\Phi$ and $\Psi$ are not computably compatible (notice that the two 
atlases are topologically compatible).

Suppose then that $a$ is not computable with respect to $\delta^1$ and that $\Phi$ and $\Psi$ are 
computably compatible. By Definition \ref{defComputablyCompatibleAtlases}, $\delta_\Phi \equiv 
\delta_\Psi$ and since $\Phi$ is the atlas induced by the identity on $\reals$, we have that 
$\delta_\Phi \equiv \delta^1$, hence $\delta_\Psi \equiv \delta^1$. The atlas $\Psi$ is 
computable, so that $\varphi_1$ is $(\delta_\Psi, \delta^1)$-computable. Since $\delta_\Psi 
\equiv \delta^1$, we can conclude that $\varphi_1$ is $(\delta^1, \delta^1)$-computable. But 
$\varphi_1$ is computable with respect to $\delta^1$ if and only if $a$ is a computable real 
number. Therefore $\Phi$ and $\Psi$ cannot be computably compatible.
\end{example}

Example \ref{exampleLineAtlasesNotCC} give us a desirable consequence of our definitions. If 
the map $\varphi_1$ is not computable in $\reals$ with the usual manifold structure, then we do not 
want the atlas $\Psi$ to be compatible with $\Phi$.

\begin{definition}\label{defComputableStructure}
 A \emph{computable structure} on 
$X$ 
is a equivalence class $\left[ \Phi \right]$ of computable atlases on $X$.
\end{definition}

\begin{definition}[Computable manifold]\label{defAbstractComputableManifold}
An $n$-dimensional \emph{computable manifold} is a 
set $M$ 
together
with a computable structure $\left[ \Phi \right]$.
\end{definition}

Thus, a computable $n$-manifold is a pair $(M, \left[ \Phi \right])$. As each computable atlas
determines a unique computable structure, we can also simply write $(M, \Phi)$ and forget 
about the brackets and sometimes, if no confusion arises, we will omit the explicit reference to 
the 
computable atlas $\Phi$. The integer $n$, the dimension of the manifold, is denoted in the usual
form as $\dim M$. All the previous examples about sets with computable atlases are actually
computable manifolds, we now give more examples.

\begin{example}\label{exampleZeroCompManifold}
  A computable 0-manifold is just a discrete computable topological space. 
\end{example}

\begin{lemma}\label{lemmaCompManifoldsProduct}
 Let $(M_1, \Phi_1),(M_2, \Phi_2)$ be computable manifolds such that $\dim M = n$ and $\dim N = m$.
Then there exists a $(n+m)$-dimensional computable atlas $\overline{\Phi}$ for the set $M_1\times
M_2$ with the following properties:
\begin{itemize}
\item  [(a)] $\lambda_{\overline{\Phi}} = \lambda_{\Phi_1} \times \lambda_{\Phi_2}$;
 \item [(b)] $\delta_{\overline{\Phi}} \equiv \left[ \delta_{\Phi_1}, \delta_{\Phi_2} \right]$;
 \item [(c)] $\mathbf{T}_{\overline{\Phi}}(M_1 \times M_2)=\overline{\mathbf{T}}$, where
$\overline{\mathbf{T}}$ is the computable topological space of Lemma
\ref{lemmaPropertiesEffectiveTOPPRODUCT} induced by the computable spaces $\mathbf{T}_{\Phi_1}(M_1)$
and $\mathbf{T}_{\Phi_2}(M_2)$.
\end{itemize}
\end{lemma}

\begin{proof}
Let $\overline{\mathbf{X}} = (\euclidean{n}\times \euclidean{m}, \overline{\tau}, \overline{\beta}, 
\overline{\nu})$ be the computable product space induced by $\mathbf{R}^{n}$ and 
$\mathbf{R}^{m}$. Then $\overline{\mathbf{X}}$ is computably homeomorphic to $\mathbf{R}^{n+m}$, 
using the canonical map $((x_1, \ldots, x_n),(y_1,\ldots,y_m)) \stackrel{g}{\mapsto} (x_1, 
\ldots,x_n,$ $y_1,\ldots,y_m)$. The set $M_1\times M_2$ endowed with the atlas
\[ \overline{\Phi} = \{ (g\circ (\varphi \times \psi), U \times V) \mid (\varphi, U) \in \Phi_1,
(\psi, V) \in \Phi_2 \} \]
is a topological manifold of dimension $n+m$, such that the topology induced in $M_1\times M_2$ by
$\overline{\Phi}$ is the product topology of the spaces $M_1$ and $M_2$.

Now we prove (a)-(c). Notice that as $\overline{\mathbf{X}}$ and $\mathbf{R}^{n+m}$ are computably 
homeomorphic, in the space $\mathbf{R}^{n+m}$, we can replace the base $\beta^{n+m}$ by the base
\[ \overline{\beta}_g = \{ \directimage{\mu^n(w) \times \mu^m(z)}{g} \mid (w,z) \in 
\domp{\mu^n} \times \domp{\mu^m} \}, \] 
(see Lemma \ref{lemmaTopSpaceCompFuncRepSet} and Corollary \ref{corCompHomeoTopSpaces}), thus the 
elements of the predicate
space $\mathbf{Z}_{\overline{\Phi}} = (M_1\times M_2,
\mathcal{B}_{\overline{\Phi}}, \lambda_{\overline{\Phi}})$ are defined as follows:
\begin{itemize}
 \item [] $\mathcal{B}_{\overline{\Phi}} = \{ \inverseimage{\mu^n(w) \times \mu^m(z) }{(\varphi_i
\times \psi_j)} \mid (\langle i,w \rangle, \langle j,z \rangle) \in \dom \nu_{\Phi_1} \times \dom
\nu_{\Phi_2} \}$;
 \item [] $\lambda_{\overline{\Phi}} \colon \sigmastringset \to \mathcal{B}_{\overline{\Phi}}$, is given by 
$\lambda_{\overline{\Phi}}(\langle \langle i,j \rangle, \langle w,z \rangle
\rangle)=\inverseimage{\mu^n(w) \times \mu^m(z) }{(\varphi_i \times \psi_j)}$.
\end{itemize}
Hence, it follows that $\lambda_{\overline{\Phi}}(\langle \langle i,j \rangle, \langle w,z \rangle
\rangle)=\inverseimage{\directimage{\mu^n(w) \times \mu^m(z)}{g}}{(g\circ (\varphi_i \times
\psi_j))}=\directimage{\directimage{\mu^n(w) \times \mu^m(z)}{g}}{(\varphi_i^{-1} \times
\psi_j^{-1}) \circ g^{-1}}=\directimage{\mu^n(w) \times \mu^m(z)}{(\varphi_i^{-1} \times
\psi_j^{-1})}=$ $\inverseimage{\mu^n(w)}{\varphi_i} \times 
\inverseimage{\mu^m(z)}{\psi_j}=\lambda_{\Phi_1}(\langle i, w \rangle) \times
\lambda_{\Phi_2}(\langle j,z \rangle)=(\lambda_{\Phi_1} \times \lambda_{\Phi_2}) (\langle i, w
\rangle, \langle j,z \rangle)$ and this implies that $\delta_{\overline{\Phi}}
\equiv \left[ \delta_{\Phi_1}, \delta_{\Phi_2} \right]$; this proves (a) and (b). It is easy to 
verify that the computable topological space $\mathbf{T}_{\overline{\Phi}}(M_1\times M_2)=(M_1 
\times M_2,
\tau_{\overline{\Phi}}, \beta_{\overline{\Phi}}, \nu_{\overline{\Phi}})$ induced in $M_1\times M_2$
by $\overline{\Phi}$ is precisely the computable product space of $\mathbf{T}_{\Phi_1}(M_1)$ and
$\mathbf{T}_{\Phi_2}(M_2)$, so that (c) holds.
  
To finish the proof of the Lemma, we only need to show that $\overline{\Phi}$ is a computable atlas
on $M_1 \times M_2$, i.e. we have to prove that  each chart $g\circ (\varphi_i \times \psi_j) \colon
U_i\times V_j \to \euclidean{n+m} \in \overline{\Phi}$ is $(\delta_{\overline{\Phi}},
\delta^{n+m})$-computable with $(\delta^{n+m},\delta_{\overline{\Phi}})$-computable inverse and
that each set $\directimage{U_i}{\varphi_i}\times \directimage{V_j}{\psi_j}$ is
$\theta^{n+m}$-computable. To do this, we will use the computable space $\overline{\mathbf{X}}$ 
and the representations $\overline{\delta},\overline{\theta}$ of points and open sets of 
$\euclidean{n}\times \euclidean{m}$ induced by $\overline{\mathbf{X}}$. By (b), 
$\delta_{\overline{\Phi}} \equiv \left[ \delta_{\Phi_1}, \delta_{\Phi_2} \right]$ and by
part 2. of Lemma \ref{lemmaPropertiesEffectiveTOPPRODUCT}, $\varphi_i \times \psi_j$ is
$(\delta_{\overline{\Phi}}, \overline{\delta})$-computable with 
$(\overline{\delta},\delta_{\overline{\Phi}})$-computable inverse. Using part 3. of Lemma 
\ref{lemmaPropertiesEffectiveTOPPRODUCT}, the set $\directimage{U_i}{\varphi_i}\times 
\directimage{V_j}{\psi_j}$ is a $\overline{\theta}$-computable open subset of $\euclidean{n} \times 
\euclidean{m}$. Combining these facts with the computable homeomorphism $g$ between 
$\overline{\mathbf{X}}$ and the computable euclidean space $\mathbf{R}^{n+m}$, we deduce that 
$\overline{\Phi}$ is a computable atlas on $M_1\times M_2$.
\end{proof}

We can generalize this result to arbitrary finite products of computable manifolds
$M_1 \times \cdots \times M_r$.

\begin{example}\label{exampleNTORUS}
 Following Lemma \ref{lemmaCompManifoldsProduct}, we can prove that the \emph{$n$-dimensional
torus} 
\[T^n=\prod_{i=1}^n \sphere{1}\] 
is a computable $n$-manifold, where we equip $\sphere{1}$ with any of the atlases given in Examples
\ref{exampleCompAtlasEuclideanSpaceCicleWithTOP} and \ref{exampleCompAtlasnSphereSP} (they are
computably compatible).
\end{example}

\subsection{Functions between computable manifolds}

In this section, we will show that (as expected)
computable functions, as we have defined them in Section \ref{subsecCompFunctionsCompTOPSPACES}, are
just adequate to be used as morphisms between computable manifolds.

To define morphisms between two computable manifolds $(M,\Phi)$ and $(N,\Psi)$, we have two possible
choices. The first uses the induced computable topological spaces
$\mathbf{T}_\Phi(M)$ and $\mathbf{T}_\Psi(N)$, a morphism between $M$ and $N$ is just a computable 
continuous function between $\mathbf{T}_\Phi(M)$ and $\mathbf{T}_\Psi(N)$. The second approach uses 
the representations $\delta_\Phi, \delta_\Psi$ of $M$ and $N$ respectively, induced by the 
computable predicate spaces $\mathbf{Z}_\Phi$ and $\mathbf{Z}_\Psi$. The two options are equivalent 
thanks to Definition \ref{defCompFunctionCompTopSpace} and part 2. of Lemma
\ref{lemmaCompPredicateSpaceISCompTOPSPACE}. As our formal definition, we adopt the second approach.

\begin{definition}\label{defComputableFuncManifolds}
Let $(M, \Phi),(N, \Psi)$ be computable manifolds. A morphism
$f\colon \subseteq (M, \Phi)$ $\to (N, \Psi)$ between computable manifolds is a $(\delta_\Phi,
\delta_\Psi)$-computable function $f\colon \subseteq M \to N$, where $\delta_\Phi, \delta_\Psi$ are 
the representations induced by $\mathbf{Z}_\Phi$ and $\mathbf{Z}_\Psi$ respectively.
\end{definition}

\begin{lemma}\label{lemmaCompManifoldsMorphismCompFuncCOMPTOPSPACES}
 Let $(M, \Phi),(N, \Psi)$ be computable manifolds and $f\colon \subseteq M \to N$ a function. Then 
$f$ is a morphism of computable manifolds if and only if $f$ is a computable map between the 
computable spaces $\mathbf{T}_\Phi(M)$ and $\mathbf{T}_\Psi(N)$.
\end{lemma}

\begin{proof}
 $(\Rightarrow)$ If $f$ is a morphism of computable manifolds, then $f$ is $(\delta_\Phi,
\delta_\Psi)$-computable. But we know from part 2 of Lemma 
\ref{lemmaCompPredicateSpaceISCompTOPSPACE} that $\delta_\Phi$ is equivalent to the inner 
representation of points of the computable space $\mathbf{T}_\Phi(M)=T(\mathbf{Z}_\Phi)$. A similar 
statement is true for $\delta_\Psi$ and $\mathbf{T}_\Psi(N)$. Therefore $f$ is a computable 
function with respect to the inner representations of $\mathbf{T}_\Phi(M)$ and 
$\mathbf{T}_\Psi(N)$, so that by Definition \ref{defCompFunctionCompTopSpace}, $f$ is a computable 
map between $\mathbf{T}_\Phi(M)$ and $\mathbf{T}_\Psi(N)$.

$(\Leftarrow)$ If $f$ is a computable map between $\mathbf{T}_\Phi(M)$ and $\mathbf{T}_\Psi(N)$, 
then by Definition \ref{defCompFunctionCompTopSpace}, $f$ is computable with respect to the inner 
representations of points of $\mathbf{T}_\Phi(M)$ and $\mathbf{T}_\Psi(N)$ and these 
representations are equivalent to the representations $\delta_\Phi$ and $\delta_\Psi$ respectively 
(by 2 of Lemma \ref{lemmaCompPredicateSpaceISCompTOPSPACE}), so that by Definition 
\ref{defComputableFuncManifolds}, $f$ is a morphism between the computable manifolds $M$ and $N$.
\end{proof}

Recall that computable functions are continuous. Of course, we would like computable
homeomorphisms to be the standard equivalence between computable manifolds. We now prove some
results about this topic. First, we have an analog of Lemma \ref{lemmaTopSpaceCompFuncRepSet} for
computable manifolds.

\begin{lemma}\label{lemmaManifoldFuncRepSet}
 Let $(M, \left[ \Phi \right])$ be a computable $n$-manifold and $X$ a set. 
If $f\colon M \to X$ is a bijective function, then there exists a computable structure
$\left[ \Phi_f \right]$ on $X$ induced by $f$ and $\Phi$, such that for all $z\in \dom 
\lambda_{\Phi_f},
\lambda_{\Phi_f}(z)=\directimage{\lambda_\Phi(z)}{f}$ and $\delta_{\Phi_f}=f\circ \delta_\Phi$.
\end{lemma}

\begin{proof}
Let $\Phi=\{ (\varphi_i, U_i) \}_{i\in I}$. By Lemma \ref{lemmaTopSpaceCompFuncRepSet},
$\mathbf{X}_{f}$ becomes a computable topological space and $f$ is a computable homeomorphism
between $\mathbf{T}_\Phi(M)$ and $\mathbf{X}_{f}$. Now, the atlas induced by $f$ on $X$ is given by 
\[ \Phi_f = \{ (\psi_i, V_i) \mid \psi_i = \varphi_i \circ f^{-1} \text{ and } V_i =
\directimage{U_i}{f}, (\varphi_i, U_i)\in \Phi \}_{i\in I}, \]
which induces the computable predicate space $\mathbf{Z}_{\Phi_f} = (X, \mathcal{B}_{\Phi_f},
\lambda_{\Phi_f})$, where 
\begin{itemize}
 \item [] $\mathcal{B}_{\Phi_f} = \{ \inverseimage{\mu^n(w)}{(\varphi_i \circ f^{-1})} \mid (i,w)\in
\sigmastringset \times \domp{\mu^n_+} \}$,
\item [] $\lambda_{\Phi_f} \colon \subseteq \sigmastringset \to \mathcal{B}_{\Phi_f}$ is given by 
$\lambda_{\Phi_f}(\langle i,w \rangle)=\inverseimage{\mu^n_+(w)}{(\varphi_i \circ f^{-1})}$.
\end{itemize}
Clearly $\lambda_{\Phi_f}(j)=\directimage{\lambda_\Phi(j)}{f}$, so that $\delta_{\Phi_f}(p) =
f\circ \delta_\Phi(p)$ for all $p\in \dom \delta_{\Phi_f}$. For each $i\in I$, the open set
$\directimage{V_i}{\psi_i}= \directimage{\directimage{U_i}{f}}{(\varphi_i \circ
f^{-1})}=\directimage{U_i}{\varphi_i}$ is $\theta^n$-computable open in $\euclidean{n}$ and the
map $\psi_i$ is $(\delta_{\Phi_f}, \delta^n)$-computable with $(\delta^n, \delta_{\Phi_f})$-computable
inverse. It follows that $\Phi_f$ satisfies Definition \ref{defAbstractComputableAtlas}, thus it
induces a computable structure $\left[ \Phi_f \right]$ on $X$.
\end{proof}

\begin{corollary}\label{corManifoldFuncRepSetCompTOPSPACE}
 With the hypothesis of Lemma \ref{lemmaManifoldFuncRepSet},
$\mathbf{T}_{\Phi_f}(X)=\mathbf{X}_{f}$, where $\mathbf{X}_{f}$ is the computable topological space
of Lemma \ref{lemmaTopSpaceCompFuncRepSet}.
\end{corollary}

\begin{proof}
The topology of $\mathbf{T}_{\Phi_f}(X)$ is generated by the base
\[ \beta_{\Phi_f} = \{ \directimage{B}{f} \mid B \in \beta_\Phi \}, \]
and this is exactly the set $\beta_{f}$; so that $\nu_{\Phi_f} = \nu_{f}$. In
other words, $\mathbf{T}_{\Phi_f}(X)=\mathbf{X}_{f}$.
\end{proof}

\begin{corollary}\label{corManifoldFuncTopSpace}
 Let $(M, \Phi)$ be a computable manifold and $\mathbf{X} = (X, \tau, \beta,
\nu)$ be a computable topological space. If $f\colon M \to X$ is a computable homeomorphism between
$\mathbf{T}_{\Phi}(M)$ and $\mathbf{X}$, then
$\mathbf{T}_{\Phi_f}(X)$ is equivalent to $\mathbf{X}$.
\end{corollary}

\begin{proof}
 It is an immediate consequence of Lemma \ref{corCompHomeoTopSpaces} and Corollary
\ref{corManifoldFuncRepSetCompTOPSPACE}.
\end{proof}

\begin{corollary}\label{corManifoldFuncManifold}
 Let $(M, \Phi), (N, \Phi^\prime)$ be computable manifolds. If $f\colon M \to N$ is a computable
homeomorphism, then the computable structure of $N$ and the computable structure induced by
$f$ on $N$
 are the same.
\end{corollary}

\begin{proof}
By Corollary \ref{corManifoldFuncTopSpace}, $\mathbf{T}_{\Phi_f}(N)$ and
$\mathbf{T}_{\Phi^\prime}(N)$ are equivalent computable topological spaces, so that we can use Lemma
\ref{lemmaCompatibleCompAtlasesIsSameCompTOP} to conclude that $\Phi_{f}$ and $\Phi^\prime$ are
computably compatible (Definition \ref{defComputablyCompatibleAtlases}), therefore $\left[ \Phi_{f}
\right] = \left[ \Phi^\prime \right]$.
\end{proof}

With this last result, we can see that two computably homeomorphic computable manifolds are
``essentially'' the same manifold.

\begin{example}\label{exampleSphereminusPchomeomRn}
 Let $h=s^{-1}\colon \euclidean{n} \to \puncturedsphere{n}$ ($\puncturedsphere{n}=\sphere{n} - \{ P \}$) 
 be the inverse of the stereographic projection of Example \ref{exampleStereographicProjection}. 
 If we equip
$\euclidean{n}$ with the (usual) manifold structure of Example \ref{exampleCompAtlasEuclideanSpace},
then by Lemma \ref{lemmaManifoldFuncRepSet} and Corollary \ref{corManifoldFuncTopSpace}, $\puncturedsphere{n}$
becomes a computable $n$-manifold with a computable structure $\Gamma$ such that
$\mathbf{T}_\Gamma(\puncturedsphere{n})$ is equivalent to $\ctopcompsubspace{\mathbf{S}^n}{\puncturedsphere{n}}$.
\end{example}

We will have more to say about computable functions between computable manifolds in Section
\ref{secComputableSubmanifolds}, where we introduce \emph{computable submanifolds} and study
\emph{computable embeddings} of manifolds.

\section{Properties of computable manifolds}

We now present some properties of computable manifolds. First, we analyze which 
topological properties are satisfied by the induced topology in a computable manifold $M$, then we present
some results about computability in $M$. Finally, we present the computable version of
classical results about special atlases on $M$.

\subsection{Topological and computable properties}\label{subsecCompTOPProperties}

By definition, a computable topological space is a $T_0$-space and thus, every computable manifold
is a $T_0$-space. However, a manifold also satisfies all the local topological properties of
euclidean space $\euclidean{n}$. Here is a list of some basic properties that the induced
topology on a computable manifold satisfies: The $T_1$ separation axiom; the second axiom of 
countability (because every computable manifold has a coun\-ta\-ble atlas); it is locally connected;
it is locally compact. But despite the fact that $\euclidean{n}$ is a (computably) Hausdorff space, it
is not true that every computable manifold is (computably) Hausdorff. We now prove that a well known
example of a non-Hausdorff topological manifold is in fact, an example of a non-Hausdorff computable
manifold.

\begin{example}\label{exampleNONHausdorffCOMPMAN}
 Let $H\subset \euclidean{2}$ be defined by $H=\{ (s,0)\mid s\in \reals \} \cup \{(0,1)\}$. $H$ is
called \emph{the line with two origins}. Let $U$ be the subset of $H$ of all points of the form
$(s,0)$ with $s\in \reals$ and $U^\prime=(U - \{ (0,0) \}) \cup \{ (0,1) \}$. Define charts $f
\colon U \to \reals$ and $f^\prime \colon U^\prime \to \reals$ of $H$ into $\reals$ by
\begin{itemize}
 \item [] $f(s,0)= s$
 \item [] $f^\prime(s,0)= s$ for $s\neq 0$ and $f^\prime(0,1)= 0$
\end{itemize}
It is very easy to show that $\Phi = \{ (f, U), (f^\prime, U^\prime) \}$ is an atlas on $H$
and that each chart is $(\delta_\Phi, \delta^1)$-computable with $(\delta^1,
\delta_\Phi)$-computable inverse. Since 
$\directimage{U}{f}=\directimage{U^\prime}{f^\prime}=\reals$, 
each of these sets is $\theta^1$-computable open in the computable space $\mathbf{R}$. The atlas 
$\Phi$ satisfies Definition \ref{defAbstractComputableAtlas} so that the pair $(H, \Phi)$ becomes a 
computable 1-manifold. The proof that $H$ is a non-Hausdorff space with the topology induced by 
$\Phi$ can be found in \cite{brickellclarkDIFF,1974TOPPROPERTIES}.
\end{example}

Recall from Section \ref{subsecCompPredicateSpacebyAtlas} that given a computable manifold $(M,
\Phi)$ using the atlas $\Phi$, we cons\-truc\-ted the computable predicate space 
$\mathbf{Z}_\Phi=(M, \mathcal{B}_\Phi, \lambda_\Phi)$ where  $\mathcal{B}_\Phi$ is the set of all 
computable balls in $M$. This predicate space depends on $\Phi$ and it can happen that 
$\mathcal{B}_\Phi$ contains empty elements. 

\begin{definition}\label{defComputableBalls}
Let $(M, \Phi)$ be a computable $n$-manifold. If the subset $E \subseteq \dom \lambda_\Phi$ defined by
\begin{equation}\label{eqNonEmptyCompEuclideanBalls}
 E = \{ l \in \dom \lambda_\Phi \mid \lambda_\Phi(l)\neq \myemptyset \}
\end{equation}
is c.e., then $(M, \Phi)$ will be called a \emph{computable manifold with non-empty computable 
balls}. 
\end{definition}

The next lemma states that if the set of non-empty computable balls is
c.e., then the empty elements of $\mathcal{B}_\Phi$ can be removed from the computable
spaces $\mathbf{Z}_\Phi$ and $\mathbf{T}_\Phi(M)$\footnote{This result is reminiscent of Lemma 25 of
\cite{KWECTOP}, which says that for a computable topological space $\mathbf{X}$, the empty base
elements can be ignored if the set of non-empty base elements is c.e.}.

\begin{lemma}\label{emptyEuclideanBallsRemoved}
 Let $(M,\Phi)$ be a computable $n$-manifold with non-empty computable balls. Then there
exists a computable topological space $\mathbf{T}^\prime(M)=(M, \tau_\Phi, \beta^\prime,
\nu^\prime)$ such that $\mathbf{T}_\Phi(M)$ is equivalent to $\mathbf{T}^\prime(M)$ and 
$\beta^\prime \subset \beta_\Phi$ is the set of all finite intersections of non-empty computable 
balls of $M$.
\end{lemma}

\begin{proof}
 Let $\mathcal{E} = \lambda_\Phi(E)$. Since the set $E \subseteq \dom \lambda_\Phi$ is c.e. 
(and infinite), 
there exists an 
injective total computable function $h\colon \sigmastringset \to $ $\sigmastringset$ such that 
$\range h=E$. Define the notation $\lambda_h \colon \sigmastringset \to \mathcal{E}$ by 
$\lambda_h(w)= \lambda_\Phi \circ h (w)$. Then the triple $\mathbf{Z}_h = (M, \mathcal{E}, 
\lambda_h)$ is a computable predicate space and since $\lambda_h = \lambda_\Phi \circ h$ and 
$\lambda_\Phi = \lambda_h \circ h^{-1}$, we have that $\lambda_\Phi \equiv \lambda_h$, and this 
fact implies that $\delta_\Phi \equiv \delta_h = \delta_{\mathbf{Z}_h}$. 
The computable topological 
space 
induced by $\mathbf{Z}_h$ is 
\[ \mathbf{T}^\prime(M) = (M, \tau_\Phi, \beta^\prime, \nu^\prime), \]
where $\beta^\prime$ is the base generated by all finite intersections of the elements of
$\mathcal{E}\subseteq \mathcal{B}_\Phi$, thus $\beta^\prime \subset \beta_\Phi$. Since $\delta
\equiv \delta_\Phi \equiv \delta_h\equiv \delta^\prime$, where $\delta,\delta^\prime$ are the inner
representations of $M$ induced by $\mathbf{T}_\Phi(M)$ and $\mathbf{T}^\prime(M)$ respectively, we
conclude that $\mathbf{T}_\Phi(M)$ and $\mathbf{T}^\prime(M)$ are equivalent computable topological
spaces, hence the result follows.
\end{proof}

In a computable $n$-manifold, the computable points can be characterized by the computability of
points in the computable euclidean space $\mathbf{R}^n$. We present a simple result from which
other characterizations can be derived.

\begin{lemma}\label{lemmaComputablePointsCOMPMAN}
 Let $(M, \Phi)$ be a computable $n$-manifold and $x\in M$. Then $x$ is a $\delta_\Phi$-computable
point in $M$ if and only if there exists a $\delta^n$-computable point $y\in \euclidean{n}$ and a
computable function $f\colon \subseteq \euclidean{n} \to M$ such that $f(y)=x$.
\end{lemma}

\begin{proof}
$(\Rightarrow)$ Let $x\in M$ be a $\delta_\Phi$-computable point. There exists a chart $(\varphi, 
U)\in \Phi$ such that $x\in U$ and $\varphi$ is a computable homeomorphism between $U$ and a computable
 open 
subset of $\euclidean{n}$. In particular, $\varphi$ and $\varphi^{-1}$ are computable with 
respect to $\delta^n$ and $\delta_\Phi$, which implies that $\varphi$ and $\varphi^{-1}$ preserve 
the computability of points between $U$ and $\directimage{U}{\varphi}$. Let $y=\varphi(x)$ and 
$f=\varphi^{-1}$. Then $y$ is a $\delta^n$-computable point in $\euclidean{n}$ 
and $f$ is a computable map such that $f(y)=x$.

$(\Leftarrow)$ Since $f$ is computable, it takes the $\delta^n$-computable point $y\in 
\euclidean{n}$ onto a $\delta_\Phi$-computable point in $M$, thus $f(y)=x$ must be 
$\delta_\Phi$-computable in $M$.
\end{proof}

\subsection{Some special computable atlases}\label{subsecSpecialAtlases}

In Theorem \ref{thmCompManifoldEuclideanAtlas}, we deal with the existence of two very useful 
computable atlases. To show that
one of these atlases exists, our computable manifolds will need the property of non-empty computable
 balls and a technical result about computable homeomorphisms between rational open balls
$B(q, \epsilon)\subset \euclidean{n}$ and $\euclidean{n}$.

\begin{lemma}\label{lemmaTechResultCompHomeoBallsRn}
In the computable euclidean space $\mathbf{R}^n$, the following statements hold:
\begin{itemize}
 \item [(1)] For each $w\in \domp{\mu^n}$ such that $\mu^n(w)=B(w) =B(q, \epsilon)$, there exists a
computable homeomorphism $h_w\colon B(w) \to \euclidean{n}$ from $\ctopcompsubspace{\mathbf{R}^n}{B(w)}$ to
$\mathbf{R}^n$.
\item [(2)] The set
\begin{equation}\label{eqlemmaTechResultCompHomeoBallsRn}
 \{ (v,w,z) \in (\domp{\mu^n})^3 \mid \mu^n(v)\subset \inverseimage{\mu^n(z)}{h_w} \}
\end{equation}
is computable, where $h_w$ is the computable homeomorphism of part (1).
\end{itemize}
\end{lemma}

\begin{proof}
 (1) Let $h\colon B(0, 1) \to \euclidean{n}$ be the computable homeomorphism of Example
\ref{exampleBallCTEuclideanSpace} and let $S_\varepsilon,T_a \colon$ $\euclidean{n} \to
\euclidean{n}$ $(a\in \euclidean{n}, 0 < \varepsilon \in \reals)$ be the homeomorphisms of
$\euclidean{n}$ defined by
\[ T_a(x)=x - a \quad \text{and} \quad S_\varepsilon(x)= \varepsilon x. \]
If $(a,\varepsilon)\in \rationalsspace{n} \times \rationals^+$, then $T_a$ and $S_\varepsilon$ are
computable homeomorphisms of $\mathbf{R}^n$ onto itself. Now, if $\mu^n(w)=B(q, \varepsilon)$,
then we can define $h_w \colon  B(q, \varepsilon)\to \euclidean{n}$ as $h_w= h \circ
S_{\varepsilon^{-1}} \circ T_q$, clearly $h_w$ is a computable homeomorphism.

(2) Using the computable homeomorphisms $h_w$ of (1), we can express the property
 $\mu^n(v)\subset \inverseimage{\mu^n(z)}{h_w}$ 
in terms of polynomial functions in $n$ 
variables and
rational coefficients, this is true because we know that
\begin{equation}\label{eqlemmaTechResultCompHomeoBallsRn3}
 \mu^n(v)\subset \inverseimage{\mu^n(z)}{h_w} \Leftrightarrow (\forall x \in \euclidean{n}) (x 
\in \mu^n(v) \Rightarrow h_w(x) \in \mu^n(z))
\end{equation}
and it can be seen that the right side of \eqref{eqlemmaTechResultCompHomeoBallsRn3} is equivalent 
to the expression
\begin{equation}\label{eqlemmaTechResultCompHomeoBallsRn32}
 (\forall x \in \euclidean{n}) (d(x, q_v) < \varepsilon_v \Rightarrow d(h_w(x), q_z) < 
\varepsilon_z),
\end{equation}
where $\mu^n(v)=B(q_v, \varepsilon_v)$ and $\mu^n(z)=B(q_z,\varepsilon_z)$. Now, using the 
formula to compute the function $h_w$, which is defined by 
\[ h_w(x) = \frac{\varepsilon^{-1}(x - q)}{1 - \lVert \varepsilon^{-1}(x - q) \rVert}, \quad 
(\mu^n(w)=B(q,\varepsilon)) \]
for the distance $d(h_w(x), q_z)$, we have that 
\begin{eqnarray*}
d(h_w(x), q_z)
& = &
 d\biggl(\frac{\varepsilon^{-1}(x - q)}{1 - \lVert \varepsilon^{-1}(x - q) \rVert}, q_z\biggr) \\
& = &
\frac{1}{1 - \lVert \varepsilon^{-1}(x - q) \rVert} d(\varepsilon^{-1}(x - q), (1 - 
\lVert \varepsilon^{-1}(x - q) \rVert) q_z).
\end{eqnarray*}
Using the last expression, we can deduce that 
\begin{equation}
 d(h_w(x), q_z) < \varepsilon_z \Leftrightarrow 
 d(\varepsilon^{-1}(x - q), (1 - \lVert \varepsilon^{-1}(x - q) \rVert) q_z) < \varepsilon_z(1 - 
\lVert \varepsilon^{-1}(x - q) \rVert).
\end{equation}
Let $y=\varepsilon^{-1}(x - q)$. Now we can write Equation 
\eqref{eqlemmaTechResultCompHomeoBallsRn32} as 
\[ (\forall x,y \in \euclidean{n}) (y=\varepsilon^{-1}(x - q) \wedge d(x, q_v) < \varepsilon_v 
\Rightarrow d(y, (1 - \lVert y \rVert) q_z) < \varepsilon_z(1 - 
\lVert y \rVert)). \]
Now let $r=1 - \lVert y \rVert \in \reals$, then $\lVert y \rVert^2=(1 - r)^2$ and the above 
equation becomes
\begin{equation}\label{eqlemmaTechResultCompHomeoBallsRn34}
 (\forall x,y \in \euclidean{n},\forall r \in \reals) (y=\varepsilon^{-1}(x - q) \wedge 
\lVert y \rVert^2=(1 - r)^2 \wedge d(x, q_v) < \varepsilon_v 
\Rightarrow d(y, r q_z) < \varepsilon_zr).
\end{equation}
This last expression can be 
easily converted into an elementary expression and it can be seen that given $v,w,z\in 
\domp{\mu^n}$ the polynomial expression defining Equation 
\eqref{eqlemmaTechResultCompHomeoBallsRn34} can be constructed algorithmically, so that the
de\-ci\-da\-bi\-li\-ty of \eqref{eqlemmaTechResultCompHomeoBallsRn} can be verified by a single
Turing machine (using Theorem \ref{thmTarskiDecisionMethod}) uniformly in $(v,w,z) \in
(\domp{\mu^n})^3$.
\end{proof}

\begin{theorem}[Special computable atlases]\label{thmCompManifoldEuclideanAtlas}
 Let $(M,\Phi)$ be a computable $n$-manifold.
\begin{itemize}
 \item [(a)] There exists a computable atlas $\Gamma$ on $M$, computably compatible with $\Phi$
and such that for every chart $h \colon V\to \euclidean{n}\in \Gamma$,
$\directimage{V}{h}=B(q,\epsilon)$ for some $(q,\epsilon)\in \rationals^n \times \rationals^+$.
\item [(b)] If $(M,\Phi)$ has the property of non-empty computable euclidean balls, there exists a
computable atlas $\Psi$ on $M$, computably compatible with $\Phi$ and such that for every chart
$\psi \colon V\to \euclidean{n}$ of $\Psi$, $\directimage{V}{\psi}=\euclidean{n}$.
\end{itemize}
\end{theorem}

\begin{proof}
 (a) For the computable atlas $\Phi$, recall that an element $B_j\in \mathcal{B}_\Phi$ is defined
as $B_j=\inverseimage{\mu^n(w)}{\varphi_i}$, where $j=\langle i, w \rangle$. If $B_j\neq
\myemptyset$, then we can define a chart $h_j\colon B_j \to \euclidean{n}$ by
$h_j=\varphi_i|_{B_j}$. Let $\Gamma$ be the set of all such charts, $\Gamma$ is clearly a computable
atlas on $M$.
Now we have to show that $\Phi$ and $\Gamma$ are computably compatible. Notice
that if $z = \langle j, w \rangle= \langle \langle i, u \rangle, w \rangle\in \dom \lambda_\Gamma$
is such that $\lambda_\Gamma(z) \neq \myemptyset$, then
\[ \lambda_\Gamma(z) = \inverseimage{\mu^n(w)}{h_j} 
 = \inverseimage{\mu^n(w)}{(\varphi_i|_{B_j})} 
 = \inverseimage{\mu^n(w)}{\varphi_i} 
 = \lambda_\Phi(\langle i, w \rangle).\]
and if $v=\langle i, u \rangle\in \dom \lambda_\Phi$ with $\lambda_\Phi(v)\neq \myemptyset$, then
\[ \lambda_\Phi(v) = \inverseimage{\mu^n(u)}{\varphi_i} 
 = \inverseimage{\mu^n(u)}{(\varphi_i|_{B_v})}
 = \lambda_\Gamma(\langle v, u \rangle).\]
There is a Type-2 machine $T$ that on input $p\in \dom \delta_\Gamma$, extracts each string $z =
\langle \langle i, u \rangle, w \rangle \ll p$, computes the string $l=\langle i, w \rangle\in \dom
\lambda_\Phi$ and prints $\iota (l)$ on the output tape. The function $f_T$ calculated by $T$
translates $\delta_\Gamma$-names into $\delta_\Phi$-names. A Type-2 machine which translates
$\delta_\Phi$-names into $\delta_\Gamma$-names is build similarly. Therefore, $\delta_\Phi \equiv
\delta_\Gamma$.

(b) If $(M, \Phi)$ has the property of non-empty computable euclidean balls, then we can assume
that for all $k\in \dom \lambda_\Phi$, $\lambda_\Phi(k)\neq \myemptyset$. For each $B_{j}\in
\mathcal{B}_\Phi$ ($j= \langle i, w \rangle$), define a chart
$\psi_{j}\colon B_{j}\to \euclidean{n}$ as $\psi_{j}=h_w \circ \varphi_i|_{B_j}$, where $h_w$ is the
computable homeomorphism $h_w\colon \mu^n(w) \to \euclidean{n}$ of Lemma
\ref{lemmaTechResultCompHomeoBallsRn}. The set of charts $\Psi = \{ \psi_{l} \colon B_{l} \to
\euclidean{n} \}$ is a topological atlas on $M$, compatible with $\Phi$. Moreover, $\Psi$ is
computable, because the $\psi_{l}$'s and their inverses are computable with respect to
$\delta_{\Psi}$ and $\delta^n$. Also, $\euclidean{n}$ is trivially a $\theta^n$-computable open set.
We claim that $(M, \Psi)$ is a computable manifold with non-empty computable euclidean balls. If
$u\in \dom \lambda_\Psi$, then $\lambda_\Psi(u) = \inverseimage{\mu^n(z)}{\psi_j}$, where $u=\langle
j, z \rangle$ and $j=\langle i, w \rangle \in \dom \lambda_\Phi$. Now,
$\lambda_\Psi(u) \neq \myemptyset \Leftrightarrow \inverseimage{\mu^n(z)}{\psi_{j}} \neq \myemptyset 
\Leftrightarrow \inverseimage{\mu^n(z)}{(h_w \circ \varphi_i)} \neq \myemptyset 
\Leftrightarrow j=\langle i,w \rangle \in \dom \lambda_\Phi \wedge \mu^n(z) \subset
\directimage{\directimage{B_{j}}{\varphi_i}}{h_w} 
\Leftrightarrow j=\langle i,w \rangle \in \dom \lambda_\Phi \wedge \mu^n(z) \subset
\directimage{\mu^n(w)}{h_w}
\Leftrightarrow j \in \dom \lambda_\Phi \wedge \mu^n(z) \subset
\euclidean{n} 
\Leftrightarrow j \in \dom \lambda_\Phi$. 
This shows that the set $\{ u\in \dom \lambda_\Psi \mid \lambda_\Psi(u) \neq \myemptyset \}$ is
c.e., so that $(M, \Psi)$ is a computable manifold with non-empty computable
euclidean balls.

To finish the proof of the theorem, we only need to show that $\Phi$ and $\Psi$ are computably
compatible. By Lemma \ref{lemmaCompatibleCompAtlasesIsSameCompTOP}, we can do this by proving that
$\mathbf{T}_\Phi(M)$ and $\mathbf{T}_{\Psi}(M)$ are equivalent computable spaces, which means that
we have to show that $\nu_\Phi \leq \theta_{\Psi}$ and $\nu_{\Psi} \leq \theta_\Phi$.
\begin{itemize}
 \item [] $\nu_\Phi \leq \theta_{\Psi}$. An element $B$ of $\beta_\Phi$ is a finite
intersection of the form
\[ B = \nu_\Phi(\iota(j_1) \cdots \iota(j_k)) = \bigcap_{j \in F} \lambda_\Phi(j)\quad F\subset \dom
\lambda_\Phi, \]
where $F= \{ j_1, \ldots, j_k \}$. Notice that for any $j \in F$ $(j = \langle i, w \rangle)$, 
we have that 
\begin{multline*}
\lambda_\Phi(j)
= \inverseimage{\mu^n(w)}{\varphi_i}  
= \inverseimage{\inverseimage{\euclidean{n}}{h_w}}{\varphi_i} 
=\inverseimage{\inverseimage{\bigcup_{z\in \domp{\mu^n}} \mu^n(z)}{h_w}}{\varphi_i} \\
= \bigcup_{z\in \domp{\mu^n}}\inverseimage{\mu^n(z)}{(h_w \circ \varphi_i)} 
= \bigcup_{z\in \domp{\mu^n}}\inverseimage{\mu^n(z)}{\psi_{j}} 
= \bigcup_{z\in \domp{\mu^n}}\lambda_\Psi(\langle j, z \rangle),
\end{multline*}
hence 
\begin{multline*}
B = \nu_\Phi(\iota(j_1) \cdots \iota(j_k)) = \bigcap_{j \in F} \lambda_\Phi(j) 
= \bigcap_{j \in F} \Bigl(\bigcup_{z_r \in \domp{\mu^n}} \lambda_\Psi(\langle j, z_r \rangle)
\Bigr) \\
= \bigcup_{z_r\in \domp{\mu^n}} \Bigl(\bigcap_{j \in F} \lambda_\Psi(\langle j, z_r
\rangle) \Bigr) = \bigcup_{\substack{z_r \in \domp{\mu^n} \\ j_l \in F}} \nu_\Psi(\iota(\langle
j_{i_1}, z_1 \rangle) \cdots \iota(\langle j_{i_k}, z_k \rangle)).
\end{multline*}
With all this data, we can construct a Type-2 machine that, 
on input $u = \iota(j_1) \cdots \iota(j_k) \in \dom \nu_\Phi$, computes an element $p_u\in
\sigmaomegastringset$ such that $p_u$ is a list of all strings $\iota(\langle j_{i_1}, z_1 \rangle)
\cdots $ $\iota(\langle j_{i_k}, z_k \rangle)$ where $z_l \in \domp{\mu^n}$ and
$j_{i_1},\ldots,j_{i_k} \in F$. This reduction works because $(M,\Phi)$ and $(M, \Psi)$ have the
property of non-empty euclidean balls (it allows us to avoid translate the name of an empty base 
element of $\beta_\Phi$ into a non-empty element of $\tau_\Phi=\tau_\Gamma$). We conclude that
$\nu_\Phi \leq \theta_{\Psi}$. 
\item [] $\nu_{\Psi} \leq \theta_\Phi$. An argument similar to the previous one can be
used to prove this case, because for any $l = \langle j, z \rangle \in \dom \lambda_\Psi$ 
$(j = \langle i, w \rangle)$, 
$\lambda_\Psi(\langle j, z \rangle) 
= \inverseimage{\mu^n(z)}{\psi_{j}} 
= \inverseimage{\mu^n(z)}{(h_w \circ \varphi_i)} 
=  \inverseimage{\inverseimage{\mu^n(z)}{h_w}}{\varphi_i} 
= \bigcup_{u\in C_{wz}} \inverseimage{\mu^n(u)}{\varphi_i} 
= \bigcup_{u\in C_{wz}} \lambda_\Phi(\langle i, u \rangle)$, 
where $C_{wz} = \{ v\in \domp{\mu^n} \mid \mu^n(v)\subset \inverseimage{\mu^n(z)}{h_w} \}$ 
is a computable subset of $\domp{\mu^n}$ (uniformly in $w,z$, apply (b) of Lemma
\ref{lemmaTechResultCompHomeoBallsRn}). Using this fact, we can construct the Type-2 machine which
computes the function that translates $\lambda_\Psi$-names into $\theta_\Phi$-names. Again, the
property of non-empty euclidean balls is being used to avoid translating the name of an empty element
into a non-empty element.
\end{itemize}
We have proven that $\mathbf{T}_{\Psi}(M)$ and $\mathbf{T}_\Phi(M)$ are equivalent
computable topological spaces. By Lemma \ref{lemmaCompatibleCompAtlasesIsSameCompTOP}, $\Phi$ and
$\Psi$ are computably compatible atlases on $M$. The result follows.
\end{proof}

{\it Remark.} 
Since any compact computable manifold admits a finite computable atlas, all such
manifolds have the property of non-empty computable euclidean balls, thus (b) of Theorem
\ref{thmCompManifoldEuclideanAtlas} is valid for these manifolds.

\section{Computable submanifolds}\label{secComputableSubmanifolds}

One of the most important concepts in the theory of manifolds is that of \emph{submanifold}, that
is, a manifold which is a subset of another manifold. In this section, we will develop the
corresponding concept of \emph{computable submanifold}. 

\begin{definition}\label{defCompSubmanifold}
 A computable manifold $(M, \Phi)$ is a \emph{computable submanifold} of $(N, \Psi)$ if and only if
$M\subset N$ and the inclusion $i\colon M \hookrightarrow N$ is a computable embedding of
$\mathbf{T}_\Phi(M)$ into $\mathbf{T}_\Psi(N)$.
\end{definition}

A computable submanifold is just a subset of a computable manifold $N$, which is also a
computable manifold in its own right, with computable subspace topology. A simple
example of a computable submanifold is the computable $n$-sphere $\sphere{n}\subset
\euclidean{n+1}$. From the definition, it is clear that the computable topology of subspace
characterizes the computable structure of a computable submanifold. The following is an important
example of a computable submanifold of a computable manifold.

\begin{proposition}\label{lemmaOpenSubsetSubmanifoldIF}
 Let $(M,\Phi)$ be a computable manifold and $\myemptyset \neq W\subset M$ a
$\theta_\Phi$-computable open subset of $M$. Then there exists a computable atlas $\Psi$ on $W$
which makes $W$ into a computable submanifold of $M$.
\end{proposition}

\begin{proof}
 Let $W\subseteq M$ be any non-empty $\theta_\Phi$-computable open subset of $(M, \Phi)$. We give
$W$ a computable structure induced by that of $M$. Since $W$ is a $\theta_\Phi$-computable open set in 
$\mathbf{T}_\Phi(M)$, there exists a computable infinite string $q_W\in \sigmaomegastringset$ such 
that
\[ \theta_\Phi(q_W)=W. \]
For each $j \ll q_W$ such that $\nu_\Phi(j)=V_j$ and $ V_j \neq \myemptyset$, there exists a chart 
$(\varphi_i, U_i)\in \Phi$ such that $V_j \subset U_i$ and $\psi_{\langle i, j \rangle}=\varphi_i 
|_{V_j} \colon V_j \to 
\euclidean{n}$ is a homeomorphism onto an open subset of $\euclidean{n}$, so that the pair 
$(\psi_{\langle i, j \rangle}, V_j)$ is a chart on $W$. Let $\Phi_W$ be defined by 
\[ \Phi_W = \{ (\psi_{\langle i, j \rangle}, 
V_j) \mid j \ll q_W \wedge V_j \neq \myemptyset \wedge V_j \subset W \cap U_i \wedge \psi_{\langle 
i, j \rangle}= \varphi_i |V_j \wedge (\varphi_i, U_i) \in \Phi \}. \]
 $\Phi_W$ is a topological atlas on $W$, and it can be easily verified that all charts $\psi_{\langle i, j
\rangle}$ with their respective inverses are computable with respect to $\delta_{\Phi_W}$ and
$\delta^n$ and the sets $\directimage{V_j}{\psi_{\langle i, j \rangle}}$ are $\theta^n$-computable
in $\euclidean{n}$. By Definition \ref{defAbstractComputableAtlas}, $\Phi_W$ is a computable atlas
on $W$, thus the pair $(W,\Phi_W)$ is a computable $n$-manifold. Let 
$\mathbf{T}=\mathbf{T}_\Phi(M)$, we now prove that $\ctopcompsubspace{\mathbf{T}}{W}$ and 
$\mathbf{T}_{\Phi_W}(W)$ are 
equivalent computable topological spaces by showing that $\delta_W\equiv\delta_{\Phi_W}$.

$\delta_{\Phi_W} \leq \delta_W$. For any $k=\langle j, z \rangle \in \dom \lambda_{\Phi_W}$ with 
$j=\langle i,w \rangle$ and $\lambda_{\Phi_W}(k) \neq \myemptyset$, we have that
\[ 
\lambda_{\Phi_W}(k)=\inverseimage{\mu^n(z)}{\psi_j}=\inverseimage{\mu^n(z)}{(\varphi_{i}|_{V_j})
} =\inverseimage{\mu^n(z)}{\varphi_{i}}=\lambda_\Phi(\langle i, z \rangle) \subset W. \]
Using the equality $\lambda_{\Phi_W}(k)=\lambda_\Phi(\langle i, z \rangle)$, it is easy to build a 
Type-2 machine which translates $\delta_{\Phi_W}$-names into $\delta_{W}$-names.

$\delta_W \leq \delta_{\Phi_W}$. This is the part of the proof where the computability of $W$ 
comes into play, we also need to use a ``time-sharing'' technique. A Type-2 machine $\mathfrak{M}$ 
can be 
build with the following program: On input $p\in \dom \delta_W$, $\mathfrak{M}$ starts enumerating 
all strings 
$z\in \sigmastringset$ such that $z=w y$; $w\ll q_W$ and $y \ll p$. This enumeration can be 
computed 
by $\mathfrak{M}$, because $q_W$ can be calculated. Notice also that $w\in \dom \nu_\Phi$, $y \in 
\dom 
\nu_W=\dom \nu_\Phi$, so that $z\in \dom \nu_\Phi$ and the open set $\nu_\Phi(z)$ is 
$\theta_\Phi$-computable. Let $x=\delta_W(p)=\delta_\Phi(p)$, for each enumerated string $z$, 
$\mathfrak{M}$  
tries to determine if $x \in \nu_\Phi(z)$, if so, then $\mathfrak{M}$ prints the string 
$\iota(z^\prime)$ on 
the output tape, where  
\[ z^\prime = \iota(\langle \langle i_1, e_1 \rangle, e_1 \rangle)\cdots 
\iota(\langle \langle i_s, e_s \rangle, e_s \rangle)\iota(\langle \langle j_1, f_1 \rangle, f_1 
\rangle)\cdots \iota(\langle \langle j_r, f_r \rangle, f_r \rangle) \in \dom \nu_{\Phi_W} \] 
and $w=\iota(\langle i_1, e_1 \rangle)\cdots \iota(\langle i_s, e_s \rangle)$ and $y=\iota(\langle 
j_1, f_1 \rangle)\cdots \iota(\langle j_r, f_r \rangle)$. The machine $\mathfrak{M}$ must execute 
simultaneously this step for multiple enumerated strings $z$, and from time to time, $\mathfrak{M}$ 
must begin 
executing new tests. This completes the specification of $\mathfrak{M}$.

By part 1 of Lemma \ref{lemCorollary14ECT}, the decision problem ``$x\in \nu_\Phi(z)$'' is 
$(\delta_\Phi, 
\theta_\Phi)$-c.e., thus if $x\in \nu_\Phi(z)$, $\mathfrak{M}$ will finish executing this step for 
the string 
$z$. But if $x\notin \nu_\Phi(z)$, $\mathfrak{M}$ might not be able to finish this part of its 
program in 
finite time. This is the reason why $\mathfrak{M}$ must run multiple tests ``$x\in \nu_\Phi(z)$'' 
simultaneously, advancing each test a few steps at a time. Because $x \in W$, it cannot happen that 
all of the test executed by $\mathfrak{M}$ are unsuccessful. Therefore the output of the machine 
$\mathfrak{M}$ is 
non-trivial, it is an infinite string $p^\prime \in \sigmaomegastringset$ such that 
\[ (\forall z^\prime \in \sigmastringset) (z^\prime \ll p^\prime \Leftrightarrow z^\prime \in \dom 
\nu_{\Phi_W} \wedge \delta_{\Phi_W}(p^\prime) \in \nu_{\Phi_W}(z^\prime)), \]
and $x=\delta_W(p)=\delta_{\Phi_W}(p^\prime)$. Therefore $\delta_W \leq \delta_{\Phi_W}$.

We have shown that $\delta_W \equiv \delta_{\Phi_W}$, hence the inclusion $i\colon W 
\hookrightarrow M$ is a computable embedding of $\mathbf{T}_{\Phi_W}(W)$ onto $\mathbf{T}_\Phi(M)$.
By Definition \ref{defCompSubmanifold}, $W$ is an (open) computable submanifold of $M$.
\end{proof}

When $W$ is not $\theta_\Phi$-computable, we do not know whether $(W, \Phi_W)$ is a computable
submanifold of $M$, we only have that $\nu_W \leq \theta_{\Phi_W}$, i.e., everything that
is computable in $\mathbf{T}_{\Phi_W}(W)$ is computable in $\ctopcompsubspace{\mathbf{T}}{W}$.

In Section \ref{subsecCompFunctionsCompTOPSPACES}, we introduced computable embeddings of computable
topological spaces, the next lemma shows their relationship with computable manifolds. Informally,
it says that every computable manifold that is computably embedded in another computable manifold,
can be thought as a computable submanifold of the latter.

\begin{lemma}\label{lemmaImbeddingImageManifold}
 Let $(M^n, \Phi), (N^m,\Psi)$ be computable manifolds and $h\colon M \hookrightarrow N$ be a
computable embedding. Then the subset $M^\prime=\directimage{M}{h}$ has the structure of a
computable submanifold of $N$, such that $M \chomeom M^\prime$.
\end{lemma}

\begin{proof}
Let $\mathbf{T}=\mathbf{T}_{\Psi}(N)$. By Corollary \ref{corManifoldFuncTopSpace}, the computable
homeomorphism $h \colon M \to M^\prime$ induces a computable structure $\Phi_h$ on $M^\prime$ such
that $\mathbf{T}_{\Phi_h}(M^\prime)$ is equivalent to $\ctopcompsubspace{\mathbf{T}}{M^\prime}$, 
therefore the
inclusion of $M^\prime$ into $N$ is a computable embedding of $\mathbf{T}_{\Phi_h}(M^\prime)$ into
$\mathbf{T}$, so that $(M^\prime, \Phi_h)$ is a computable submanifold of $(N, \Psi)$. Since 
$\mathbf{T}_\Phi(M) \chomeom \mathbf{T}_{\Phi_h}(M^\prime)$, we are done.
\end{proof}

\begin{example}\label{exampleTorusDefinitionsCHOmeom}
In example \ref{exampleNTORUS}, if we take $n=2$, the 2-dimensional Torus $T^2=\sphere{1} \times
\sphere{1}$ is a computable submanifold of $\euclidean{4}$. But it is known from standard topology
that the map $h\colon T^2 \hookrightarrow \euclidean{3}$ given by  
\[ h(x_u,y_u,x_v,y_v)=((2 + y_v)y_u, (2 + y_v)x_u, x_v), \]
is a homeomorphism of $T^2$ onto the set $T^\prime=\directimage{T^2}{h}$ (that is, it is an
embedding). Clearly, $h$ is a computable embedding of $\mathbf{T}_\Phi(T^2)$ into $\mathbf{R}^3$.
By Lemma \ref{lemmaImbeddingImageManifold}, $h$ induces a computable structure on $T^\prime$, such
that $\ctopcompsubspace{\mathbf{R}^3}{T^\prime}$ is equivalent to $\mathbf{T}_{\Phi_h}(T^\prime)$, 
so that, $(T^\prime, \Phi_h)$ becomes a computable submanifold of $\mathbf{R}^3$.
\end{example}

\begin{lemma}\label{lemmaCompSubManifoldSCT2ManifoldIsSCT2}
Every computable submanifold of a computably Hausdorff computable manifold is a computably
Hausdorff manifold.
\end{lemma}

\begin{proof}
 This is an immediate consequence of Theorem \ref{thmPropertiesSCT2} and Definition
\ref{defCompSubmanifold}.
\end{proof}

\section{Computable submanifolds of computable euclidean spaces}\label{secImbeddings}

In the general theory of manifolds, it is shown that abstract Hausdorff topological $n$-manifolds 
are no more general that $n$-dimensional submanifolds of euclidean spaces, thus for all practical 
purposes, one can work with submanifolds of $\euclidean{q}$ and this yields some good advantages. 
The main step in proving that general abstract manifolds can be reduced to submanifolds of 
$\euclidean{q}$, is to show that for any $n$-manifold $M$, there is an embedding of $M$ into some 
euclidean space.

In this section, we will prove a computable version of the embedding theorem for topological 
manifolds. We will show that any abstract compact computable $n$-manifold 
that is computably Hausdorff, can be embedded in some computable space $\mathbf{R}^q$, for $q$ 
sufficiently large. 

\subsection{Computable embeddings of manifolds in computable euclidean spaces}

There are many facts \cite{kirby1977foundational} in the theory of manifolds (topological,
differentiable and/or PL) which can be shown to be true using the well known result that 
every Hausdorff $n$-manifold $M$ (of any kind) embeds in some high dimensional euclidean space
$\euclidean{q}$, where $q$ depends on $n$. Many versions of this embedding theorem exist, 
the difference between them is 
the dimension of the space $\euclidean{q}$. It was proven by Whitney
\cite{WhitneyDIFF36} that if $M$ is smooth, then it embeds in $\euclidean{2n}$ and this is the best
possible result. The same is true for the piecewise linear case using similar constructions to those
used in the smooth case. If $M$ has no additional structure, it can be embedded in
$\euclidean{2n+1}$ and again, this is the lowest possible dimension for the euclidean space. This
last result can be proven by means of dimension theory \cite{munkres2000topology,hurewicz1948dimension,engelking1995theory}.

Our aim now is to prove a computable version of an embedding theorem for compact computable
manifolds which are computably Hausdorff. We will show that every such manifold can be embedded in
computable euclidean space $\mathbf{R}^{q}$ (where $q$ must be large enough, we will not try to
optimize $q$) with a computable embedding. 
The question remains open if a compact computably Hausdorff computable manifold can be computably
embedded in a lower dimensional euclidean space. Notice that by Lemma
\ref{lemmaCompSubManifoldSCT2ManifoldIsSCT2}, any computable manifold embedded in computable
euclidean space must have the property of being computably Hausdorff.

\subsection{An embedding theorem}

From now on, all computable manifolds are assumed to be computably Hausdorff. We now show that every
compact computably Hausdorff computable manifold has a computable embedding in some euclidean space of
sufficiently high dimension. We develop a computable version of the proof of the classical 
embedding theorem which can be found in \cite{1974TOPPROPERTIES}.

Let $(M, \Phi)$ be a compact computable manifold of dimension $n$ and suppose that $\Phi$ is the
atlas given in part (a) of Theorem \ref{thmCompManifoldEuclideanAtlas}. 
Let $(\varphi, U) \in \Phi$, where $\directimage{U}{\varphi}=\mu^n(z)$ for some $z\in 
\domp{\mu^n}$. Let the map $h_z\colon \mu^n(z) \to \euclidean{n}$ be the computable homeomorphism 
of Lemma \ref{lemmaTechResultCompHomeoBallsRn} and let 
$s\colon \sphere{n} - \{ P \} \to \euclidean{n}$ ($P=(0, \ldots, 0, 1)$) be the computable 
stereographic projection given in Example \ref{exampleStereographicProjection}. Define a function 
$g\colon M\to \sphere{n}$ by
\begin{equation*}
g(x)=\begin{cases}
   s^{-1} \circ h_z \circ \varphi(x) & \text{if } x\in U, \\
  P & \text{if } x\in M - U.
  \end{cases}
\end{equation*}

\begin{lemma}\label{lemmaOneCoordinateImbedding}
 The function $g$ is computable and hence it is continuous.
\end{lemma}

\begin{proof}
We prove that $g$ is a computable map from $\mathbf{T}_\Phi(M)$ to $\mathbf{S}^n$ by showing that
the map $B\mapsto \inverseimage{B}{g}$ is $(\nu_{\sphere{n}}, \theta_\Phi)$-com\-pu\-ta\-ble (See
Theorem \ref{thmequivCompFunctions}). There is a Type-2 machine $\mathfrak{M}$ that on 
input $V = B
\cap \sphere{n}$ ($B = B(r, \epsilon)$ is the open ball in $\euclidean{n+1}$), does the following:
\begin{enumerate}
 \item If $V$ contains the point $P$, then execute the following steps:
\begin{itemize}
 \item [1.1] Calculate the compact set $K=\sphere{n} - V$;
 \item [1.2] compute the compact set $K^\prime=\inverseimage{K}{g}$;
 \item [1.3] output the set $V^\prime=M-K^\prime$.
\end{itemize}
\item If $V$ does not contains $P$, output the open set $\inverseimage{V}{g}$.
\end{enumerate}
First, we show that each step can be executed by $\mathfrak{M}$. The test $P \in V$ in steps 1 and 
2 
can be
done by $\mathfrak{M}$ in finite time, because $P\in \rationalsspace{n}$ and $V= B \cap \sphere{n}$ 
is
specified by rational numbers. To construct the set $K$ in step 1.1, $\mathfrak{M}$ needs to 
compute 
an element
$p\in \sigmaomegastringset$ such that $K \in \kappa_{\sphere{n}}(p)$ and this can be done by 
$\mathfrak{M}$
because the map $V \mapsto \sphere{n} - V$ is $(\nu_{\sphere{n}}, \kappa_{\sphere{n}})$-computable.
Step 1.2 can be accomplished because in the subset $U\subset \dom g = M$, $g^{-1}=\varphi^{-1} 
\circ 
f_z^{-1} \circ s$ exists and it is a computable function, thus the function $K \mapsto 
\inverseimage{K}{g}$ is $(\kappa_{\sphere{n}}, \kappa_{\Phi})$-computable. To execute step 1.3, 
$\mathfrak{M}$ 
can use the computable function $f\colon \subseteq \sigmaomegastringset \to \sigmaomegastringset$ 
of 
the reduction $\kappa_{\Phi} \leq \psi^{-}_\Phi$. This reduction exists because 
$\mathbf{T}_\Phi(M)$ 
is computably Hausdorff (part 4 of Theorem \ref{thmPropertiesSCT2}). The construction of the set 
$\inverseimage{V}{g}$ in step 2 can be executed by $\mathfrak{M}$ because as $P\notin V$, $V\subset 
\sphere{n} 
- \{ P \}$, thus in $V$, $g$ is a computable function. This proves that each step of $\mathfrak{M}$ 
can be done in finite time, so that $\mathfrak{M}$ is a valid Type-2 machine.

We now prove the correctness of this pseudocode. Assume that $P\in V$, then $\mathfrak{M}$ goes
on to execute step 1.1 and compute the set $K=\sphere{n} - V$. Notice that as $P\in V$, $K\subset
\directimage{U}{g}$, so that in $K$, $g=s^{-1} \circ f_z \circ \varphi$ and $\mathfrak{M}$ can use 
$s,f_z$ 
and $\varphi^{-1}$ to compute the compact set $K^\prime=\inverseimage{K}{g}\subset U$ in step 1.2. 
Finally, in step 1.3, $\mathfrak{M}$ computes the open set $V^\prime=M-K^\prime$, and this set is 
such that
$V^\prime=\inverseimage{V}{g}$. Suppose that $P\notin V$. Then $V\subset \sphere{n} - \{ P \}$ and 
$\mathfrak{M}$ uses $g$ to compute the open set $\inverseimage{V}{g}\subset U$.

This proves that the machine $\mathfrak{M}$ computes a function which realizes the map $V 
\mapsto
\inverseimage{V}{g}$ with respect to $\nu_{\sphere{n}}$ and $\theta_\Phi$. Therefore, the map $V 
\mapsto \inverseimage{V}{g}$ is $(\nu_{\sphere{n}}, \theta_\Phi)$-computable, thus by Theorem 
\ref{thmequivCompFunctions}, $g$ is a computable function from $M$ to $\sphere{n}$. 
By Proposition \ref{propCompFunctContFunc}, $g$ is continuous. 
\end{proof}

We are ready to prove the main result of this section.

\begin{theorem}\label{thmCompactT2CompManEmbedsRq}
 For any compact computable manifold $M^n$, there exists a computable embedding of $M$ into
$\mathbf{R}^{q}$ for $q$ sufficiently large.
\end{theorem}

\begin{proof}
 Let $M$ be a compact computable manifold of dimension $n$ and assume that $\sphere{n}$ has the
computable subspace topology induced by $\mathbf{R}^{n+1}$. By compactness of $M$ and part (a) of
Theorem \ref{thmCompManifoldEuclideanAtlas}, we can find a finite atlas $\{ (\varphi_1, U_1),
\ldots, (\varphi_l, U_l) \}$ such that $\directimage{U_i}{\varphi_i}=\mu^n(z_i)$ for all 
$i=1,\ldots,l$. Using Lemma \ref{lemmaOneCoordinateImbedding}, we 
construct computable functions $g_i\colon M\to \sphere{n}$. Now let
\[X= \underbrace{\sphere{n} \times \cdots \times \sphere{n}}_{l \text{ times}} \quad
\text{and} \quad \overline{\mathbf{X}}=(X, \overline{\tau}, \overline{\beta},\overline{\nu})\]
and define $G\colon M \to X$ of $\mathbf{T}_\Phi(M)$ into $\overline{\mathbf{X}}$ by $G=(g_1,
\ldots, g_l)$. 
We now prove that $G$ is injective. Take $x,y\in M$ with $x\neq y$. If there is a chart $U_i$ such 
that $x,y\in U_i$ then $g_i(x)\neq g_i(y)$ because $g_i$ is inyective in $U_i$. If $x$ and $y$ are 
not in the same chart, then $x\in U_i$ for some $i$ and $y\notin U_i$, hence $g_i(x)\neq P$ and 
$g_i(y)=P$. Since $M$ is compact and $X$ is Hausdorff, then $G$ is a closed map, therefore $G$ is a 
topological embedding. 
Because each $g_i$ is computable, $G$ is computable (use part 2. of Lemma
\ref{lemmaPropertiesEffectiveTOPPRODUCT}). To see that the inverse function $G^{-1}$ is
computable, let $\mathfrak{M}^\prime$ be a Type-2 machine that on input $x\in \range h$, executes 
the following
steps:
\begin{enumerate}
 \item Compute each component $x_i \in \sphere{n}$ ($i=1,\ldots,l$) of $x$.
 \item Find $j$ such that $x_j\in \sphere{n} - P$.
\item Output $y=g_j^{-1}(x_j)$.
\end{enumerate}
Step 1 can be computed because by Lemma \ref{lemmaPropertiesEffectiveTOPPRODUCT}, the map
$x\mapsto x_i$ is $(\overline{\delta}, \delta_{\sphere{n}})$-computable for each $i$; step 2 is
finished in finite time because each $x_i\in \sphere{n}\subset \euclidean{n+1}$ is
$\delta_{\sphere{n}}$-computable, the set $\sphere{n} - P$ is $\theta_{\sphere{n}}$-computable open
in $\sphere{n}$ and by part 1 of Lemma \ref{lemCorollary14ECT}, the decision problem ``$x_i \in 
\sphere{n} -
P$'' is $(\delta_{\sphere{n}}, \theta_{\sphere{n}})$-c.e. ($x_j\neq P$ for at least one index $j$);
step 3 is easily calculated because $g_j^{-1}$ exists and it is $(\delta_{\sphere{n}}, 
\delta_\Phi)$-computable in $\sphere{n} - P$. Hence, the function computed by
$\mathfrak{M}^\prime$ realizes $g^{-1}$, that is, it is a computable function.

Therefore $G$ is a computable embedding of $\mathbf{T}_\Phi(M)$ into $\overline{\mathbf{X}}$, which is a
computable subspace of $\overline{\mathbf{X}}^\prime=(X^\prime, \overline{\tau}^\prime,
\overline{\beta}^\prime, \overline{\nu}^\prime)$ ($X^\prime=\prod_i^l \euclidean{n+1}$, each
factor is to be understood as the computable space $\mathbf{R}^{n+1}$), which is
equivalent to the computable euclidean space $\mathbf{R}^{l(n+1)}$, so that by combining $G$ with
the computable inclusion of $\overline{\mathbf{X}}$ into $\overline{\mathbf{X}}^\prime$ and the
equivalence of $\overline{\mathbf{X}}^\prime$ with $\mathbf{R}^{l(n+1)}$, we obtain the desired
computable embedding of $M$ into $\euclidean{l(n+1)}$.
\end{proof}

With the computable embedding constructed in Theorem 
\ref{thmCompactT2CompManEmbedsRq}, we can deduce that abstract compact computably Hausdorff 
computable manifolds and compact computable submanifolds of euclidean spaces are essentially the 
same.

\section{Final remarks}\label{secFinalRemarks}

In this paper, we have proposed a starting point to build a computable theory 
for topological manifolds, viewing computability as a structure that we impose to topological 
manifolds. We have provided the basic results needed to give computable versions 
of the standard definitions and theorems. We also studied computable functions 
between computable manifolds and defined computable submanifolds. Finally, we proved an embedding 
theorem for compact computably Haussdorf computable manifolds, which is a computable version of the 
result which states that every compact manifold embeds in some euclidean space.

%
%

The computable theory of topological manifolds that we present in this paper could 
also be used as a new part of the standard theory of manifolds. It is known that the most 
important structures (topological, smooth and piecewise linear) coincide in dimension at most three 
and they are different in dimensions at least four\footnote{For the case of dimension four, 
smooth and piecewise linear structures coincide, but there exists a topological 4-manifold which 
has no smooth or picewise linear structure.} \cite{kirby1977foundational,KervaireManifoldNODIFF,RudyakPLStructures}. Also, it was recently proved in \cite{MR3402697} that in dimensions 
at least 5, there are topological manifolds that cannot be triangulated as simplicial complexes.  
We have 
introduced computability as a structure that we impose to topological 
manifolds. Which is the relationship of computable structures with the classical structures and 
simplicial triangulations? An important open question (among many others) in this direction, 
concerning the computable theory is the following: Does there exists a second 
countable topological manifold $E$ such that $E$ does not admit a computable structure ?

\section*{Acknowledgement}

We would like to thank 
Professor Francisco
Gonz\'alez Acu\~na for very valuable discussions. Also, special thanks to all the 
members of the \href{http://cca-net.de/list.html}{Computability and Complexity in Analysis}
 mailing list for their help, comments, 
 corrections and tips about computable analysis and topology. In 
particular, thanks to Vasco Brattka, Klaus Weihrauch, Mathieu Hoyrup and Dimiter Skordev.

\bibliographystyle{unsrt}       
\bibliography{bib/biblio}           

\begin{thebibliography}{10}

\bibitem{Turing1936}
Alan Turing.
\newblock On computable numbers, with an application to the
  entscheidungsproblem.
\newblock {\em Proceedings of the London Mathematical Society}, 42:230--265,
  1936.

\bibitem{Lacombe}
Daniel Lacombe.
\newblock Extension de la notion de fonction r\'ecursive aux fonctions d'une ou
  plusieurs variables r\'eelles. {II}, {III}.
\newblock {\em C. R. Acad. Sci. Paris}, 241:13--14, 151--153, 1955.

\bibitem{Grzegorczyk}
A.~Grzegorczyk.
\newblock On the definitions of computable real continuous functions.
\newblock {\em Fund. Math.}, 44:61--71, 1957.

\bibitem{pourElRichards}
Marian~B. Pour-El and J.~Ian Richards.
\newblock {\em Computability in analysis and physics}.
\newblock Perspectives in Mathematical Logic. Springer-Verlag, Berlin, 1989.

\bibitem{kerikocomplexityrealfunctions}
Ker-I Ko.
\newblock {\em Complexity theory of real functions}.
\newblock Progress in Theoretical Computer Science. Birkh\"auser Boston Inc.,
  Boston, MA, 1991.

\bibitem{weihrauch2000computable}
K.~Weihrauch.
\newblock {\em Computable analysis: an introduction}.
\newblock Texts in theoretical computer science. Springer, 2000.

\bibitem{borodin1975computational}
A.~Borodin and I.~Munro.
\newblock {\em The computational complexity of algebraic and numeric problems}.
\newblock Theory of computation series. American Elsevier Pub. Co., 1975.

\bibitem{burgisser1997algebraic}
P.~B{\"u}rgisser, M.~Clausen, and M.A. Shokrollahi.
\newblock {\em Algebraic complexity theory}.
\newblock Grundlehren der mathematischen Wissenschaften. Springer, 1997.

\bibitem{BSS89}
Lenore Blum, Mike Shub, and Steve Smale.
\newblock On a theory of computation and complexity over the real numbers:
  {NP}-completeness, recursive functions and universal machines.
\newblock {\em Bull. Amer. Math. Soc. (N.S.)}, 21(1):1--46, 1989.

\bibitem{BCSS98}
Lenore Blum, Felipe Cucker, Michael Shub, and Steve Smale.
\newblock {\em Complexity and real computation}.
\newblock Springer-Verlag, New York, 1998.
\newblock With a foreword by Richard M. Karp.

\bibitem{theoryrepresentations}
Christoph Kreitz and Klaus Weihrauch.
\newblock Theory of representations.
\newblock {\em Theoret. Comput. Sci.}, 38(1):35--53, 1985.

\bibitem{Weihrauch:jucs_14_6:the_computable_multi_functions}
Klaus Weihrauch.
\newblock The computable multi-functions on multi-represented sets are closed
  under programming.
\newblock {\em Journal of Universal Computer Science}, 14(6):801--844, mar
  2008.

\bibitem{GTMs}
N.~{Tavana} and K.~{Weihrauch}.
\newblock {Turing machines on represented sets, a model of computation for
  Analysis}.
\newblock {\em ArXiv e-prints}, May 2011.

\bibitem{KlausCompMetricStructs}
Klaus Weihrauch.
\newblock Computability on computable metric spaces.
\newblock {\em Theoretical Computer Science}, 113(2):191 -- 210, 1993.

\bibitem{Bra03ComputabilityModels}
Vasco Brattka.
\newblock {Computability over Topological Structures}.
\newblock In S.~Barry Cooper and Sergey~S. Goncharov, editors, {\em
  Computability and Models}, pages 93--136. Kluwer Academic Publishers, New
  York, 2003.

\bibitem{CCAtutorial}
Vasco Brattka, Peter Hertling, and Klaus Weihrauch.
\newblock A tutorial on computable analysis.
\newblock In S.~Barry Cooper, Benedikt Löwe, and Andrea Sorbi, editors, {\em
  New Computational Paradigms}, pages 425--491. Springer New York, 2008.

\bibitem{schroder2003admissible}
M.~Schr{\"o}der.
\newblock {\em Admissible representations for continuous computations}.
\newblock Informatik-Berichte. Fernuniv., Fachbereich Informatik, 2003.

\bibitem{GW05ComputableDini}
Tanja Grubba and Klaus Weihrauch.
\newblock A computable version of dini's theorem for topological spaces.
\newblock In {\em Proceedings of the 20th international conference on Computer
  and Information Sciences}, ISCIS'05, pages 927--936, Berlin, Heidelberg,
  2005. Springer-Verlag.

\bibitem{GSW07Computablemetrization}
Tanja Grubba, Matthias Schröder, and Klaus Weihrauch.
\newblock Computable metrization.
\newblock {\em Mathematical Logic Quarterly}, 53(4-5):381--395, 2007.

\bibitem{KWECTOP}
K.~Weihrauch and T.~Grubba.
\newblock Elementary computable topology.
\newblock {\em Journal of Universal Computer Science}, 15(6):1381--1422, 2009.

\bibitem{GrubbaEffectiveStoneWeierstrass}
Tanja Grubba, Klaus Weihrauch, and Yatao Xu.
\newblock Effectivity on continuous functions in topological spaces.
\newblock {\em Electronic Notes in Theoretical Computer Science}, 202(0):237 --
  254, 2008.
\newblock Proceedings of the Fourth International Conference on Computability
  and Complexity in Analysis (CCA 2007).

\bibitem{Weihrauch:jucs_16_18:computable_separation_in_topology}
Klaus Weihrauch.
\newblock Computable separation in topology, from t0 to t2.
\newblock {\em Journal of Universal Computer Science}, 16(18):2733--2753, sep
  2010.

\bibitem{DBLP:journals/corr/RettingerW13}
Robert Rettinger and Klaus Weihrauch.
\newblock Products of effective topological spaces and a uniformly computable
  tychonoff theorem.
\newblock {\em Logical Methods in Computer Science}, 9(4), 2013.

\bibitem{kirby1977foundational}
R.C. Kirby and L.~Siebenmann.
\newblock {\em Foundational essays on topological manifolds, smoothings, and
  triangulations}.
\newblock Annals of mathematics studies. Princeton University Press, 1977.

\bibitem{MR3402697}
Ciprian Manolescu.
\newblock Pin(2)-equivariant {S}eiberg-{W}itten {F}loer homology and the
  triangulation conjecture.
\newblock {\em J. Amer. Math. Soc.}, 29(1):147--176, 2016.

\bibitem{simplicialtriangulationstop}
David~E. Galewski and Ronald~J. Stern.
\newblock Classification of simplicial triangulations of topological manifolds.
\newblock {\em The Annals of Mathematics}, 111(1):pp. 1--34, 1980.

\bibitem{WhitneyDIFF36}
Hassler Whitney.
\newblock Differentiable manifolds.
\newblock {\em Ann. of Math.}, 37(3):pp. 645--680, 1936.

\bibitem{brickellclarkDIFF}
F.~Brickell and R.~S. Clark.
\newblock {\em Differentiable manifolds: An introduction}.
\newblock Van Nostrand Reinhold Company London, 1970.

\bibitem{hormander1990introduction}
L.~H{\"o}rmander.
\newblock {\em An introduction to complex analysis in several variables}.
\newblock North-Holland mathematical library. North-Holland, 1990.

\bibitem{hudson1969piecewise}
J.F.P. Hudson.
\newblock {\em Piecewise linear topology}.
\newblock Mathematics lecture note series. W. A. Benjamin, 1969.

\bibitem{rourke1982introduction}
C.P. Rourke and B.J. Sanderson.
\newblock {\em Introduction to piecewise-linear topology}.
\newblock Ergebnisse der Mathematik und ihrer Grenzgebiete. Springer-Verlag,
  1982.

\bibitem{KervaireManifoldNODIFF}
Michel~A. Kervaire.
\newblock A manifold which does not admit any differentiable structure.
\newblock {\em Comment. Math. Helv.}, 34:257--270, 1960.

\bibitem{poincareconjecturegt4}
Stephen Smale.
\newblock Generalized {P}oincar\'e's conjecture in dimensions greater than
  four.
\newblock {\em Ann. of Math. (2)}, 74:391--406, 1961.

\bibitem{Freedman4manifolds}
Michael~Hartley Freedman.
\newblock The topology of four-dimensional manifolds.
\newblock {\em J. Differential Geom.}, 17(3):357--453, 1982.

\bibitem{1974TOPPROPERTIES}
D.~B. Gauld.
\newblock Topological properties of manifolds.
\newblock {\em Amer. Math. Monthly}, 81(6):pp. 633--636, 1974.

\bibitem{munkres2000topology}
J.R. Munkres.
\newblock {\em Topology}.
\newblock Prentice Hall, Incorporated, 2000.

\bibitem{realcomputablemanifolds}
Wesley Calvert and Russell Miller.
\newblock Real computable manifolds and homotopy groups.
\newblock In {\em Proceedings of the 8th International Conference on
  Unconventional Computation}, UC '09, pages 98--109, Berlin, Heidelberg, 2009.
  Springer-Verlag.

\bibitem{DBLP:journals/corr/Iljazovic13}
Zvonko Iljazovic.
\newblock Compact manifolds with computable boundaries.
\newblock {\em Logical Methods in Computer Science}, 9(4), 2013.

\bibitem{DBLP:journals/corr/BurnikI14}
Konrad Burnik and Zvonko Iljazovic.
\newblock Computability of 1-manifolds.
\newblock {\em Logical Methods in Computer Science}, 10(2), 2014.

\bibitem{armstrong}
Mark~A. Armstrong.
\newblock {\em Basic Topology}.
\newblock Springer-Verlag, 1983.

\bibitem{engelking1989general}
R.~Engelking.
\newblock {\em General topology}.
\newblock volume 6 of Sigma series in pure mathematics. Heldermann, Berlin,
  1989.

\bibitem{hurewicz1948dimension}
W.~Hurewicz and H.~Wallman.
\newblock {\em Dimension theory}.
\newblock Princeton mathematical series. Princeton University Press, 1948.

\bibitem{engelking1995theory}
R.~Engelking.
\newblock {\em Theory of dimensions, finite and infinite}.
\newblock Sigma series in pure mathematics. Heldermann, 1995.

\bibitem{kozen1997automata}
D.~Kozen.
\newblock {\em Automata and computability}.
\newblock Undergraduate texts in computer science. Springer, 1997.

\bibitem{cooper2004computability}
S.B. Cooper.
\newblock {\em Computability theory}.
\newblock Chapman \& Hall/CRC mathematics. Chapman \& Hall/CRC, 2004.

\bibitem{Tarski51}
A.~Tarski.
\newblock {\em A decision method for elementary algebra and geometry}.
\newblock Rand report. Rand Corporation, 1948.

\bibitem{MatthiasAdmisibility}
M.~Schr{\"o}der.
\newblock Extended admissibility.
\newblock {\em Theoretical Computer Science}, 284(2):519 -- 538, 2002.

\bibitem{zhoucompfuncopenclosed}
Qing Zhou.
\newblock Computable real-valued functions on recursive open and closed subsets
  of euclidean space.
\newblock {\em Mathematical Logic Quarterly}, 42(1):379--409, 1996.

\bibitem{BrattkaW99}
Vasco Brattka and Klaus Weihrauch.
\newblock Computability on subsets of euclidean space i: Closed and compact
  subsets.
\newblock {\em Theor. Comput. Sci.}, 219(1-2):65--93, 1999.

\bibitem{BrattkaP03}
Vasco Brattka and Gero Presser.
\newblock Computability on subsets of metric spaces.
\newblock {\em Theor. Comput. Sci.}, 305(1-3):43--76, 2003.

\bibitem{RudyakPLStructures}
Y.~B. {Rudyak}.
\newblock {Piecewise linear structures on topological manifolds}.
\newblock {\em ArXiv Mathematics e-prints}, May 2001.

\end{thebibliography}

\end{document}